\documentclass[12pt,a4paper,twoside,openright]{report}
\usepackage[italian,english]{babel}
\pdfoutput=1
\usepackage{fancyhdr} 
\usepackage{amsmath}
\usepackage{amsthm}
\usepackage{amssymb}
\usepackage{enumerate}
\usepackage{mathrsfs}
\usepackage{braket}

\usepackage[alphabetic]{amsrefs}

\usepackage{fancyhdr}
\author{Matteo Capoferri}
\title{Master Thesis}

\usepackage[colorlinks=true,linkbordercolor={1 0 0},citebordercolor={0 1 0},linkcolor=red,citecolor=cyan]{hyperref}  

\newtheorem{thm}{Theorem}[chapter]
\newtheorem{lem}[thm]{Lemma}
\newtheorem{prop}[thm]{Proposition}
\newtheorem{cor}[thm]{Corollary}
\theoremstyle{definition}
\newtheorem{defn}[thm]{Definition}
\newtheorem{rem}[thm]{Remark}
\theoremstyle{definition}
\newtheorem{ex}[thm]{Example}

\usepackage{color}

\usepackage{titlesec}
\titleformat{\chapter}[display]
{\normalfont\LARGE\bfseries}{\chaptertitlename\ \thechapter}{20pt}{\Huge}

\newcommand{\pa}[2]{\langle #1 \, , \, #2 \rangle}

\newcommand{\curv}[0]{\text{curv}\,}
\newcommand{\cha}[0]{\text{char}\,}
\newcommand{\Z}[0]{\mathbb{Z}}
\newcommand{\R}{\mathbb{R}}
\newcommand{\T}{\mathbb{T}}
\newcommand{\N}{\mathbb{N}}
\renewcommand{\S}[0]{\mathbb{S}}
\newcommand{\Hom}{\mathrm{Hom}}
\renewcommand{\arg}{\,\cdot\,}
\newcommand{\tor}{\mathrm{tor}}
\newcommand{\free}{\mathrm{free}}
\renewcommand{\mod}{\text{mod}\,\,}
\renewcommand{\tilde}[1]{\widetilde{#1}}
\renewcommand{\hat}[1]{\widehat{#1}}
\newcommand{\C}{\mathfrak{C}}
\newcommand{\D}{\mathfrak{D}}
\newcommand{\norm}[1]{\left\| #1 \right\|}
\renewcommand{\i}[1]{{#1}_\iota}
\newcommand{\ii}[1]{{#1}_{\iota, I}}
\newcommand{\I}[1]{{#1}_I}
\renewcommand{\braket}[2]{\left\langle #1 \, \vline \, #2 \right\rangle}

\usepackage[all]{xy} 

\allowdisplaybreaks 

\newenvironment{enumerates}{\vspace{-0.25cm} \begin{enumerate}}{ \end{enumerate}}

\usepackage[final]{pdfpages}

\begin{document}
\includepdf[pages=-]{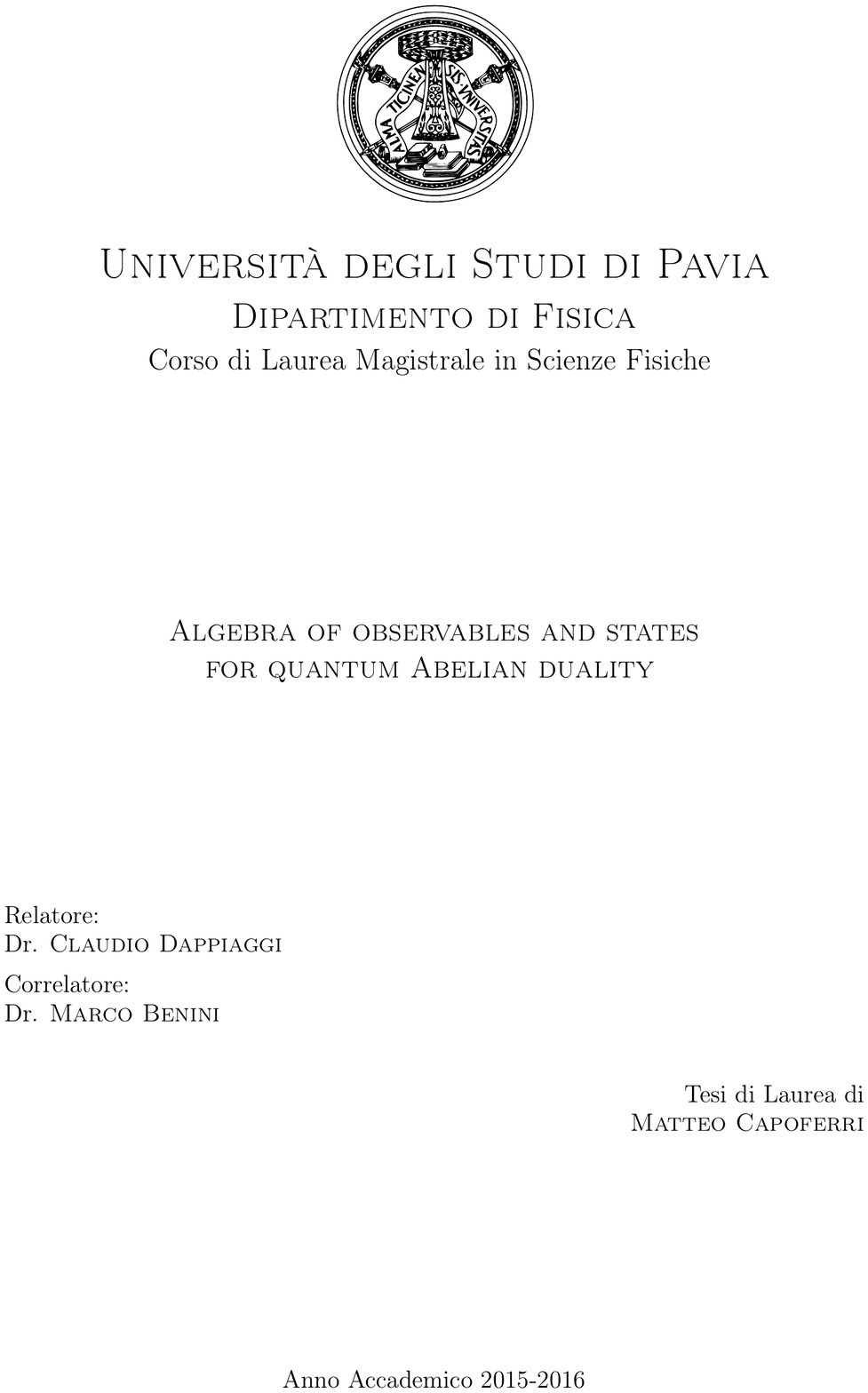}
{\thispagestyle{empty}\cleardoublepage }
\tableofcontents \markboth{\contentsname}{\contentsname}

\chapter*{Introduction} \addcontentsline{toc}{chapter}{Introduction} \markboth{Introduction}{Introduction}

The advent of Quantum Field Theory - henceforth QFT - brought about an actual revolution in physics. 
Its inspiring conceptual idea, namely the description of the elementary constituents of matter as \emph{fields}, allowed to give solution to puzzling, long-standing issues and to make predictions verified with unprecedented accuracy in the experimental arena. QFT was firstly formulated in Minkowski spacetime, an ambient space with a high degree of symmetry. The Poincar\'e group of isometries, together with positivity of energy, permits, in particular, to select a distinguished state, the \emph{vacuum state}, which enjoys the remarkable property of being unique.\\

In 1964 Haag and Kastler proposed a novel framework in which to formulate  QFT on Minkowski spacetime, the \emph{algebraic approach} \cite{HK64}, paving the way for the birth of what is nowadays known as Algebraic Quantum Field Theory (AQFT). This approach prescinds from the choice of a preferred Hilbert space and it encodes all the physical information of the dynamics in an abstract $C^\ast$-algebra, the algebra of observables, endowed with commutation or anticommutation relations, implementing in a natural way causality, locality and Lorentz covariance. Haag and Kastler theory offers a powerful and mathematically rigorous framework, appropriate for being applied to curved spacetimes. In fact, as soon as the background metric is slightly perturbed, thus giving rise to a non-vanishing curvature, the mathematics of standard QFT on Minkowski spacetime breaks down, and substantial modifications of the theory are in order.\\

Over the years, the Haag-Kastler theory, refined and extended in its range of application, led to a precise formulation of QFT on curved spacetimes \cite{WAL94,Dim80}, triggering the discovery of new physical effects never predicted before, preserving, at the same time, its original flavour. AQFT provides promising tools to tackle the description of quantum fields in presence of an external gravitational field. In the framework of AQFT there is not, so far, a fully fledged model of quantum gravity, since, for instance, the background spacetime is fixed by hand at the beginning. Nonetheless, by electing an abstract $C^\ast$-algebra as its key ingredient instead of a Hilbert space, it accomplishes the goal of encompassing all the cases where the absence of a timelike Killing vector field prevents to put forward the standard procedure to identify a global vacuum state. The Hilbert space can change in the Minkowski case as well; however, the advantage of AQFT is to prescind from a case by case analysis as far as the structural aspects of the observables of the theory are concerned.\\

The assignment of \emph{states}, videlicet positive and normalized linear functionals on the algebra of observables, allows to recover the probabilistic description of quantum theory, the expectation value of an observable being the real number returned by the state. Via the GNS construction, each state induces a cyclic representation of the algebra on a Hilbert space, unique up to unitary isomorphisms; different states may lead to inequivalent representations. This fact motivates the importance of the approach suggested by AQFT, that consists in describing a quantum field theory by an abstract $C^\ast$-algebra encoding the usual commutation relation, however without choosing one of its (in general inequivalent) representations right from the biginning. There exists a distinguished class of states, characterized by the short-distance behaviour of their two-point functions, commonly accepted as good physical states: the \emph{Hadamard states} \cite{DB60,KW91}. These states generalize the positive energy condition of the Minkowski vacuum, possessing the same ultraviolet behaviour, simultaneously guaranteeing that the quantum fluctuations of all observables are bounded. In 1996 Radzikowski \cite{Rad96} showed that the Hadamard condition can be formulated in terms of a constraint on the wavefront set of the two-point function of the state, thus bringing the powerful tecniques of microlocal analysis \cite{FIO1,FIO2} eminently inside QFT. This made it possible, for example, the development of a rigorous perturbation theory on any globally hyperbolic spacetime. \\

A further theoretical accomplishment that marks a milestone in the field was made in 2003 by Brunetti, Fredenhagen and Verch \cite{BFV03}, who, in a category-theoretic setting, proposed a model-independent, functorial approach to QFT. They suggested that each quantum field theory should be realized by a functor from the category of globally hyperbolic spacetimes to the category of $\ast$-algebras, called \emph{locally covariant quantum field theory}. Using this language, one can discribe the quantization of a given field theory on all spacetimes at once in a coherent way. Their approach possesses, furthermore, the advantage of implementing automatically the covariance requirements imposed by general relativity.\\

Even though the existence of Hadamard states is well-established for a wide class of spacetimes, a direct exhibition of a concrete Hadamard state has often proved to be a challenging task. The aim of our thesis is the construction of states on globally hyperbolic spacetimes with compact Cauchy surfaces for a quantum field theory described in terms of differential cohomology.\\

In 1985 Cheeger and Simons developed the theory of differential characters \cite{CS85}, a refinement of singular cohomology by means of differential forms. Later developments, e.g. \cite{SS08,BB14}, showed that the Cheeger-Simons differential character provide a model for the abstract theory of differential cohomology, described by a contravariant functor from the category of smooth manifolds to the category of Abelian groups. In recent years, differential cohomology has drawn the attention of theoretical physics. Classical Maxwell theory can be regarded as the theory of a connection on a given spacetime. More precisely, denoting by $L\to M$ a line bundle over a given spacetime $M$, we can call Maxwell field a connection on $L$ whose curvature satisfies Maxwell's equation. The Faraday two-form arises as the curvature of the connection. Since we are dealing wih a gauge theory, we have to introduce an equivalence relation between Maxwell fields, implemented by the gauge group of the theory $\mathcal{G}$. Then, the relevant configurations will be equivalence classes of Maxwell fields up to gauge transformations. If $\mathcal{A}(L)$ denotes the space of connections on $L$ satisfying Maxwell's equation, the gauge equivalence classes of Maxwell fields will be given by $\mathcal{A}(L)/\mathcal{G}$. So far we considered a fixed line bundle, however on the same manifold there can be several inequivalent ones, labelled by the second cohomology group with integer coefficients $H^2(M;\mathbb{Z})$. To fully describe Maxwell theory on a given spacetime, it is important to take into account also gauge classes of Maxwell fields associated to inequivalent line bundles. This brief outline shows that, if we aim at generalizing Maxwell theory to an Abelian gauge theory in arbitrary spacetime dimension, the differential and the topological information are equally relevant and substantial. This is why differential cohomology appears to be the appropriate environment in which to set such a theory. Furthermore, this framework offers the chance to implement the quantization of the electric and magnetic charges in a very natural way.\\

One of the first successful applications of differential cohomology to QFT can be found in \cite{FMSa,FMSb}, where the authors perform their analysis, via a Hamiltonian approach, in the case of ultrastatic spacetimes. Becker, Schenkel and Szabo in \cite{BSS14} put forward the investigation, showing that, in the framework of covariant quantum field theory, it is possible to construct a QFT for differential cohomology on a generic globally hyperbolic Lorentzian manifold, without further restrictions on the metric. A natural continuation of this paper is represented by the work \cite{BPhys} by Becker, Benini, Schenkel and Szabo, where the issue of an Abelian gauge theory with duality on globally hyperbolic spacetimes is tackled in a fully covariant fashion. This is made possible by a modification of Cheeger-Simons differential characters encompassing differential characters with compact support and smooth Pontryagin duality, discussed in detail in \cite{BMath}.\\

It is from here that our work starts. We confront ourselves with the question of how to build a state for such a QFT. Differential cohomology leads to codify the observables into commutative diagrams of Abelian groups whose rows and columns are short exact sequences. The central object, which encodes all the information, is troublesome to manipulate.
As we aim at constructing a state on it, the idea that guides us is to present the central pre-symplectic space as the direct sum of other pre-symplectic spaces of the diagram, more familiar and easier to handle, so that this decomposition is induced at the level of the associated algebras. We can hence split up the problem and pursue the construction of states on the various sectors in which the total algebra has been decomposed separately.

Looking at the diagram of differential cohomology with compact support \cite[Diagram (5.30)]{BMath}, we can confidently expect to obtain three different sectors for our theory: a torsion sector, a topological sector, and a differential sector. As far as the last one is concerned, we expect that, on account of the dynamical constraints, the two-point function of a state on its algebra of observables will have a too small domain to be a bidistribution. Therefore, strictly speaking, it does not make sense to inquire after the structure of its wave front set. What we will actually do is to look for states that are Hadamard in a weak sense, that is to say, states that are the restriction of some proper Hadamard state.

In \cite{BPhys} it is shown that the QFT on differential cohomology is compatible with quantum Abelian duality, where with Abelian duality we mean the generalisation to arbitrary spacetime dimension of the duality between electric and magnetic field in Maxwell theory. Once we have a state, we may wonder how it connects with duality. In particular, we can investigate if duality is compatible with the above decomposition and how it emerges at the level of the GNS triple.\\

After having sketched the framework and motivated the problems we are going to deal with, let us briefly outline the content of the thesis, articulated into four chapters.

In the first chapter we recall some definitions and results in a fragmentary way, just to fix the notation and to make fully intelligible the terms used throughout the work.

In the second chapter, Cheeger-Simons differential characters are presented and differential cohomology theory is built out of it. Then, differential cohomology with compact support and Pontryagin duality are addressed, following \cite{BMath}. Incidentally, a number of pairings is derived, for a later use. Lastly, the construction of the convariant QFT is discussed, according to \cite{BPhys}.

The third and the fourth chapters represent the original part of the thesis. In the third one, after constructing the pre-symplectic structures on the diagram of the observables, we show that there exist splittings realising the desired decomposition. We then discuss the two-dimensional and the four-dimensional cases separately. After observing that in a wide range of examples the torsion sector is vanishing, we build a state for the topological sector and a Hadamard state in the weak sense for the differential sector. In particular, the case $M=\R\times \S^1$ is discussed in full detail, from the differential characters to the state. In conclusion, some remarks are given for spacetimes with non-compact Cauchy surface.

The fourth chapter concerns quantum Abelian duality. After exhibiting a duality for our QFT, we show that it is possible to choose the splittings so that they are compatible with the duality. We then prove that at the level of the GNS Hilbert space the duality is implemented by unitary operators and we discuss the properties of the GNS representation.

\chapter{Notation and conventions}

In this first chapter we will briefly recall some definitions and results, just in order to fix the notation and to make clear the conventions adopted throughout the thesis, without any attempt of completeness.

\section{Fourier transform}
We will make use of the so-called \emph{ordinary frequency convention} for the Fourier transform.

\begin{defn}[Fourier transform]
Let $f:\mathbb{R} \to \mathbb{C}$ be an integrable function. We call \emph{Fourier transform} of $f$ the function defined by:
\begin{equation*}
\hat{f}(\xi)=\int_{-\infty}^{+\infty} dx\, f(x)\, e^{-2\pi i \xi x}.
\end{equation*}
\end{defn}

\begin{defn}[Inverse Fourier transform]
Let $f:\mathbb{R} \to \mathbb{C}$ be an integrable function. We call \emph{inverse Fourier transform} of $f$ the function defined by:
\begin{equation*}
\check{f}(\xi)=\int_{-\infty}^{\infty} dx\, f(x)\,e^{2\pi i \xi x}.
\end{equation*}
\end{defn}
The convention for the Fourier series is chosen accordingly.\\

We denote the n-sphere by $\S^n$. In particular, we assume the radius of $\S^1$ to be equal to $1/2\pi$, i.e. we identify $\S^1\simeq \R/\Z$. Consequently:
$$
\int_{\S^1} d\mu(\S^1) =1,
$$
where $\mu$ is the Haar measure on $\S^1$.

\section{Smooth manifolds}
Smooth manifolds will be the background of our construction. We refer the interested reader to \cite{LEE03} for additional insights.

\begin{defn}
We call \emph{smooth manifold} of dimension $n$ a second-countable, Hausdorff topological space that is locally homeomorphic to $\R^n$, equipped with an atlas whose transition maps are smooth.
\end{defn}

Unless otherwise stated, all manifolds are assumed to be smooth and paracompact.

\begin{defn}
A \emph{Lorentzian manifold} $({M}_n,g)$ is a smooth $n$-dimensional manifold ${M}_n$ endowed with a non-degenerate, smooth, symmetric metric tensor $g\in \otimes^2_s T^\ast {M}_n$ with signature $(n-1,1)$\footnote{A metric tensor has signature $(p,q)$ if, as a quadratic form, it has $p$ positive eigenvalues and $q$ negative eigenvalues.}. An oriented and time-oriented Lorenzian manifold is called \emph{spacetime}.
\end{defn}

\begin{defn}[Globally hyperbolic spacetime]
Let $(M,g)$ be a spacetime. We call \emph{Cauchy surface} a closed achronal set $\Sigma\subset M$ whose domain of dependence $\mathcal{D}(\Sigma)$ satisfies $\mathcal{D}(\Sigma)=M$. A spacetime $(M,g)$ is called \emph{globally hyperbolic} if it admits a Cauchy surface.
\end{defn}

\section{Algebraic formulation and states}
In the algebraic formulation of quantum theory \cite{Haag}, a physical system is described in terms of an abstract $C^\ast$-algebra, whose self-adjoint elements represent the physical observables. The advantage of such an approach is that it provides with a powerful and versatile theory, independent of the choice of a specific Hilbert space. For a complete description and characterization of the algebraic approach in QFT see \cite{Mor13,Str05,BDFY15}.\\

The properties of the system can be reconstructed once a \emph{state} is assigned.

\begin{defn}
Let $\mathcal{A}$ be a unital $\ast$-algebra. We call \emph{state} a linear functional $\omega:\mathcal{A}\to \mathbb{C}$ with the following properties:
\begin{enumerate}[(i)]
\item $\omega(a^\ast a)\geq 0 \quad \forall a\in \mathcal{A}$ (Positivity);
\item $\omega(\mathbb{I})=1$ (Normalization),
\end{enumerate}
where $^\ast:\mathcal{A}\to \mathcal{A}$ is the $\ast$-algebra involution and $\mathbb{I}$ is the unit of $\mathcal{A}$.
\end{defn}

The following result, due to Gelfand, Najmark and Segal - henceforth GNS theorem - states that the Hilbert space description can be recovered from the algebraic one, in a way which is unique up to unitary isomorphisms.

\begin{thm}[GNS Theorem]
Let $\mathcal{A}$ be a unital $C^\ast$-algebra with unit $\mathbb{I}$ and $\omega:\mathcal{A}\to \mathbb{C}$ a positive linear functional such that $\omega(\mathbb{I})=1$. Then:
\begin{enumerates}[(a)]
\itemsep-.3em
\item There exists a triple $(\mathcal{H}_\omega,\pi_\omega,\Psi_\omega)$, where $\mathcal{H}_\omega$ is a Hilbert space, $\pi_\omega:\mathcal{A}\to \mathcal{B}\mathcal{L}(\mathcal{H}_\omega)$ a representation of $\mathcal{A}$ over $\mathcal{H}_\omega$ and $\Psi_\omega$ a vector of $\mathcal{H}_\omega$, for which:

\begin{enumerates}[(i)]
\item $\Psi_\omega$ is cyclic for $\pi_\omega$, i.e. $\pi_\omega (\mathcal{A})\Psi_\omega$ is dense in $\mathcal{H}_\omega$;
\item $\langle \Psi_\omega\,\vline\,\pi_\omega(a)\Psi_\omega\rangle_{\mathcal{H}_\omega}=\omega(a)\quad \forall a\in\mathcal{A}$.
\end{enumerates}

\item If $(\mathcal{H},\pi,\Psi)$ is another triple satisfying (i) and (ii), then there exists a unitary operator $U:\mathcal{H}_\omega\to\mathcal{H}$ such that $\Psi=U\Psi_\omega$ and
$$
\pi(a)=U\pi_\omega(a)U^{-1}\quad \forall a\in \mathcal{A}.
$$
\end{enumerates}
\end{thm}

Let $({M},g)$ be a connected globally hyperbolic spacetime and let $P: C^\infty (M) \to C^\infty(M)$ a Cauchy hyperbolic operator \cite{B15}. Define $S:=\{f\in C^\infty(M)\,\,\vline\,\, Pf=0 \}$. 

\begin{defn}[Advanced and retarded Green operators {\cite[Definition 3.4.1]{BGP}}]
A linear map $G^+:C_c^\infty(M)\to C^{\infty}(M)$ satisfying:
\begin{enumerates}[(i)]
\item $P\circ G^+=\mathrm{Id}_{C_c^\infty(M)}$;
\item $G^+\circ P\upharpoonright _{C_c^\infty(M)}=\mathrm{Id}_{C_c^\infty(M)}$;
\item $\mathrm{supp}(G^+ f)\subseteq J_+(\mathrm{supp}(f))$, for all $f\in C_c^\infty(M)$,
\end{enumerates}
is called \emph{advanced Green operator for $P$}. 

Similarly, a linear map $G^-:C_c^\infty(M)\to C^{\infty}(M)$ satisfying (i), (ii) and
\begin{enumerate}[(iii')]
\item $\mathrm{supp}(G^- f)\subseteq J_-(\mathrm{supp}(f))$, for all $f\in C_c^\infty(M)$
\end{enumerate}
is called \emph{retarded Green operator for $P$}.
\end{defn}

If $P$ is a normally hyperbolic operator the advanced and retarded Green operators always exist; what is more, they are unique \cite{BGP}. This is the case, for instance, of the d'Alambert-de Rham operator.

\begin{defn}[Causal propagator]
The linear map $$G=G^+-G^-: C_c^\infty(M)\to C^\infty(M)$$ is called \emph{causal propagator for $P$}.
\end{defn}

The salient properties of the causal propagator are summarized by the following theorem.
\begin{thm}
With the above notation, the following sequence is exact:
\begin{equation}
\xymatrix{
0 \ar[r] & C_c^\infty(M) \ar[r]^-P &  C_c^\infty(M) \ar[r]^-G & C^\infty(M) \ar[r]^-P & C^\infty(M) \ar[r] & 0.
}
\end{equation}
\end{thm}
The causal propagator allows to define a weakly non-degenerate bilinear map:
$$
\sigma(f,h):=\int_M f (Gh) \mathrm{dVol}_M, \qquad f,h\in C_c^\infty(M).
$$
Let $\mathcal{A}$ be the Weyl algebra whose generators are the abstract symbols $\mathcal{W}(Gf)$, $f\in C_c^\infty(M)$ satisfying the relations:
\begin{align*}
\mathcal{W}(Gf)\mathcal{W}(Gh)=e^{2\pi i \sigma(f,h)}\mathcal{W}(G(f+h)).
\end{align*}
The involution:
$$
^\ast: \mathcal{A}\to \mathcal{A}, \qquad \mathcal{W}(Gf) \mapsto\mathcal{W}(-Gf),
$$
endows $\mathcal{A}$ with a $\ast$-algebra structure.

\begin{defn}
Let $\omega:\mathcal{A} \to \mathbb{C}$ be an analytic state \cite{BR03}. We call \textit{two-point function} of $\omega$ the map:
\begin{eqnarray*}
&\omega_2: C_c^\infty(M)\times  C_c^\infty(M) \to \mathbb{C} &\\
&(f_1,f_2)\mapsto -\dfrac{\partial^2}{\partial s \partial t}\left[ \omega(\mathcal{W}(sGf_1 + tGf_2))e^{\frac{1}{2}ist\sigma(f_1,f_2)}    \right]_{|_{s=t=0}}.&
\end{eqnarray*}
\end{defn}

The higher $n$-point functions
$$
\omega_n: \underset{\mathrm{n}\,\mathrm{times}}{\underbrace{C_c^\infty(M)\times \cdots \times C_c^\infty(M)}} \to \mathbb{C} 
$$
are defined similarly (see \cite{WAL94} for details). Assuming that the $n$-point functions are continuous with respect to the standard topology on $C_c^\infty(M)$, the Schwartz kernel theorem \cite[Theorem 5.2.1]{Hor} allows to write them in terms of their distributional kernel:
\begin{equation}
\omega_n(f_1,\dots,f_n)=\int_{M^n} \omega_n(x_1,\dots,x_n)f_1(x_1)\cdots f_n(x_n) \,\mathrm{dVol_{M^n}}.
\end{equation}

\begin{defn}[Quasifree state {\cite[Definition 5.2.22]{BDFY15}}]
A state $\omega:\mathcal{A}\to \mathbb{C}$ is called \emph{quasifree} if its $n$-point functions satisfy:
\begin{equation*}
\begin{cases}
& \omega_n(f_1,\dots,f_n)=0, \qquad \text{for }n\text{ odd}\\
& \omega_n(f_1,\dots,f_n)=\sum_{\Pi} \omega_2(f_{i_1},f_{i_2})\cdots\omega_2(f_{i_{n-1}},f_{i_n}), \qquad \text{for }n\text{ even},
\end{cases}
\end{equation*}
where $\Pi$ denotes all the possible partitions of the set $\{1,2,\dots,n\}$ into pairs $$\{i_1,i_2\},\dots,\{i_{n-1},i_n\}$$ with $i_{2j-1}< i_{2k}$ for $j=1,2,\dots, n/2$.
\end{defn}

There is a distinguished class of states, called \emph{Hadamard states}, of particular interest to physics characterized by a constraint on the distributional singularity structure of their two-point functions \cite{Rad96}. In fact, an arbitrary state is in general too singular for physical purposes. Hadamard states, besides having an ultraviolet behaviour mimicking that of the Poincar\'e vacuum, guarantee that the quantum fluctuations of all observables are bounded. They also allow for an extension of the algebra of fields to encompass the Wick polynomials. The definition of a Hadamard state relies upon microlocal analysis tools and concepts: we refer the reader to \cite{Hor} for a complete and detailed explanation of the mathematical objects involved.

\begin{defn}[Hadamard state]
A state $\omega:\mathcal{A}\to \mathbb{C}$ is a \emph{Hadamard state} if its two-point function $\omega_2$ satisfies the \emph{microlocal spectrum condition}:
\begin{equation}
WF(\omega_2)=\{(x,y,k_x,-k_y)\in T^\ast M^2\setminus\{0\} \,\,\vline\,\, (x,k_x)\sim (y,k_y), k_x \vartriangleright 0 \},
\end{equation}
where $WF$ denotes the wavefront set and $(x,k_x)\sim (y,k_y)$ means that there is a lightlike geodesic connecting $x$ to $y$, such that $k_x$ is cotangent to the geodesic and $k_y$ is the parallel transport of $k_x$ along the geodesic. By $k_x \vartriangleright 0$ we mean that $k_x$ is non-vanishing and $k_x(v)\geq 0$ for all future-directed $v\in T_xM$.
\end{defn}

Observe that the above definition implies that, given two Hadamard states $\omega$ and $\tilde{\omega}$, the difference of their two-point functions is smooth:
$$
\omega_2-\tilde{\omega}_2\in C^\infty (M\times M;\mathbb{C}).
$$

\section{Fock space}\label{sec:1:Fock}
Let $\mathcal{H}$ be a complex Hilbert space and denote by $\mathcal{H}^{\otimes n}=\mathcal{H}\otimes \mathcal{H}\otimes \cdots \otimes \mathcal{H}$ the $n$-fold tensor product of $\mathcal{H}$ with itself.

\begin{defn}[Fock space]
The direct sum:
$$
\mathfrak{F}(\mathcal{H}):=\bigoplus_{n\in \N} \mathcal{H}^{\otimes n}
$$
with $\mathcal{H}^{\otimes 0}=\mathbb{C}$ is called \emph{Fock space}.
\end{defn}

A generic element in $\mathfrak{F}(\mathcal{H})$ can be seen as a sequence $\{\psi_n\}_{n\in\N}$, with $\psi_n\in\mathcal{H}^{\otimes n}$. Introduce on the operators $P_{\pm}:\mathfrak{F}(\mathcal{H})\to \mathfrak{F}(\mathcal{H})$ defined by the continuous linear extension of:
\begin{align*}
& P_+(\varphi_1\otimes \varphi_2 \otimes \cdots \otimes \varphi_n)=\dfrac{1}{n!}\sum_{\pi\in S_n} \varphi_{\pi(1)}\otimes \varphi_{\pi(2)} \otimes \cdots \otimes \varphi_{\pi(n)}, \\
& P_-(\varphi_1\otimes \varphi_2 \otimes \cdots \otimes \varphi_n)=\dfrac{1}{n!}\sum_{\pi\in S_n} \sigma(\pi)\,\, \varphi_{\pi(1)}\otimes \varphi_{\pi(2)} \otimes \cdots \otimes \varphi_{\pi(n)}, 
\end{align*}
where $S_n$ is the symmetric group of degree $n$ and $\sigma(\pi)$ denotes the sign of the permutation $\pi\in S_n$.

\begin{defn}
We call \emph{symmetrized Fock space} the image of $\mathfrak{F}(\mathcal{H})$ via $P_+$:
$$
\mathfrak{F}_+(\mathcal{H}):=P_+ \mathfrak{F}(\mathcal{H}).
$$
Analogously, we call \emph{antisymmetrized Fock space} the image of $\mathfrak{F}(\mathcal{H})$ via $P_-$:
$$
\mathfrak{F}_-(\mathcal{H}):=P_- \mathfrak{F}(\mathcal{H}).
$$
\end{defn}

We introduce the \emph{number operator} $N$:
\begin{eqnarray*}
& \mathcal{D}(N):=\{\{\psi_n\}_{n\in\N}\in \mathfrak{F}(\mathcal{H})\,\,\vline\,\, \sum_{n\in\N} n^2 \norm{\psi_n}^2<+\infty    \}, &\\
& N: \mathcal{D}(N)\to \mathfrak{F}(\mathcal{H})&\\
&\{\psi_n\}_{n\in\N}\mapsto \{n\psi_n\}_{n\in \N},&
\end{eqnarray*}
which we interpret as ``counting'' the number of particles in each $n$-particle subspace $\mathcal{H}^{\otimes n}$ of $\mathfrak{F}(\mathcal{H})$. $N$ is essentially self-adjoint, as can be seen directly from its spectral decomposition. 

There exists a distinguished vector in $\mathfrak{F}(\mathcal{H})$ called \emph{vacuum vector}:
$$
\Psi:=(1,0,0,\dots).
$$
We can resort to the vacuum vector to define the creation and the annihilation operators, $a^\dagger$ and $a$, which prove to be crucial in many respects. For each $\varphi\in\mathcal{H}$, $a^\dagger(\varphi)$ and $a(\varphi)$ are defined as follows: 
\begin{align*}
& a(\varphi)\Psi=0,\\
& a^\dagger(\varphi)\Psi=(0,\varphi,0,\dots),\\
&a(\varphi)(\varphi_1\otimes\varphi_2\otimes \cdots \otimes \varphi_n)=n^{1/2} \langle \varphi\, \vline\, \varphi_1\rangle_\mathcal{H} \varphi_2 \otimes \cdots \otimes \varphi_n,\\
&a^\dagger(\varphi)(\varphi_1\otimes\varphi_2\otimes \cdots \otimes \varphi_n)=(n+1)^{1/2} \varphi \otimes \varphi_1\otimes\varphi_2\otimes \cdots \otimes \varphi_n
\end{align*}
The creation and annihilation operators can be extended by linearity to densely defined operators on $\mathfrak{F}(\mathcal{H})$; it can be proved that they admit a well-defined extension to the domain $\mathcal{D}(N^{1/2})$ of the operator $N^{1/2}$. At last, we define symmetrized/antisymmetrized creation and annihilation operators on the Fock space $\mathfrak{F}(\mathcal{H})$ by:
\begin{align*}
& a_\pm(\varphi):=a(\varphi)P_\pm,\\
& a^\dagger_\pm(\varphi):=P_\pm a^\dagger (\varphi).
\end{align*}

On Fock spaces, there is a procedure that allows to promote operators defined on the one-particle Hilbert space $\mathcal{H}$ to operators on the whole symmetrized/antisymmetrized Fock spaces. This procedure is called \emph{second quantization} \cite{BR03}. Let us discuss it for unitary operators. Let $U:\mathcal{H}\to \mathcal{H}$ be a unitary operator. Set $U_0=\mathbb{I}$ and:
\begin{equation*}
U_n(P_\pm(\varphi_1\otimes\varphi_2\otimes \cdots \otimes \varphi_n)):=P_\pm(U\varphi_1\otimes U\varphi_2 \otimes \cdots \otimes U\varphi_n)
\end{equation*}
for $n\in\N_0$, and then extend them by continuity. We call second quantization of $U$ the operator $\Gamma(U):\mathfrak{F}_\pm(\mathcal{H})\to \mathfrak{F}_\pm(\mathcal{H})$ given by:
\begin{equation}\label{eq:1:unitary_second_quantization}
\Gamma(U):= \bigoplus_{n\in\N} U_n.
\end{equation}
It is a straightforward check that $\Gamma(U)$ is unitary.
\chapter{Differential Cohomology and Covariant QFT}\label{chapter:2}
In this chapter, after introducing briefly the salient elements of Cheeger-Simons theory of differential characters, we will analyse in quite some detail how a covariant quantum field theory can be built out of it. 

\section{Differential cohomology}
Differential cohomology is an algebro-differential construction which puts in connection algebraic homology and cohomology with the geometry of differential forms. To some extent, it allows to refine cohomology through differential forms. 

Cheeger and Simons \cite{CS85} were the first, in 1985, who investigated differential cohomology along with its graded commutative ring structure. Various constructions later developed, such as smooth Deligne cohomology, differential cocycles and De Rham-Federer currents, were recognised to be different models for the same theory. Simons-Sullivan \cite{SS08} and B\"ar-Becker \cite{BB14} showed that differential cohomology is uniquely determined up to unique natural equivalences. In our work, for practical purposes and for the sake of clarity we will adopt the Cheeger-Simons model; notwithstanding, the overall philosophy is model-independent.

\subsection{Smooth singular homology and cohomology}
Just to fix the notation, we succinctly recall a few definitions and results regarding smooth singular homology and cohomology. We refer the interested reader to \cite{HAT02} for further insights and detailed proofs.\\

Let $M$ be a smooth manifold. We denote by $C_\ast(M;\mathbb{Z})$ the chain complex of smooth singular chains in $M$ with coefficients in $\Z$. For each $k\geq 0$ there exists a boundary map
$$
\partial_k: C_k(M;\Z) \to C_{k-1}(M;\mathbb{Z})
$$
which is a homomorphism of free Abelian groups fulfilling $\partial_{k+1}\circ \partial_k=0$. The subscript $_k$ will be dropped in the following whenever the domain and the target spaces are clear from the context. There are two distinguished subgroups of $C_k(M;\Z)$ for every $k\in\N$: the \emph{smooth singular k-cycles} $Z_k(M;\Z):= Ker(\partial_{k})$ and the \emph{smooth singular k-boundaries} $B_k(M;\Z):=Im(\partial_{k+1})$.
\begin{defn}
The \emph{k-th smooth singular homology group} is the quotient:
\begin{equation}
H_k(M;\Z):=\dfrac{Z_k(M;\Z)}{B_k(M;\Z)}=\dfrac{Ker(\partial_k)}{Im(\partial_{k+1})}.
\end{equation}
\end{defn} 

Let $S\subseteq M$ be a smooth submanifold of $M$. The inclusion $\iota: S \hookrightarrow M$ allows to identify $C_\ast(S;\Z)$ as a subcomplex in $C_\ast(M;\Z)$. We define the complex of \emph{smooth singular chains on $M$ relative to $S$} to be the quotient:
\begin{equation*}
C_\ast(M,S;\Z):=\dfrac{C_\ast(M;\Z)}{C_\ast(S;\Z)}.
\end{equation*}
It can be shown that $C_k(M,S;\Z)$ turns out to be a free Abelian group and that, again, the boundary homomorphism $\partial_k:C_k(M,S;\Z)\to C_{k-1}(M,S;\Z)$ is well defined. Its kernel is the Abelian group of relative cycles $Z_k(M,S;\Z)$, whereas its image is the Abelian group of relative boundaries $B_{k-1}(M,S;\Z)$.  Therefore, we call \emph{relative k-th homology group} the quotient
$$
H_k(M,S;\Z):=\dfrac{Z_k(M,S;\Z)}{B_k(M,S;\Z)}=\dfrac{Ker(\partial_k)}{Im(\partial_{k+1})}.
$$

Given an Abelian group $G$, the cochain complex $C^\ast(M;G)$ of G-valued smooth singular cochains in $M$ is defined by 
$$C^k(M;G):=\Hom(C_k(M;\Z);G),$$ where the coboundary map is obtained by dualizing the boundary homomorphism:
\begin{gather*}
\delta_k: C^k(M;G)\to C^{k+1}(M;G)\\
\varphi\mapsto \delta_k \varphi= \varphi \circ \partial_{k+1}.
\end{gather*}
The Abelian groups of $k$-th $G$-cocycles and $k$-th $G$-coboundaries are respectively defined as $Z^k(M;G):=Ker(\delta^k)$ and $B^k(M;G):=Im(\delta^{k-1})$. 
\begin{defn}
We define the \emph{$k$-th smooth singular cohomology group with $G$ coefficients as} the quotient:
$$
H^k(M;G):=\dfrac{Z^k(M;G)}{B^k(M;G)}.
$$
\end{defn}
The relative version of a smooth singular cohomology is obtained analogously, replacing everywhere the homology groups with their relative counterparts.

As far as categorical properties are concerned, let $\mathsf{Man}$ be the category whose objects are smooth manifolds and whose morphisms are smooth maps between smooth manifolds. Furthermore, let $\mathsf{Ab}$ be the category whose objects are Abelian groups and whose morphisms are group homomorphisms. Then, $H_k(\,\cdot\,;\Z):\mathsf{Man} \to \mathsf{Ab}$ is a covariant functor from $\mathsf{Man}$  to $\mathsf{Ab}$, whereas $H^k(\arg;G): \mathsf{Man}\to \mathsf{Ab}$ is a contravariant functor between the same categories.\\

The cohomology group $H^k(M;G)$ is not isomorphic to $\Hom(H_k(M;\Z);G)$ in general. Their relation is set by the universal coefficient theorem for cohomology \cite[Theorem 3.2]{HAT02}:
\begin{equation}\label{seq:2:Universal_coefficient_cohomology}
\xymatrix@C-=0.5cm{
0 \ar[r] & \text{Ext}(H_{k-1}(M;\Z),G) \ar[r] & H^k(M;G) \ar[r] & \Hom(H_k(M;\Z);G) \ar[r]  &0.\\
}
\end{equation}
It entails that the homomorphism $$h: H^k(M;G) \to \Hom(H_k(M;\Z);G)$$ is always an epimorphism but it may fail being a monomorphism. However, when $G=\R,\T$, $h$ turns out to be an isomorphism, because $\text{Ext}(\arg,\R)=0$ and $\text{Ext}(\arg, \T)=0$, being $\R$ and $\T$ divisible. Contrariwise, for $G=\Z$, $\text{Ext}(\arg,\Z)\neq 0$ and $H^k(M;\Z)\not\simeq \Hom(H_k(M;\Z);\Z)$ in general. We denote by $H_{\tor}^k(M;\Z)$ the image of $\text{Ext}(H_{k-1}(M;\Z),\Z)$ in $H^k(M;\Z)$ in \eqref{seq:2:Universal_coefficient_cohomology}, namely the \emph{k-th torsion subgroup}, and by $H_{\free}^k(M;Z):=H^k(M;\Z)/H^k_{\tor}(M;\Z)$ the \emph{k-th free cohomology group}.

\subsection{Differential Characters}
\begin{defn}
Let $M$ be a smooth manifold and let $\Omega^k(M)$ be the vector space of differential forms of degree $k$ on $M$. The Abelian group of Cheeger-Simons differential characters of degree $k\in\N\setminus \{0\}$ is defined by:
\begin{equation}\label{def:2:differential_characters}
\hat{H}^k(M;\Z):=\left\{ h\in \Hom(Z_{k-1}(M;\Z);\T)\,\,|\,\, h\circ \partial \in \Omega^k(M) \right\},
\end{equation}  
where $h\circ\partial \in\Omega^k(M)$ means that there exists $\omega_h\in \Omega^k(M)$ such that
\begin{equation}\label{eq:2:differential_characters_omega1}
h(\partial \gamma)=\int_\gamma \omega_h \quad\mod \Z
\end{equation}
for all $\gamma\in C_k(M;\Z)$.
\end{defn}
The Abelian group structure is inherited from $\T$. Resorting to the additive notation, we have that $(h+h^\prime)(z):=h(z)+h^\prime(z)$, $\forall h,h^\prime \in \hat{H}^k(M;\Z)$ and $\forall z\in Z_{k-1}(M;\Z)$. The unit element is the constant homomorphism $\textbf{0}\in \hat{H}^k(M;\Z)$, $z\mapsto \textbf{0}(z)=1\in \T\,\,\forall z\in Z_{k-1}(M;\Z)$, while for every $h\in \hat{H}^k(M;\Z)$ the inverse element $-h$ is defined by $-h(z):=(h(z))^{-1}$ $\forall z\in Z_{k-1}(M;\Z)$.\\

There exist four noteworthy homomorphisms having $\hat{H}^k(M;\Z)$ as domain or target. Firstly, it is straightforward to prove that $\omega_h$ is uniquely determined by \eqref{eq:2:differential_characters_omega1}. Furthermore, $\omega_h$ is closed and with integer periods. Introducing the notation
$$
\Omega_\Z ^k(M):=\left\{ \omega \in \Omega^k(M) \,\, \vline \,\, \int_z \omega \in \Z\quad \forall z\in Z_k(M;\Z) \right\},
$$
we have that $\omega_h\in \Omega_\Z^k(M)\subseteq \Omega_d^k(M)\subseteq \Omega^k(M)$. It is therefore defined the \emph{curvature} homomorphism:
\begin{gather}\label{def:2:curvature}
\curv : \hat{H}^k(M;\Z)\to \Omega_\Z^k(M) \nonumber\\
h \mapsto \omega_h.
\end{gather}

Secondly, consider any element $h\in \hat{H}^k(M;\Z)$. Since $Z_{k-1}(M;\Z)$ is a free $\Z$-module, $h$ can be lifted to $\tilde{h}\in \Hom(Z_{k-1}(M;\Z);\R)$ along the quotient map $\R\to \T$. Construct a k-cochain via
\begin{gather*}
\eta_{\tilde{h}} : C_k(M;\Z) \to \R \\
\gamma \mapsto \eta_{\tilde{h}}(\gamma) = \int_\gamma \curv h - \tilde{h}\circ \partial (\gamma).
\end{gather*}
It is an easy check that $\eta_{\tilde{h}}$ defines an integer-valued coclosed cochain, i.e. $\delta \eta_{\tilde{h}}=0$. Therefore, the cohomology class $[\eta_{\tilde{h}}]\in H^k(M;\Z)$ is well defined and independent of the choice of the lifting. We call \emph{characteristic class} the homomorphism given by:
\begin{gather}\label{def:2:characteristic_class}
\cha : \hat{H}^k(M;\Z) \to H^k(M;\Z) \nonumber\\
h\mapsto [\eta_{\tilde{h}}].
\end{gather}

$\mathrm{curv}$ and $\mathrm{char}$ are surjective but, in general, not injective. This means that differential characters contain further information stored in the kernel of the above maps. Consider an arbitrary element $u\in H^{k-1}(M;\T)$. As $\T$ is a divisible group, the universal coefficient theorem yields an isomorphism:
$$
iso: H^{k-1}(M;\T)\to \Hom(H_{k-1}(M;\Z);\T).
$$
Denoting by $\pi: Z_{k-1}(M;\Z) \to H_{k-1}(M;\Z)$ the quotient homomorphism, set 
$$
\kappa (u):=iso(u) \circ \pi : Z_{k-1}(M;\Z) \to \T.
$$
We call \emph{inclusion of flat classes} the homomorphism:
\begin{gather}\label{def:2:kappa}
\kappa: H^{k-1}(M;\T) \to \hat{H}^k(M;\Z) \nonumber\\
u\mapsto \kappa (u).
\end{gather}
By construction the curvature of $\kappa (u)$ is vanishing.

At last, let $A\in \Omega^{k-1}(M)$. It defines a differential character $h_A\in \hat{H}^k(M;\Z)$ via:
\begin{equation}
h_A(z)=\int_z A \quad\mod\Z, \qquad \forall z\in Z_{k-1}(M;\Z).
\end{equation}
It is well-defined and its curvature is given by $\curv A=dA$. As a matter of fact,
$$
h_A(\partial \gamma)=\int_{\partial \gamma} A \quad \mod\Z= \int_\gamma dA \quad \mod \Z, \qquad \forall \gamma \in C_k(M;\Z).
$$
The characteristic class of $h_A$ is trivial, as one realizes picking the real lift $\int_\cdot A \in \Hom (Z_{k-1}(M;\Z);\R)$. We call \emph{topological trivialization} the injective homomorphism
\begin{gather}	\label{def:2:iota}
\iota: \dfrac{\Omega^{k-1}(M)}{\Omega_\Z^{k-1}(M)}\to \hat{H}^k(M;\Z) \nonumber\\
[A]\mapsto h_A.
\end{gather}

Cheeger and Simons \cite{CS85} showed that differential characters together with the maps just defined fit into the following commutative diagram, in which rows and columns are short exact sequences:

\begin{equation}\label{diag:2:diff_char}
\xymatrix{
  & 0 \ar[d] & 0\ar[d] & 0\ar[d] &\\
0\ar[r] & \dfrac{H^{k-1}(M;\R)}{H^{k-1}_{\free}(M;\Z)} \ar[d] \ar[r]& \dfrac{\Omega^{k-1}(M)}{\Omega^{k-1}_\Z(M)} \ar[d]^-\iota\ar[r]^-d & d\Omega^{k-1}(M)\ar[d] \ar[r] & 0\\
0\ar[r] & H^{k-1}(M;\T) \ar[d]\ar[r]^-\kappa& \hat{H}^k(M;\Z) \ar[d]^-\cha\ar[r]^-\curv& \Omega^k_\Z(M)\ar[d] \ar[r]& 0\\
0\ar[r] & H^k_{\tor}(M;\Z)\ar[d]\ar[r] & H^{k}(M;\Z) \ar[d]\ar[r]& H^k_{\free}(M;\Z) \ar[d]\ar[r]&0\\
  & 0 & 0 & 0 &\\
}
\end{equation}

In addition to this, the assignment of $\hat{H}^k(M;\Z)$ to $M$ defines a contravariant functor from the category of smooth manifolds $\mathsf{Man}$ to the category of Abelian groups $\mathsf{Ab}$; the same holds true for all of the other Abelian groups in diagram \eqref{diag:2:diff_char}. Furthermore, the homomorphisms in diagram \eqref{diag:2:diff_char} are natural transformations.

It is worth noticing that differential cohomology is uniquely defined by diagram \eqref{diag:2:diff_char} up to unique natural isomorphisms.

\begin{defn}[Differential cohomology theory {\cite[Definition 5.9]{BB14}}]   We call differential cohomology theory a contravariant functor $\tilde{H}^\ast(\arg;\Z)$ from the category of smooth manifold $\mathsf{Man}$ to the category of $\Z$-graded Abelian groups $\mathsf{Ab}^\Z$, together with four natural transformations:
\begin{enumerates}[$\bullet$]
\item $\widetilde{\curv}: \widetilde{H}^\ast(\arg;\Z) \Rightarrow \Omega^\ast_\Z(\arg)$, called \emph{curvature};
\item $\widetilde{\cha}:\widetilde{H}^\ast(\arg;\Z) 	\Rightarrow H^\ast(\arg;\Z) $, called \emph{characteristic class};
\item $\tilde{\iota}: \dfrac{\Omega^{\ast-1}(\arg)}{\Omega^{\ast-1}_\Z(\arg)}\Rightarrow \widetilde{H}^\ast(\arg;\Z)$, called \emph{topological trivialization};
\item $\tilde{\kappa}: H^{\ast-1}(\arg;\T) \Rightarrow \widetilde{H}^\ast(\arg;\Z)$, called \emph{inclusion of flat classes};
\end{enumerates}
such that, for any differentiable manifold $M$, the following diagram is commutative and has exact rows and columns:
$$
\xymatrix{
  & 0 \ar[d] & 0\ar[d] & 0\ar[d] &\\
0\ar[r] & \dfrac{H^{\ast-1}(M;\R)}{H^{\ast-1}_{\free}(M;\Z)} \ar[d] \ar[r]& \dfrac{\Omega^{\ast-1}(M)}{\Omega^{\ast-1}_\Z(M)} \ar[d]^-{\tilde{\iota}} \ar[r]^-d & d\Omega^{\ast-1}(M)\ar[d] \ar[r] & 0\\
0\ar[r] & H^{\ast-1}(M;\T) \ar[d]\ar[r]^-{\tilde{\kappa}}& \tilde{H}^\ast(M;\Z) \ar[d]^-{\tilde{\cha}}\ar[r]^-{\tilde{\curv}} & \Omega^\ast_\Z(M)\ar[d] \ar[r]& 0\\
0\ar[r] & H^\ast_{\tor}(M;\Z)\ar[d]\ar[r] & H^{\ast}(M;\Z) \ar[d]\ar[r]& H^\ast_{\free}(M;\Z) \ar[d]\ar[r]&0\\
  & 0 & 0 & 0 &\\
}
$$
\end{defn}

\begin{thm}[Uniqueness {\cite[Theorems 5.11-5.15]{BB14}}] 
Given an arbitrary differential cohomology theory $(	\tilde{H}^\ast(\arg;	Z), \tilde{\mathrm{curv}},\tilde{\mathrm{char}}, \tilde{\iota}, \tilde{\kappa})$ there exists a unique natural isomorphism
\begin{equation*}
\Xi: \tilde{H}^\ast(\arg;\Z) \Rightarrow \hat{H}^\ast(\cdot;\Z)
\end{equation*}
such that
$$
\Xi\circ \tilde{\mathrm{curv}}=\mathrm{curv},\qquad \Xi\circ\tilde{\mathrm{char}}=\mathrm{char},\qquad \Xi\circ \tilde{\iota}=\iota, \qquad \Xi\circ\tilde{\kappa}=\kappa.
$$
\end{thm}

Differential characters can be given a unique graded commutative ring structure. Following \cite{BB14}, we will adopt an axiomatic point of view:
\begin{defn}[Internal product] \label{def:2:ring_structure}
An \emph{internal product} of differential characters is a map:
\begin{equation}
\begin{split}
\cdot : \hat{H}^k(M;\Z)\times \hat{H}^l(M;\Z) &\to \hat{H}^{k+l}(M;\Z)\\
(h,h')&\mapsto h\cdot h'
\end{split}
\end{equation}
such that the following requirements are fulfilled:
\begin{enumerates}[(i)]
\item \emph{Ring Structure}. $(\hat{H}^\ast(M;\Z),+,\cdot )$ is a ring, i.e.\ $\cdot$ is associative and $\Z$-bilinear;
\item \emph{Graded commutativity}. It holds $$h\cdot h'=(-1)^{kl}h'\cdot h$$ for any $h\in \hat{H}^k(M;\Z)$ and $h'\in \hat{H}^l(M;\Z)$;
\item \emph{Compatibility with curvature}.  It holds $$\curv (h\cdot h')=\curv h \wedge \curv h'$$ for any $h\in \hat{H}^k(M;\Z)$ and $h'\in \hat{H}^l(M;\Z)$, i.e. the curvature is a ring homomorphism;
\item \emph{Compatibility with characteristic class}. It holds $$\cha(h\cdot h')=\cha h\smile \cha h'$$ for any $h\in \hat{H}^k(M;\Z)$ and $h'\in \hat{H}^l(M;\Z)$, where $\smile$ denotes the cup product of singular cohomology, i.e. the characteristic class is a ring homomorphism;
\item \emph{Compatibility with topological trivialization}. It holds $$\iota(A)\cdot h=\iota(A\wedge \curv h)$$ for any $A\in\frac{\Omega^{k-1}(M)}{\Omega_\Z^{k-1}(M)}$ and $h\in\hat{H}^l(M;\Z)$, where $\wedge$ denotes the wedge product of differential forms;
\item \emph{Compatibility with inclusion of flat classes}. It holds $$\kappa (u)\cdot h=\kappa(u\smile \cha h)$$ for any $u\in H^{k-1}(M;\T)$ and $h\in\hat{H}^l(M;\Z)$.
\end{enumerates} 
\end{defn}

\begin{rem}
Observe that in (v) and (vi) $\iota$ and $\kappa$ act on the first entry of the inner product respectively. When they happen to be in the second entry, the above formulas must be duly modified, taking into account the grading of the product:
\begin{equation}\label{def:2:ring_structure_SecondEntry}
h\cdot \iota(A)=(-1)^k \iota(\curv h\wedge A), \qquad h\cdot \kappa(u)=(-1)^k \kappa(\cha h \smile u)
\end{equation}
for all $h\in\hat{H}^k(M;\Z)$, $A\in \frac{\Omega^{l-1}(M)}{\Omega_\Z^{l-1}(M)}$ and $u\in H^{l-1}(M;\T)$.
\end{rem}

\begin{thm}[Uniqueness of ring structure {\cite[Corollary 6.5]{BB14}}] The ring structure on differential cohomology is determined uniquely by axioms $(i)-(v)$ in Definition \ref{def:2:ring_structure}.
\end{thm}

\subsection{Relative differential cohomology}\label{subsec:2:rel_diff_coho}
As a prolegomenon to differential cohomology with compact support, this section will be devoted to a review of differential characters on relative cycles. \\

Let $S\subseteq M$ be a submanifold. The Abelian group of $k$-differential forms vanishing on $S$ will be denoted by
$$
\Omega^k(M,S):=\left\{ \omega\in\Omega^k(M)\,\,\vline\,\, \omega_{|S}=0 \right\},
$$
while $\Omega^k_\Z(M,S)$ will denote the Abelian subgroup in $\Omega^k(M,S)$ of relative $k$-forms with integral periods on relative cycles. Since the exterior derivative $d$ preserves these subgroups, $(\Omega^\ast(M,S),d)$ is a subcomplex of $(\Omega^\ast(M),d)$. The de Rham cohomology of such a subcomplex is called \emph{de Rham cohomology of M relative to S} and denoted by $H_{dR}^\ast(M,S)$.

In the following, we will specialize to properly embedded submanifolds $S\subseteq M$, i.e. to submanifolds such that the embedding map $i:S\hookrightarrow M$ is proper. Such a restriction is motivated by the fact that, in this case, a relative version of the De Rham theorem is available. In fact, when $S$ is properly embedded, every differential form on $S$ admits a smooth extension to $M$. As a consequence, the following sequence is short exact:
\begin{equation*}
\xymatrix{
0 \ar[r]& \Omega^\ast(M,S) \ar[r]& \Omega^\ast(M) \ar[r] & \Omega^\ast(S) \ar[r] & 0.
}
\end{equation*}
Every differential form can be thought of as a cochain via integration; this yields a commutative diagram of short exact sequences:
\begin{equation*}
\xymatrix{
0 \ar[r] & \Omega^\ast(M,S) \ar[r] \ar[d]^-{\int_\cdot} & \Omega^\ast(M) \ar[r] \ar[d]^-{\int_\cdot} & \Omega^\ast(S)  \ar[r] \ar[d]^-{\int_\cdot} &0\\
0 \ar[r] & C^\ast(M,S;\R) \ar[r] & C^\ast(M;\R) \ar[r] & C^\ast(S;\R)  \ar[r]  &0
}
\end{equation*}
and a corresponding commutative diagram of long exact sequences in cohomology:
{\footnotesize
\begin{equation*}
\xymatrix@C-=0.45cm{
\cdots \ar[r] & H^{k-1}_{dR}(M) \ar[r] \ar[d]^-{\simeq} & H^{k-1}_{dR}(S) \ar[r] \ar[d]^-{\simeq}& H^{k}_{dR}(M,S) \ar[r] \ar[d] &  H^{k}_{dR}(M) \ar[r] \ar[d]^-{\simeq} & H^{k}_{dR}(S) \ar[r] \ar[d]^-{\simeq} &  \cdots\\
\cdots \ar[r] & H^{k-1}(M;\R) \ar[r]  & H^{k-1}(S;\R) \ar[r] & H^{k}(M,S;\R) \ar[r]  &  H^{k}(M;\R) \ar[r] & H^{k}(S;\R) \ar[r] &  \cdots
}
\end{equation*}
}
De Rham theorem \cite[Theorem 5.36]{WER83} guarantees that the first and the last two vertical arrows are isomorphisms; from the five lemma it ensues that also $H^k_{dR}(M,S)\to H^k(M,S;\R)$ is an isomorphism, thus providing a relative version of de Rham theorem.\\

We are now in the position to define relative differential characters:

\begin{defn}[Differential characters on relative cycles]
Let $M$ be a smooth manifold. The Abelian group of Cheeger-Simons relative differential characters on $M$ with respect to $S\subseteq M$ of degree $k\in\N\setminus \{0\}$ is defined by:
\begin{equation}\label{def:2:differential_characters}
\hat{H}^k(M,S;\Z):=\left\{ h\in \Hom(Z_{k-1}(M,S;\Z);\T)\,\,|\,\, h\circ \partial \in \Omega^k(M) \right\},
\end{equation}  
where $h\circ\partial \in\Omega^k(M)$ means that there exists $\omega_h\in \Omega^k(M)$ such that
\begin{equation}\label{eq:2:differential_characters_omega}
h(\partial \gamma)=\int_\gamma \omega_h \quad\mod \Z
\end{equation}
for all $\gamma\in C_k(M;\Z)$.
\end{defn}

\begin{rem}
Observe that $\omega_h$ actually lies in $\Omega^k(M,S)$ and has, in particular, integral periods. In fact, choosing $\gamma\in C_k(S)$ we obtain ${\omega_h }_{|S}=0$. Picking, moreover, $\gamma\in Z_{k}(M,S)$ gives $\int_\gamma \omega_h \in \Z$.
\end{rem}

With minor adaptations, retracing the reasoning made for absolute differential cohomology it is possible to introduce curvature, characteristic class, inclusion of flat classes and topological trivialization maps for the relative case.

\begin{prop}[{\cite[Theorem 3.5]{BMath}}] \label{prop:2:diag_rel_diff_char}  Let $M$ be a manifold and $S\subseteq M$ a properly embedded submanifold. Then the diagram
\begin{equation}\label{diag:2:relative_diff_char}
\xymatrix{
  & 0 \ar[d] & 0\ar[d] & 0\ar[d] &\\
0\ar[r] & \dfrac{H^{k-1}(M,S;\R)}{H^{k-1}_{\free}(M,S;\Z)} \ar[d] \ar[r]& \dfrac{\Omega^{k-1}(M,S)}{\Omega^{k-1}_\Z(M,S)} \ar[d]^-\iota\ar[r]^-d & d\Omega^{k-1}(M,S)\ar[d] \ar[r] & 0\\
0\ar[r] & H^{k-1}(M,S;\T) \ar[d]\ar[r]^-\kappa& \hat{H}^k(M,S;\Z) \ar[d]^-{\mathrm{char}}\ar[r]^-{\mathrm{curv}}& \Omega^k_\Z(M,S)\ar[d] \ar[r]& 0\\
0\ar[r] & H^k_{\tor}(M,S;\Z)\ar[d]\ar[r] & H^{k}(M,S;\Z) \ar[d]\ar[r]& H^k_{\free}(M,S;\Z) \ar[d]\ar[r]&0\\
  & 0 & 0 & 0 &\\
}
\end{equation}
is commutative and has exact rows and columns.
\end{prop}

\begin{rem}
Notice that the result just stated relies heavily on $S$ being properly embedded. If the inclusion $i:S\hookrightarrow M$ is not proper, several rows and columns may fail being exact.
\end{rem}

As for the absolute case, the assignment of $\hat{H}^k(M,S;\Z)$ to $(M,S)$ defines a contravariant functor
$$
\hat{H}^k(\arg;\Z): \mathsf{Pair} \to \mathsf{Ab}
$$
from the category $\mathsf{Pair}$ \footnote{The category $\mathsf{Pair}$ is the category whose objects are pairs $(M,S)$, with $M$ in $\mathsf{Man}$ and $S\subseteq M$ a submanifold, and whose morphisms $f: (M,S)\to (M',S')$ are those morphisms $f:M\to M'$ in $\mathsf{Man}$ such that $f(S)\subseteq S'$.} to the category of Abelian groups $\mathsf{Ab}$. Furthermore,
\begin{eqnarray}
&\curv: \hat{H}^k(\arg;\Z) \Rightarrow \Omega^k_\Z(\arg), \qquad & \cha :\hat{H}^k(\arg;\Z) \Rightarrow H^k(\arg;\Z),\\
&\iota: \dfrac{\Omega^{k-1}(\arg)}{\Omega^{k-1}(\arg)}\Rightarrow\hat{H}^k(\arg;\Z), \qquad &\kappa: H^{k-1}(\arg;\T)\Rightarrow \hat{H}^k(\arg;\Z),
\end{eqnarray}
are natural transformations between functors from $\mathsf{Pair}$ to $\mathsf{Ab}$.\\

There exists a natural homomorphism, of particular relevance to our work, mapping relative differential characters into absolute ones:
\begin{equation}\label{map:2:I_relative_diff_char}
I: \hat{H}^k(M,S;\Z)\to \hat{H}^k(M;\Z).
\end{equation}
Consider the quotient map $$\pi:C_\ast(M;\Z)\to C_\ast(M,S;\Z)=\dfrac{C_\ast(M;\Z)}{C_\ast(S;\Z)}.$$ 
Since $\pi$ preserves the boundary morphism $\partial$, it maps the group $Z_\ast(M;\Z)$ of cycles to the group $Z_\ast(M,S;\Z)$ of relative cycles. Then, for every $h\in \hat{H}^k(M,S;\Z)$, we set $I h := h\circ \pi\in \hat{H}^k(M;\Z)$. It is straightforward to check that the following diagram is commutative:
\begin{equation}
\xymatrix{
Z_{k-1}(M;\Z) \ar[r]^-\pi \ar[d]_{f_\ast} & Z_{k-1}(M,S;Z) \ar[d]^-{f_\ast}\\
Z_{k-1}(M';\Z) \ar[r]_-\pi &Z_{k-1}(M',S';\Z)\\
}
\end{equation}
for every morphism $f:(M,S)\to(M',S')$ in $\mathsf{Pair}$. This entails the naturality of $I$:
\begin{equation}
\xymatrix{
\hat{H}^k(M',S';\Z) \ar[r]^-I \ar[d]_-{f^\ast} & \hat{H}^k(M';\Z) \ar[d]^-{f^\ast}\\
\hat{H}^k(M,S;\Z) \ar[r]_-I \ & \hat{H}^k(M;\Z) \\
}
\end{equation}
Furthermore, the diagrams:
\begin{equation}
\xymatrix{
\hat{H}^k(M,S;\Z) \ar[r]^-I\ar[d]_-\curv &\hat{H}^k(M;\Z)\ar[d]^-\curv & \hat{H}^k(M,S;\Z) \ar[r]^-I\ar[d]_-\cha &\hat{H}^k(M;\Z)\ar[d]^-\cha\\
\Omega^k_\Z(M,S) \ar[r] & \Omega^k_\Z(M) & H^k(M,S;\Z)\ar[r] &H^k(M;\Z)\\ 
H^{k-1}(M,S;\T) \ar[r] \ar[d]_-\kappa& H^{k-1}(M;\T) \ar[d]^-\kappa & \dfrac{\Omega^{k-1}(M,S)}{\Omega_\Z^{k-1}(M,S)} \ar[r]\ar[d]_-\iota & \dfrac{\Omega^{k-1}(M)}{\Omega^{k-1}_\Z (M)} \ar[d]^-\iota\\
\hat{H}^k(M,S;\Z) \ar[r]_-I &\hat{H}^k(M;\Z) & \hat{H}^k(M,S;\Z) \ar[r]_-I &\hat{H}^k(M;\Z)\\
}
\end{equation}
are commutative for every $(M,S)$ in $\mathsf{Pair}$. 

In general, the homomorphism I is not injective. In fact, the commutativity of the diagram
\begin{equation*}
\xymatrix{
0 \ar[r] & H^{k-1}(M,S;\T) \ar[r]^-\kappa  \ar[d]& \hat{H}^k(M,S;\Z) \ar[d]^-I \ar[r]^-\curv & \Omega^k_\Z(M,S)\ar[r]\ar[d] & 0\\
0 \ar[r] & H^{k-1}(M;\T) \ar[r]_-\kappa & \hat{H}^k(M;\Z) \ar[r]_-\curv & \Omega^k_\Z(M)\ar[r] & 0
}
\end{equation*}
along with the exactness of the rows, implies that $I$ is a monomorphism if and only if $H^{k-1}(M,S;\T)\to H^{k-1}(M;\T)$ is such, the right vertical arrow being nothing but an inclusion, therefore an injective map. Picking $M=\R^m$ and $S=M\setminus (R^{m-k+1}\times \mathbb{B}^{k-1})$, where $\mathbb{B}^n$ is the n-dimensional closed unit Euclidean ball, we have $H^{k-1}(M,S;\T)\simeq \T$ and $H^{k-1}(M;\T)\simeq 0$, thus exhibiting a counterexample.

\subsection{Differential cohomology with compact support}
In this section, we will discuss differential cohomology with compact support. At variance with above, much more attention will be devoted to details, in view of its prominence in the construction of the algebraic states.

\begin{defn}[Directed set]
We call \emph{directed set} a non-empty set I endowed with a reflexive and transitive binary relation $\preceq$, such that for all $\alpha, \beta \in I$ there exists $\gamma \in I$ for which $\alpha\preceq \gamma$ and $\beta \preceq \gamma$. 
\end{defn}

Let $\mathsf{DSet}$ be the category whose objects are directed sets and whose morphisms are functions preserving the pre-order relation. Introduce the functor:
\begin{equation}\label{eq:2:functor_K}
\mathcal{K}: \mathsf{Man} \to \mathsf{DSet}
\end{equation}
between the category of smooth manifolds and the category of directed sets. $\mathcal{K}$ assigns to each object $M$ in $\mathsf{Man}$ the directed set $\mathcal{K}_M:=\{K\subseteq M\,\,\vline\,\, K\, \text{is compact}\}$, where $\preceq$ is given by the inclusion, and to each morphism $f:M\to M'$ in $\mathsf{Man}$ the morphism $\mathcal{K}_f:\mathcal{K}_M \to \mathcal{K}_{M'}$, $K\mapsto f(K)$ in $\mathsf{DSet}$.

Interpreting the directed set $\mathcal{K}_M$ as a category, construct the contravariant functor:
\begin{gather*}
(M,M\setminus \arg): \mathcal{K}_M \to \mathsf{Pair}\\
K\mapsto (M,M\setminus K).
\end{gather*}
Composing such a functor with the relative differential cohomology functor yields the covariant functor:
\begin{equation}\label{eq:2:functor_diff_char_comp_supp}
\hat{H}^k(M,M\setminus \arg;\Z)= \hat{H}^k(\arg;\Z) \circ (M,M\setminus\arg): \mathcal{K}_M \rightarrow \mathsf{Ab}.
\end{equation}

\begin{defn}\label{def:2:diff_char_comp_supp}
The Abelian group of \emph{differential characters with compact support} of degree $k$ over $M$ is defined by the colimit:
\begin{equation}
\hat{H}^k_c(M;\Z):= \mathrm{colim}\left( \hat{H}^k(M,M\setminus \arg;\Z):\mathcal{K}_M \rightarrow \mathsf{Ab} \right)
\end{equation}
of the functor \eqref{eq:2:functor_diff_char_comp_supp} over $\mathcal{K}_M$.
\end{defn}

The procedure adopted in Definition \ref{def:2:diff_char_comp_supp} represents an intuitively clear approach. There is, nonetheless, an alternative, yet equivalent, way to present differential cohomology with compact support, which has the advantage of involving only properly embedded submanifolds.

Introduce the directed set:
\begin{equation*}
\mathcal{O}_M ^c:=\{O\subseteq M\,\,\vline\,\, \text{O is open}, \overline{O}\in \mathcal{K}_M, \partial \overline{O}\text{ is smooth}\},
\end{equation*}
where, again, the preorder relation $\preceq$ is given by the inclusion. By construction, the complement $M\setminus O$ is a properly embedded submanifold for every $O\in \mathcal{O}_M^c$. Composing the functor 
\begin{gather*}
(M,M\setminus \arg): \mathcal{O}_M^c\rightarrow \mathsf{Pair}\\
O\mapsto (M,M\setminus O)
\end{gather*}
with the relative cohomology functor yields:
\begin{equation}\label{eq:2:functor_comp_diff_coho_v2}
\hat{H}^k(M,M\setminus\arg;\Z):= \hat{H}^k(\arg;\Z)\circ (M,M\setminus \arg): \mathcal{O}_M^c\rightarrow \mathsf{Ab}.
\end{equation}

\begin{prop}
The colimit of the functor \eqref{eq:2:functor_comp_diff_coho_v2} is isomorphic to the Abelian group of differential characters with compact support $\hat{H}^k_c(M;\Z)$:
\begin{equation}\label{eq:2:colim_iso_diff_char_cs}
\hat{H} ^k _c(M;\Z)\simeq \mathrm{colim} \left( \hat{H}^k (M,M\setminus\arg;\Z):\mathcal{O}_M^c\rightarrow \mathsf{Ab}\right).
\end{equation}
\end{prop}

\begin{proof}
Define:
$$
\mathcal{S}_M:=\mathcal{K}_M\cup \mathcal{O}^c_M.
$$
This allows the construction of a functor:
\begin{equation*}
\hat{H}^k(M,M\setminus \arg;\Z): \mathcal{S}_M \rightarrow \mathsf{Ab}. 
\end{equation*}
Since $\mathcal{O}_M^c$ and $\mathcal{K}_M$ are cofinal in $\mathcal{S}_M$, the following chain of isomorphisms holds:
\begin{gather*}
\text{colim} ( \hat{H}^k (M,M\setminus\arg;\Z):\mathcal{O}_M^c\rightarrow \mathsf{Ab})
 \simeq \text{colim} ( \hat{H}^k (M,M\setminus\arg;\Z):\mathcal{S}_M \rightarrow \mathsf{Ab})\\
\simeq \text{colim} ( \hat{H}^k (M,M\setminus\arg;\Z):\mathcal{K}_M\rightarrow \mathsf{Ab})=\hat{H}^k_c(M;\Z).
\end{gather*}
The proposition is thus proved.
\end{proof}

In order to obtain a commutative diagram of short exact sequences for compactly supported differential cohomology, define:
$$
\Omega^k_{c,\Z}(M):=\text{colim} \left(\Omega^k_\Z(M,M\setminus \arg):\mathcal{K}_M\rightarrow \mathsf{Ab}\right).
$$
The properties of \text{colim} for AB5 categories yield:
$$
\text{colim}\left( \dfrac{\Omega^k(M,M\setminus \arg)}{\Omega^k_\Z(M,M\setminus\arg)}:\mathcal{K}_M \rightarrow \mathsf{Ab}\right)=\dfrac{\Omega^k_c(M)}{\Omega^k_{c,\Z}(M)}.
$$

\begin{thm}
Let $M$ be a smooth manifold. Then the diagram:
\begin{equation}\label{diag:2:diff_char_compact_support}
\xymatrix{
  & 0 \ar[d] & 0\ar[d] & 0\ar[d] &\\
0\ar[r] & \dfrac{H^{k-1}_c(M;\R)}{H^{k-1}_{c,\free}(M;\Z)} \ar[d] \ar[r]& \dfrac{\Omega^{k-1}_c(M)}{\Omega^{k-1}_{c,\Z}(M)} \ar[d]^-\iota\ar[r]^-d & d\Omega^{k-1}_c(M)\ar[d] \ar[r] & 0\\
0\ar[r] & H^{k-1}_c(M;\T) \ar[d]\ar[r]^-\kappa& \hat{H}_c^k(M;\Z) \ar[d]^-{\mathrm{char}}\ar[r]^-{\mathrm{curv}}& \Omega^k_{c,\Z}(M)\ar[d] \ar[r]& 0\\
0\ar[r] & H^k_{c,\tor}(M;\Z)\ar[d]\ar[r] & H^{k}_c(M;\Z) \ar[d]\ar[r]& H^k_{c,\free}(M;\Z) \ar[d]\ar[r]&0\\
  & 0 & 0 & 0 &\\
}
\end{equation}
is commutative and has exact rows and columns.
\end{thm}
\begin{proof}
The assertion follows immediately from diagram \eqref{diag:2:relative_diff_char}, from Proposition \ref{prop:2:diag_rel_diff_char} and from the fact that colim is an exact functor for diagrams of Abelian groups over directed sets.
\end{proof}

\begin{rem}
Observe that while invoking Proposition \ref{prop:2:diag_rel_diff_char} we have implicitly used the isomorphism \eqref{eq:2:colim_iso_diff_char_cs} which allows us to deal with properly embedded submanifolds only.
\end{rem}

Differential cohomology with compact support is a functor:
\begin{prop}\label{prop:2:diff_coho_comp_is_functor}
Let $\mathsf{Man}_{m,\hookrightarrow}$ be the category of smooth m-dimensional manifolds whose morphisms are open embeddings. Differential cohomology with compact support is a covariant functor from $\mathsf{Man}_{m,\hookrightarrow}$ to the category of Abelian groups $\mathsf{Ab}$:
\begin{equation}
\hat{H}^k_c(\arg;\Z): \mathsf{Man}_{m,\hookrightarrow} \rightarrow \mathsf{Ab}
\end{equation}
\end{prop}
\begin{proof}
Let $f:M\to M'$ be a morphism in $\mathsf{Man}_{m,\hookrightarrow}$, i.e. an open embedding. Then $f$ can be written as the composition of an embedding $i:f(M)\to M'$ and a diffeomorphism $g:M\to f(M)$: $f=i\circ g$. For every $K'\in\mathcal{K}_{f(M)}$ it holds $f(M)\setminus K'\subseteq M'\setminus K'$; therefore, $i$ induces a natural transformation
\begin{equation*}
i^\ast: \hat{H}^k(M',M'\setminus \arg;\Z)\circ \mathcal{K}_i \Rightarrow \hat{H}^k(f(M),f(M)\setminus \arg;\Z)
\end{equation*}
between functors from $\mathcal{K}_{f(M)}$ to $\mathsf{Ab}$. The excision theorem for relative differential cohomology \cite[Theorem 3.8]{BMath} guarantees that $i^\ast$ is, in addition, an isomorphism. Denote its inverse by:
\begin{equation*}
(i^\ast)^{-1}: \hat{H}^k(f(M),f(M)\setminus \arg;\Z) \Rightarrow \hat{H}^k(M',M'\setminus \arg;\Z)\circ \mathcal{K}_i.
\end{equation*}
As $g$ is a diffeomorphism, we have immediately that it induces a natural isomorphism:
\begin{equation*}
g^\ast: \hat{H}^k(f(M),f(M)\setminus \arg;\Z) \circ \mathcal{K}_g \Rightarrow \hat{H}^k(M,M\setminus \arg;\Z)
\end{equation*}
whose inverse is:
\begin{equation*}
(g^\ast)^{-1}: \hat{H}^k(M,M\setminus \arg;\Z) \Rightarrow \hat{H}^k(f(M),f(M)\setminus \arg;\Z) \circ \mathcal{K}_g.
\end{equation*}
What we have obtained is that
\begin{equation*}
f^\ast=g^\ast\circ (i^\ast \mathcal{K}_g): \hat{H}^k(M',M'\setminus \arg;\Z)\circ \mathcal{K}_f \Rightarrow \hat{H}^k(M,M\setminus \arg;\Z)
\end{equation*}
is a natural isomorphism. Denoting by $(f^\ast)^{-1}=((i^\ast)^{-1}\mathcal{K}_g)\circ (g^\ast)^{-1}$ its inverse, by resorting to the universal property of the colimit we get a canonical homomorphism of Abelian groups:
\begin{equation*}
f_\ast:\hat{H}^k_c(M;\Z) \to \hat{H}^k_c(M';\Z).
\end{equation*}
The assignment $f\mapsto \hat{H}^k_c(f;\Z):=f_\ast$ is well-behaved with respect to identity and composition. The proposition is proved.
\end{proof}

Naturality of curvature, characteristic class, inclusion of flat classes and topological trivialization is preserved when passing to the colimit; they can therefore be promoted to natural transformations between functors from $\mathsf{Man}_{m,\hookrightarrow}$ to $\mathsf{Ab}$.

\begin{rem}
Contrary to ordinary cohomology, cohomology with compact support is \emph{not} an invariant of homotopy type. See, e.g., \cite[Example 3.34]{HAT02}. This fact will have some drawbacks when using differential cohomology with compact support to construct field theories.
\end{rem}

It can be proved \cite{BMath} that differential cohomology with compact support $\hat{H}^\ast_c(M;\Z)$ can be given a module structure over the ring $\hat{H}^\ast(M;\Z)$ of differential characters:
\begin{gather}
\arg : \hat{H}^k(M;\Z)\cdot \hat{H}^l_c(M;\Z)\to \hat{H}^{k+l}_c(M;\Z) \nonumber\\
(h,h')\mapsto h\cdot h'.
\end{gather}
The module structure is natural and compatible with the homomorphisms $\curv$, $\cha$, $\iota$, $\kappa$:
\begin{subequations}\label{eq:2:module_structure}
\begin{align}
& \curv(h\cdot h')=\curv h\wedge \curv h', &&\cha(h\cdot h')=\cha h \smile \cha h'\\
& \iota[A] \cdot h'= \iota [A\wedge \curv h'], && h\cdot \iota[A']= (-1)^k \iota [\curv h \wedge A']\\
& \kappa u \cdot h'= \kappa( u\smile \cha h'), && h\cdot \kappa u'=(-1)^k \kappa(\cha h \smile u'),
\end{align}
\end{subequations}
for every $h\in\hat{H}^k(M;\Z)$, $h' \in \hat{H}^l_c(M;\Z)$, $[A]\in\frac{\Omega^{k-1}(M)}{\Omega^{k-1}_\Z(M)}$, $[A']\in \frac{\Omega_c^{l-1}(M)}{\Omega_{c,\Z}^{l-1}(M)}$, $u\in H^{k-1}(M;\T)$, $u'\in H^{l-1}_c(M;\T)$.\\

There is a natural way to map differential characters with compact support into differential characters. As a matter of fact, applying the colimit prescription to \eqref{map:2:I_relative_diff_char} results in a natural homomorphism, which we denote again by $I$ with a slight abuse of notation:
\begin{equation}\label{eq:2:I}
I:\hat{H}^k_c(M;\Z)\to \hat{H}^k(M;\Z).
\end{equation}
As above, it is in general neither injective nor surjective.

\section{Duality and Pairings}\label{section_two_duality_pairing}
The next step consists of showing that, under suitable assumptions, differential characters and differential characters with compact support are one the dual of the other. In doing this, following \cite{BMath}, we will obtain as by-products several pairings between spaces of interest, which will represent one of the technical tools needed to build our theory.\\

Henceforth, we will restrict the analysis to the category $\mathsf{oMan}_{m,\hookrightarrow}$ of oriented, connected $m$-dimensional smooth manifolds with orientation preserving open embeddings as morphisms. Such assumption yields a wide range of useful results; for instance, a volume form is available.

\begin{defn}[Pontryagin dual]
Let $G$ be an Abelian group. The \emph{Pontryagin dual} of $G$, or character group of $G$, is the group:
$$
G^\star:=\Hom (G,\T).
$$
\end{defn}

For our purposes, elements in the Pontryagin duals of the Abelian groups appearing in diagrams \eqref{diag:2:diff_char} and \eqref{diag:2:diff_char_compact_support} are too singular. Therefore, we need to refine our concept of duality, in order to rule out ill-behaved elements.

\begin{defn}[Smooth Pontryagin dual: differential forms]
We call \emph{smooth Pontryagin dual} of $\Omega^k_c(M)$ the Abelian group $\Omega_c^k(M)^\star _\infty$ defined by:
{\footnotesize
\begin{equation*}
\Omega_c^k(M)^\star_\infty:=\left\{ \varphi\in \Omega^k_c(M)^\star \,\, \vline \,\, \exists \omega \in \Omega^{m-k}(M) \,\,\text{s.t.}\,\, \varphi(\psi)=\int_M \omega \wedge \psi\,\,\mod \Z\,\, \forall\psi\in \Omega^k_c(M)  \right\}.
\end{equation*} 
}
Analogously, we call \emph{smooth Pontryagin dual} of $\Omega^k(M)$ the Abelian group $\Omega^k(M)^\star_\infty$ defined by:
{\footnotesize
\begin{equation*}
\Omega^k(M)^\star_\infty:=\left\{ \varphi\in \Omega^k(M)^\star \,\, \vline \,\, \exists \psi \in \Omega^{m-k}_c(M) \,\,\text{s.t.}\,\, \varphi(\omega)=\int_M \omega \wedge \psi\,\,\mod \Z\,\, \forall\omega\in \Omega^k(M)  \right\}.
\end{equation*} 
}
\end{defn}

We introduce the weakly non-degenerate bilinear pairing:
\begin{eqnarray}\label{pairing:2:differential_forms}
&\pa{\arg}{\arg}_\Omega: \Omega^k(M)\times \Omega^{m-k}_c(M)\to \T \nonumber&\\
&(\omega, \psi)\mapsto \pa{\omega}{\psi}_\Omega=(-1)^k \int_M \omega \wedge \psi\,\,\mod\Z.&
\end{eqnarray}
It induces the isomorphisms:
\begin{equation*}
\begin{split}
&\Omega^k(M) \overset{\simeq}{\longrightarrow} \Omega^{m-k}_c(M)^\star_\infty, \qquad \omega \mapsto \pa{\omega}{\arg}_\Omega\\
&\Omega^{m-k}_c(M) \overset{\simeq}{\longrightarrow} \Omega^{k}(M)^\star_\infty, \qquad \psi \mapsto \pa{\arg}{\psi}_\Omega.
\end{split}
\end{equation*}
In fact, surjectivity holds by definition, whilst injectivity ensues from the weak non-degeneracy of the pairing.\\

\cite[Lemma A.1]{BMath} and \cite[Lemma A.2]{BMath} allow to characterize in an alternative way the Abelian group of differential forms with integer period and the Abelian group of differential forms with integer period and compact support:
\begin{equation*}
\begin{split}
&\Omega^k_\Z(M)=\left\{ \omega\in\Omega^k(M)\,\,\vline\,\, \int_M \omega\wedge \Omega^{m-k}_{c,\Z}(M)\subseteq \Z \right\},\\
&\Omega^{m-k}_{c,\Z}(M)=\left\{ \psi\in\Omega_c^{m-k}(M) \,\,\vline \,\, \int_M \Omega_\Z^k(M) \wedge\psi \subseteq \Z \right\}.
\end{split}
\end{equation*}
The latter requires the additional assumption of having a manifold $M$ of finite type.
In view of this, \eqref{pairing:2:differential_forms} induces two additional pairings:
\begin{gather}
\pa{\arg}{\arg}_\Omega: \Omega^k_\Z(M)\times \dfrac{\Omega_c^{m-k}(M)}{\Omega_{c,\Z}^{m-k}(M)}\to \T \nonumber\\
(\omega,[\psi])\mapsto \pa{\omega}{\psi}_\Omega=(-1)^k \int_M \omega \wedge \psi\,\,\mod\Z,
\end{gather}
and
\begin{gather}
\pa{\arg}{\arg}_\Omega: \dfrac{\Omega^k(M)}{\Omega^k_\Z(M)}\times \Omega_{c,\Z}^{m-k}(M)\rightarrow \T \nonumber\\
([\omega],\psi)\mapsto \pa{\omega}{\psi}_\Omega=(-1)^k \int_M \omega\wedge\psi\,\,\mod\Z,
\end{gather}
from which the following isomorphisms stem:
\begin{subequations}\label{pairing:2:Z}
\begin{align}
&\Omega^k_\Z(M) \overset{\simeq}{\longrightarrow} \left( \dfrac{\Omega_c^{m-k}(M)}{\Omega_{c,\Z}^{m-k}(M)}\right)^\star_\infty, \quad &&\omega \mapsto \pa{\omega}{\arg}_\Omega,\label{pairing:2:Z1}\\
&\dfrac{\Omega_c^{m-k}(M)}{\Omega_{c,\Z}^{m-k}(M)} \overset{\simeq}{\longrightarrow} \Omega^k_\Z(M)^\star_\infty, \quad&&[\psi]\mapsto \pa{\arg}{\psi}_\Omega,\label{pairing:2:Z2}\\
&\dfrac{\Omega^k(M)}{\Omega^k_\Z(M)} \overset{\simeq}{\longrightarrow} \Omega_{c,\Z}^{m-k}(M)^\star_\infty, \quad &&[\omega]\mapsto \pa{\omega}{\arg}_\Omega,\label{pairing:2:Z3}\\
&\Omega_{c,\Z}^{m-k}(M) \overset{\simeq}{\longrightarrow} \left( \dfrac{\Omega^k(M)}{\Omega^k_\Z(M)} \right)^\star_\infty, \quad &&\psi \mapsto \pa{\arg}{\psi}_\Omega.\label{pairing:2:Z4}
\end{align}
\end{subequations}
\eqref{pairing:2:Z2} and \eqref{pairing:2:Z4} require $M$ to be of finite type.\\

Let us come eventually to differential characters.
\begin{defn}[Smooth Pontryagin dual: differential characters] We define as the \emph{smooth Pontryagin dual} of $\hat{H}^k(M;\Z)$ the Abelian group:
\begin{equation*}
\hat{H}^k(M;\Z)^\star_\infty:=(\iota^\star)^{-1} \left( \dfrac{\Omega^{k-1}(M)}{\Omega^{k-1}_\Z(M)}\right)^\star_\infty,
\end{equation*}
where $\iota^\star:=\Hom (\iota,\T): \hat{H}^k(M;\Z)^\star\to \left(\frac{\Omega^{k-1}(M)}{\Omega^{k-1}_\Z(M)}\right)^\star$ is the Pontryagin dual of the topological trivialization $\iota$.

We define as the \emph{smooth Pontryagin dual} of $\hat{H}^k_c(M;\Z)$ the Abelian group:
\begin{equation*}
\hat{H}^k_c(M;\Z)^\star_\infty:=(\iota^\star)^{-1}\left( \dfrac{\Omega^{k-1}_c(M)}{\Omega^{k-1}_{c,\Z}(M)} \right)^\star_\infty,
\end{equation*}
where $\iota^\star:=\Hom(\iota;\T):\hat{H}^k_c(M;\Z)^\star\to \left(\frac{\Omega^{k-1}_c(M)}{\Omega^{k-1}_{c,\Z}(M)}\right)^\star$ is the Pontryagin dual of the topological trivialization $\iota$ (in the compact support case).
\end{defn}

Smooth Pontryagin duals of differential characters possess a functorial behaviour:

\begin{prop}
The smooth Pontryagin dual of differential characters and the smooth Pontryagin dual of differential characters with compact support define functors:
\begin{gather*}
\hat{H}^k(\arg;\Z)^\star_\infty: \mathsf{oMan}_{m,\hookrightarrow} \to \mathsf{Ab},\\
\hat{H}^k_c(\arg;\Z)^\star_\infty: \mathsf{oMan}_{m,\hookrightarrow} \to \mathsf{Ab}.
\end{gather*}
\end{prop}

Consider the following pairing:
\begin{eqnarray}\label{pairing:2:duality_diff_char}
&\pa{\arg}{\arg}_c: \hat{H}^k(M;\Z)\times \hat{H}^{m-k}_c(M;\Z) \to \T \nonumber&\\
&(h,h')\mapsto \pa{h}{h'}_c=\left( h\cdot h'\right) \mu.&
\end{eqnarray}
The right-hand side must be interpreted as follows: Pick a representative of $h'\in \hat{H}^{m-k}_c(M;\Z)$, still denoted by the same symbol, $h'\in\hat{H}^{m-k}(M,M\setminus K;\Z)$ for some compact $K\subseteq M$; $(h\cdot h')$ is then an element in $\hat{H}^{m}(M,M\setminus K;\Z)$. Evaluate it on the unique relative homology class $\mu\in H_m(M,M\setminus K;\Z)$ which restricts to the orientation of $M$ for each point of $K$ \cite[Lemma 3.27]{HAT02}.
By partial evaluation, it induces two homomorphisms:
\begin{subequations}\label{eq:2:Pontryagin_duality_diff_char}
\begin{align}
&\hat{H}^k(M;\Z) \to \hat{H}^{m-k}_c(M;\Z)^\star _\infty, \qquad h	\mapsto \pa{h}{\arg}_c,\\
&\hat{H}^{m-k}_c(M;\Z) \to \hat{H}^k(M;\Z)^\star _\infty, \qquad h'\mapsto \pa{\arg}{h'}_c.
\end{align}
\end{subequations}

Similarly, introduce a pairing for the singular cohomology:
\begin{eqnarray}\label{pairing:2:duality_cohomology}
&\pa{\arg}{\arg}_H: H^k(M;\T) \times H_c^{m-k}(M;\Z) \to \T \nonumber&\\
& (u,u')\mapsto \pa{u}{u'}_H=(u\smile u')\mu.&
\end{eqnarray}
Again, we choose a representative of $u'\in H_c^{m-k}(M;\mathbb{Z})$, $u'\in H^{m-k}(M,M\setminus K;\Z)$ for some compact $K\subseteq M$ and we evaluate $(u\smile u')\in H^m(M,M\setminus K;\T)$ on the unique relative homology class $\mu\in H_m(M,M\setminus K;\Z)$ which restricts to the orientation of $M$ for every point of $K$.

By partial evaluation, it induces the homomorphisms:
\begin{subequations}\label{eq:2:Pontryagin_duality_coho}
\begin{align}
 & H^k(M;\T) \to H_c^{m-k}(M;\Z)^\star, \qquad u \mapsto \pa{u}{\arg}_H \label{eq:2:Pontryagin_duality_coho_A}\\
 & H_c^{m-k}(M;\Z) \to H^k(M;\T)^\star, \qquad u'\mapsto \pa{\arg}{u'}_H \label{eq:2:Pontryagin_duality_coho_B}.
\end{align}
\end{subequations}
By \cite[Lemma 5.3]{BMath}, \eqref{eq:2:Pontryagin_duality_coho_A} is an isomorphim, while \eqref{eq:2:Pontryagin_duality_coho_B} is an isomorphism if $M$ is of finite type.

\begin{thm}\label{thm:2:iso_exact_rows}
Let $M$ be an object in $\mathsf{oMan}_{m,\hookrightarrow}$. Then the diagram:
\begin{equation}
\xymatrix{
0 \ar[r] & H^{m-k-1}(M;\T) \ar[r]^-\kappa \ar[d]^-\simeq & \hat{H}^{m-k}(M;\Z) \ar[r]^-{\mathrm{char}} \ar[d]^-\simeq &\Omega_\Z^{m-k}(M)\ar[r] \ar[d]^\simeq & 0\\
0 \ar[r] & H^k_c(M;\Z)^\star \ar[r]^-{\mathrm{char}^\star} & \hat{H}^{k}_c(M;\Z)^\star_\infty\ar[r]^-{\iota^\star} & \left( \dfrac{\Omega_c^{k-1}(M)}{\Omega_{c,\Z}^{k-1}(M)} \right) ^\star _\infty \ar[r] & 0
}
\end{equation}
is commutative with exact rows. Furthermore, the vertical arrows, given by \eqref{pairing:2:Z}, \eqref{eq:2:Pontryagin_duality_coho} and \eqref{eq:2:Pontryagin_duality_diff_char}, are natural isomorphisms.

Let $M$ be an object in $\mathsf{oMan}_{m,\hookrightarrow}$ of finite type. Then the diagram:
\begin{equation}
\xymatrix{
0 \ar[r] & \dfrac{\Omega_c^{m-k-1}(M)}{\Omega_{c,\Z}^{m-k-1}(M)} \ar[r]^-{\iota} \ar[d]^-\simeq & \hat{H}^{m-k}_c(M;\Z) \ar[r]^-{\mathrm{char}} \ar[d]^-\simeq & H_c^{m-k}(M;\Z) \ar[r]\ar[d]^-\simeq & 0\\
0 \ar[r] & \Omega^{k}_\Z(M)^\star_\infty \ar[r]^-{\mathrm{curv}^\star} & \hat{H}^k(M;\Z)^\star_\infty \ar[r]^-{\kappa^\star} & H^{k-1}(M;\T)^\star \ar[r] & 0
}
\end{equation}
is commutative with exact rows. Furthermore, the vertical arrows, given by \eqref{pairing:2:Z}, \eqref{eq:2:Pontryagin_duality_coho} and \eqref{eq:2:Pontryagin_duality_diff_char}, are natural isomorphisms.
\end{thm}

The above theorem establishes, in particular, that differential cohomology and differential cohomology with compact support are one the smooth Pontryagin dual of the other:

\begin{cor}\label{cor:2:non_deg_pair_c}
Let $M$ be an oriented, connected m-dimensional smooth manifold of finite type. Then, the pairing \eqref{pairing:2:duality_diff_char} is weakly non degenerate.
\end{cor}

To conclude the present section, let us show how it is possible to introduce, thanks to the above results, a $\T$-valued pairing on differential characters with compact support. Define:
\begin{eqnarray}\label{pairing:2:comp_supp_char}
&\pa{I\arg}{\arg}_c: \hat{H}_c^k(M;\Z)\times \hat{H}^{m-k}_c(M;\Z) \to \T \nonumber&\\
&(h,h')\mapsto \pa{Ih}{h'}_c,&
\end{eqnarray}
where $I: \hat{H}^k_c(M;\Z)\to \hat{H}^k(M;\Z)$ is the homomorphism defined in \eqref{eq:2:I}. Due to $I$ being neither injective nor surjective in general, such a pairing may be degenerate. When $M$ is of finite type, the degeneracy of the pairing coincides with the kernel of $I$. We will return to this point extensively later.


\section{Quantum Field Theory}\label{section:2.3_quantum_field_theory}

Differential characters provide a powerful and versatile mathematical tool. In the present section, in the framework of locally covariant quantum field theory, we will show how to construct a model for generalized Abelian gauge theory out of differential cohomology.\\

\subsection{Semi-classical configuration space}
From the perspective of physics, the most remarkable feature of the model will be the natural simultaneous discretization of both the electric and the magnetic charges. For the sake of clarity, we normalize the electric and magnetic fluxes in such a way that they are quantized in the same lattice $\Z\subseteq \R$. Furthermore, following \cite{BPhys}, we adopt the adjective ``semi-classical'' to mean that the Dirac charge quantization, which usually arises as a quantum effect, takes place at an earlier stage.\\

Introduce the category \textsf{Loc}$_m$ of $m$-dimensional time-oriented, globally hyperbolic, Lorentzian manifolds whose morphisms are causal embeddings, i.e.\ orientation and time-orientation preserving isometric embeddings $f:M\to M'$ with open and causally compatible\footnote{The image of an orientation and time-orientation preserving isometric embedding $f: M \to M^\prime$ is \textit{causally compatible} if $J^\pm_{M'}(f(p)) \cap f(M)=f\left( J^\pm_M(p) \right)$ for every $p\in M$.} image. Henceforth, every time we pick an object $M$ in $\textsf{Loc}_m$, it will be interpreted as the \emph{spacetime} in hand. Furthermore, for any graded Abelian group $G^\ast=\oplus _{k\in\Z}G^k$ we will adopt the notation:
$$
G^{m,n}:=G^m\times G^n.
$$

The discretization is obtained selecting a distinguished subspace in the product of differential cohomology groups by imposing suitable constraints:

\begin{defn}\label{def:2:semi_class_conf_space}
We call \emph{semi-classical configuration space} $\C^k(M;\Z)$ the Abelian group:
\begin{equation*}
\C^k(M;\Z):=\left\{  (h,\tilde{h})\in \hat{H}^k(M;\Z)\times \hat{H}^{m-k}(M;\Z)\,\,\vline\,\, \curv h=\ast\curv \tilde{h}  \right\},
\end{equation*} 
\end{defn}
where $\ast$ denotes the Hodge dual for differential forms induced by the orientation and the metric on $M$.

As a consequence, the restriction must be consistently extended to the other elements in diagram \eqref{diag:2:diff_char}.

\begin{defn}
We call \emph{semi-classical topologically trivial fields} $\mathfrak{T}^k(M;\Z)$ the Abelian group:
\begin{equation}
\mathfrak{T}^k(M;\Z):=\left\{ ([A],[\tilde{A}])\in \dfrac{\Omega^{k-1,m-k-1}(M)}{\Omega^{k-1,m-k-1}_\Z(M)}\,\,\vline\,\, dA=\ast d\tilde{A} \right\}.
\end{equation}
\end{defn}

This allows to obtain a new commutative diagram of Abelian groups, for which the central object is the semi-classical configuration space $\C^k(M;\Z)$:
\begin{prop}
Let $M$ be an object in $\mathsf{Loc}_m$. The diagram
\emph{
{\footnotesize
\begin{equation}\label{diag:2:scconfig}
\xymatrix@C-=0.6cm{
& 0 \ar[d] & 0 \ar[d] & 0 \ar[d] &\\
0\ar[r] & \dfrac{H^{k-1,m-k-1}(M;\mathbb{R})}{H^{k-1,m-k-1}_{\text{free}}(M;\mathbb{Z})} \ar[r]^-{\tilde{\kappa}\times \tilde{\kappa}} \ar[d] & \mathfrak{T}^k(M;\mathbb{Z}) \ar[r]^-{d_1} \ar[d]^-{\iota \times \iota} & d\Omega^{k-1}\cap *d\Omega^{m-k-1}(M) \ar[r]\ar[d]^-{\subseteq}& 0\\
0\ar[r]& H^{k-1,m-k-1}(M;\mathbb{T}) \ar[r]^-{\kappa \times \kappa} \ar[d] & \mathfrak{C}^k(M;\mathbb{Z})  \ar[r]^-{\text{curv}_1} \ar[d]^-{\text{char} \times \text{char}} & \Omega^{k}_{\mathbb{Z}}\cap *\Omega^{m-k}_{\mathbb{Z}}(M) \ar[r] \ar[d]^-{([\cdot],[*^{-1}\cdot])}& 0\\
0 \ar[r] & H^{k,m-k}_{\text{tor}}(M;\mathbb{Z}) \ar[r] \ar[d] & H^{k,m-k}(M;\mathbb{Z}) \ar[r]\ar[d] & H^{k,m-k}_{\free}(M;\mathbb{Z}) \ar[r] \ar[d] &0\\
& 0 & 0 & 0 & \\
}
\end{equation}
}}
is commutative and it has exact rows and columns. The homomorphisms $d_1$ and $\curv_1$ are defined respectively as:
\emph{
\begin{gather}
\curv_1:\C^k(M;\Z) \to \Omega_\Z^k\cap \ast \Omega_\Z ^{m-k}(M)\nonumber\\
(h,\tilde{h})\mapsto \curv h =\ast \curv \tilde{h}
\end{gather}}
and
\emph{\begin{gather}
d_1:\mathfrak{T}^k(M;\Z) \to d\Omega^{k-1}\cap \ast d\Omega^{m-k-1}(M)\nonumber\\
([A],[\tilde{A}] )\mapsto dA=\ast d\tilde{A}.
\end{gather}}
\end{prop}
\begin{proof}
Commutativity follows at once from diagram \eqref{diag:2:diff_char}. For the same reason, the left column and the bottom row are exact, being the Cartesian product of exact sequences, and $\tilde{\kappa}\times\tilde{\kappa}$, $\kappa\times \kappa$ and $\iota\times\iota$ are monomorphisms. The homomorphism $\subseteq$ in the last column is nothing but an inclusion, therefore injective.

As far as $\curv_1$ is concerned, consider an element $\omega=\ast \tilde{\omega}\in \Omega_\Z^k\cap \ast \Omega_\Z ^{m-k}(M)$. Due to the surjectivity of $\curv$, choose $h\in \hat{H}^k(M;\Z)$ and $\tilde{h}\in \hat{H}^{m-k}(M;\Z)$ such that $\curv h=\omega$ and $\curv \tilde{h}=\tilde{\omega}$. Then, $(h, \tilde{h})\in \C^k(M;\Z)$ and $\curv_1$ is surjective.

Let $dA=\ast d\tilde{A}\in d\Omega^{k-1}\cap \ast d\Omega^{m-k-1}(M)$. Then $([A],[\tilde{A}])$ lies in $\mathfrak{T}^k(M;\Z)$ and $d_1([A],[\tilde{A}])=dA=\ast d\tilde{A}$, entailing the surjectivity of $d_1$.

The exactness of the first two rows at the middle object is again a straightforward consequence of their counterparts in diagram \eqref{diag:2:diff_char}. The same holds true for the last two columns, due to the same reason.

Let us come to showing that $([\arg],[\ast^{-1}\arg])$ is surjective too. Pick $(z,\tilde{z})\in H^{k,m-k}_\free(M;\Z)$ and realize it via the de Rham theorem as $(z,\tilde{z})\overset{dR}{=}([\omega],[\tilde{\omega}])$ for some $\omega\in \Omega^{k}_\Z(M)$, $\tilde{\omega}\in \Omega^{m-k}_\Z(M)$. Denoting by $\delta$ the codifferential, solve for $\theta \in \Omega^{k+1}(M)$ and $\tilde{\theta}\in\Omega^{m-k+1}(M)$ the equations $[\delta \theta]=[\omega]$ and $[\delta \tilde{\theta}]=[\tilde{\omega}]$. A solution exists; in fact, let $\square:=d\delta + \delta d$ be the Laplace-de Rham operator on $k$-forms, let $G^{\pm}_k$ be its advanced/retarded Green's operators and let $G_k=G^+_k-G^-_k$ be the causal propagator. Following \cite{B15}, introduce a partition of unity $\{\chi_+,\chi_-\}$ on $M$ such that $\chi_{+}$ has past compact support and $\chi_-$ has future compact support. As regards the first equation, $\theta=G_{k+1}(d\chi_+\wedge \omega)$ is a solution; as a matter of fact:
\begin{align*}
\delta \theta & = \delta G_{k+1}(d\chi_+\omega)=\delta d G_k(\chi_+\omega)\\
&= \delta d \left(G^+_k(\chi_+\omega)-G^-_k(\chi_+\omega)-G^-_k(\chi_-\omega)+G^-_k(\chi_-\omega)\right)\\
&= \delta d\left(G^+_k(\chi_+\omega)+G^-_k(\chi_-\omega)\right) -\delta G^-_{k+1}(d\omega)\\
&= \square \left(G^+_k(\chi_+\omega)+G^-_k(\chi_-\omega)\right) - d\delta \left(G^+_k(\chi_+\omega)+G^-_k(\chi_-\omega) \right)\\
&= \omega - d\delta \left(G^+_k(\chi_+\omega)+G^-_k(\chi_-\omega) \right).
\end{align*}
A similar reasoning works for $\tilde{\theta}$.

Set $F:=\delta \theta +\ast \delta \tilde{\theta}$. The properties of the codifferential yield $[F]=[\delta\theta]=[\omega]\in H^k_{dR}(M)$ and $[\ast^{-1}F]=[\delta \tilde{\theta}]=[\tilde{\omega}]\in H_{dR}^{m-k}(M)$. Since $\omega$ and $\tilde{\omega}$ have integer periods, the same holds true for $F$ and $\ast^{-1}F$. Hence, $F\in \Omega^k_\Z\cap \ast \Omega^{m-k}_\Z(M)$.

At last, the surjectivity of $\cha\times \cha$. Choose an element $(x,\tilde{x})\in H^{k,m-k}(M;\Z)$ and let $(z,\tilde{z})$ be its image in $H^{k,m-k}_\free(M;\Z)$. As $\curv_1$ and $([\arg],[\ast^{-1}\arg])$ are both epimorphisms, their composition $([\arg],[\ast^{-1}\arg])\circ \curv_1$ is an epimorphism as well. We can, therefore, pick a pre-image $(h,\tilde{h})\in\C^k(M;\Z)$ of $(z,\tilde{z})$ via $([\arg],[\ast^{-1}\arg])\circ \curv_1$. The exactness of the bottom row entails that we can write $(x,\tilde{x})=(\cha h, \cha \tilde{h})+(t,\tilde{t})$, for some $(t,\tilde{t})\in H^{k,m-k}_{\tor}(M;\Z)$. Let $(u,\tilde{u})\in H^{k-1,m-k-1}(M;\T)$ be a pre-image of $(t,\tilde{t})$ in the middle element of the first column; thanks to commutativity of the left-bottom square, we obtain that $(h+\kappa u, \tilde{h}+ \kappa \tilde{u})$ is a pre-image of $(x,\tilde{x})$ via $\cha \times \cha$ and $\cha\times\cha$ is surjective.
\end{proof}

\begin{rem}
Inheriting functorial properties from $\hat{H}^k(\arg;\Z)$, $\C^k(\arg;\Z): \textsf{Loc}_m \to \textsf{Ab}$ is a contravariant functor. For each morphism $f:M\to M'$ in $\textsf{Loc}_m$ we denote by $f^\ast$ its counterpart $\C^k(f;\Z):\C^k(M';\Z) \to \C^k(M;\Z)$ in $\textsf{Ab}$. Notice that, as $f$ preserves the metric, the condition in Definition \ref{def:2:semi_class_conf_space} is not affected by the pull-back.\\
\end{rem}

The semi-classical configuration space can be presented as the space of solutions of a well-posed Cauchy problem and, consequently, put in bijective correspondence with initial data on an arbitrary Cauchy surface $\Sigma \subseteq M$. 

By \cite{DL12} and \cite{FL16}, denoting by $\iota_\Sigma: \Sigma\to M$ the embedding map, for every $(B,\tilde{B})\in \Omega^{k,m-k}_d(\Sigma)$ there exists a unique $F\in \Omega^k(M)$ solving the initial value problem,
\begin{subequations}\label{eq:2:cauchy_problem_F}
\begin{align}
dF=0,\qquad &\iota^\ast_\Sigma F=B,\\
d\ast^{-1}F=0,\qquad &\iota^\ast_\Sigma\ast^{-1} F=\tilde{B},
\end{align}
\end{subequations}
with $\text{supp}(F)\subseteq J(\text{supp}(B) \cup \text{supp}(\tilde{B}))$.

Analogously, with the initial condition $(\tilde{B},B)\in \Omega^{m-k,k}_d(\Sigma)$, let $\tilde{F}$ be the solution to the Cauchy problem:
\begin{subequations}\label{eq:2:cauchy_problem_F_tilde}
\begin{align}
d\tilde{F}=0,\qquad &\iota^\ast_\Sigma \tilde{F}=\tilde{B},\\
d\ast\tilde{F}=0,\qquad &\iota^\ast_\Sigma\ast \tilde{F}=B,
\end{align}
\end{subequations}
with $\text{supp}(\tilde{F})\subseteq J(\text{supp}(B) \cup \text{supp}(\tilde{B}))$. Combining \eqref{eq:2:cauchy_problem_F} and \eqref{eq:2:cauchy_problem_F_tilde}, one shows that $F-*\tilde{F}$ solves the Cauchy problem \eqref{eq:2:cauchy_problem_F} with vanishing initial data and, hence, $F=\ast \tilde{F}$. Furthermore, it can be proved by \cite[Eq. (2.11)-(2.12)]{BPhys} that if, in addition, the initial data has integer period $(B,\tilde{B})\in \Omega^{k,m-k}_\Z(M)$ so do $F$ and $\ast^{-1}F$. We thus obtain:
\begin{prop}\label{prop:2:iso_config_diffchar}
Let M be an object in $\mathsf{Loc}_m$ and $\Sigma\subseteq M$ a Cauchy surface. The embedding $\iota_\Sigma: \Sigma \to M$ induces an isomorphism of Abelian groups:
\begin{equation}
\Omega_\Z^k\cap \ast \Omega^{m-k}_\Z(M) \overset{(\iota^\ast_\Sigma, \iota^\ast_\Sigma \ast^{-1})}{\longrightarrow}   \Omega_\Z^{k,m-k}(\Sigma),
\end{equation}
whose inverse is the map assigning to each initial data $(B,\tilde{B})\in \Omega_\Z^{k,m-k}(\Sigma)$ the unique solution $F$ of the Cauchy problem \eqref{eq:2:cauchy_problem_F}.
\end{prop}

Since the inclusion of flat classes and the curvature maps are natural transformations, the diagram:
\begin{equation*}
\xymatrix{
0 \ar[r] & H^{k-1,m-k-1}(M;\T) \ar[r]^-{\kappa\times\kappa} \ar[d]^-{\iota^\ast_\Sigma\times\iota^\ast_\Sigma} & \C^k(M;\Z) \ar[r]^-{\curv_1} \ar[d]^-{\iota^\ast_\Sigma\times\iota^\ast_\Sigma} & \Omega_\Z^k\cap \ast\Omega^{m-k}_\Z(M)\ar[r]\ar[d]^-{(\iota^\ast_\Sigma,\iota^\ast_\Sigma \ast^{-1})} &0\\
0 \ar[r] & H^{k-1,m-k-1}(\Sigma;\T) \ar[r]^-{\kappa\times\kappa} & \hat{H}^{k,m-k}(\Sigma;\Z) \ar[r]^-{\curv\times\curv} & \Omega_\Z^{k,m-k}(\Sigma) \ar[r] &0\\
}
\end{equation*}
is commutative and has exact rows. In fact, the upper row is the central row in diagram \eqref{diag:2:scconfig}, whilst the lower one is obtained from the central row in diagram \eqref{diag:2:diff_char} where $M$ has been replaced with the Cauchy surface $\Sigma$. The right vertical arrow is an isomorphism by Proposition \ref{prop:2:iso_config_diffchar}; from \cite[Lemma A.1]{BPhys} it ensues that the left vertical arrow is an isomorphism too. Therefore, the central vertical arrow turns out to be an isomorphism, in view of the Five Lemma \cite[Lemma 7.1]{MAS91}.

\begin{thm}\label{thm:2:isomorphism_i_C}
Let $M$ be an object in $\mathsf{Loc}_m$ and $\Sigma\subseteq M$ a Cauchy surface. The embedding map $	\iota_\Sigma: \Sigma \to M$ induces an isomorphism of Abelian groups:
\begin{equation*}
\xymatrix{
\C^k(M;\Z) \ar[r]^-{\iota^\ast_\Sigma\times\iota^\ast_\Sigma} & \hat{H}^{k,m-k}(\Sigma;\Z).
}
\end{equation*}
\end{thm}

The above theorem states that the initial value problem:
\begin{equation*}
\curv h=\ast\curv \tilde{h}, \qquad \iota^\ast_\Sigma h=h_\Sigma, \qquad \iota^\ast_\Sigma \tilde{h}=\tilde{h}_\Sigma,
\end{equation*}
for $(h,\tilde{h})\in \C^k(M;\Z)$ with initial data $(h_\Sigma, \tilde{h}_\Sigma)\in \hat{H}^{k,m-k}(\Sigma;\Z) $ is well-posed. In particular, the semi-classical configuration space can be interpreted as the very space of solutions to such initial value problem with data specified by pairs of differential characters on a Cauchy surface.

\subsection{Semi-classical observables}
For a generic field theory, the space of observables comprises functionals over the relevant configuration space. In the case in hand, $\C^k(M;\Z)$ possesses the additional structure of an Abelian group.  Therefore, a distinguished class of functionals is available: the Abelian group of characters $\C^k(M;\Z)^\star:=\Hom(\C^k(M;\Z), \T)$. We are in a situation similar to that of Section \ref{section_two_duality_pairing}: some elements of this space are too singular for our purposes. Therefore, we refine our notion of observable, restricting $\C^k(M;\Z)^\star$ to a more regular subgroup:
\begin{defn}[Semi-classical observables]\label{def:2:semi_classical_observables}
We call \emph{semi-classical observables} the subspace $\mathfrak{D}^k(M;\Z)$ of elements $\varphi\in\C^k(M;\Z)^\star$ for which there exists $\omega=\ast\tilde{\omega}\in \Omega^k_{sc,\Z}\cap \ast\Omega_{sc,\Z}^{m-k}(M)$\footnote{The subscript $_{\text{sc}}$ denotes forms with spacelike compact support. A differential form $\omega\in\Omega^k(M)$ has spacelike compact support if there exists a compact set $K\subseteq M$ such that $\text{supp}(\omega)\subseteq J(K)$.} such that, for some Cauchy surface $\Sigma$, it holds:
\begin{equation}\label{eq:2:semi_classical_observables}
\varphi\left( (\iota\times\iota)([A],[\tilde{A}]) \right)=\int_\Sigma \left( \tilde{A}\wedge \omega -(-1)^{k(m-k)}A\wedge \tilde{\omega}\right)\,\,\mod \Z
\end{equation}
for all $([A],[\tilde{A}])\in\mathfrak{T}^k(M;\Z)$. 
\end{defn}

\begin{rem}
The definition of semi-classical observables is independent of the choice of the Cauchy surface $\Sigma$. Let us prove it explicitly. Take $\omega=\ast\tilde{\omega}\in \Omega^k_{sc,\Z}\cap \ast\Omega_{sc,\Z}^{m-k}(M)$ as in Definition \ref{def:2:semi_classical_observables}. Observe that $\omega$ and $\ast\tilde{\omega}$ are closed and coclosed; therefore, $\square\, \omega=0$ and $\square\, \tilde{\omega}=0$, where $\square=d\delta + \delta d$ is the D'Alembert-De Rham operator.
Denoting by $G^{\pm}_k$ the advanced/retarded Green operators and $G_k=G^+_k-G^-_k$ the causal propagator of the Green hyperbolic operator $\square$ acting on the $k$-forms, there exists $\tilde{\beta}\in\Omega^{m-k}_c(M)$ such that $G_{m-k}\tilde{\beta}=\tilde{\omega}$. $d\omega=0$ implies that $dG_{k}\ast\tilde{\beta}=G_{k+1} d\ast\tilde{\beta}=0$, i.e.\ $d\ast\tilde{\beta}=\square\,\alpha$ for some $\alpha\in\Omega^{k+1}_c(M)$. Likewise, $d\tilde{\omega}=0$ entails $dG_{m-k}\tilde{\beta}=G_{m-k+1}d\tilde{\beta}=0$, i.e.\ $d\tilde{\beta}=\square\,\tilde{\alpha}$ for some $\tilde{\alpha}\in \Omega^{m-k+1}_c(M)$. Interpreting $\Sigma$ as the boundary of $J^-(\Sigma)\subseteq M$ and, at the same time, as the boundary of $J^+(\Sigma)\subseteq M$ with reversed orientation, we obtain the following chain of identities:
\begin{align}\label{eq:2:independence_Sigma_scobservables}
&\varphi  \left( (\iota\times\iota)([A],[\tilde{A}])\right) = \int_\Sigma \left(\tilde{A}\wedge G_k\ast\tilde{\beta}-(-1)^{k(m-k)}A\wedge G_{m-k}\tilde{\beta} \right) \,\,\mod\Z\nonumber\\
&= \int_{J^-(\Sigma)} d\left(  \tilde{A}\wedge G^+ _k \ast\tilde{\beta}-(-1)^{k(m-k)}A\wedge G^+_{m-k}\tilde{\beta}\right)\,\,\mod\Z\nonumber\\
&+ \int_{J^+(\Sigma)} d\left(  \tilde{A}\wedge G^-_k \ast\tilde{\beta} -(-1)^{k(m-k)}A\wedge G^-_{m-k}\tilde{\beta}\right)\,\,\mod\Z\nonumber\\
&= \int_{J^-(\Sigma)} \left(  (-1)^{m-k}\tilde{A}\wedge G^+_{k+1}\square \alpha -(-1)^{k(m-k)}(-1)^k A\wedge G^+_{m-k+1}\square \tilde{\alpha}\right)\,\,\mod\Z\nonumber\\
&+ \int_{J^+(\Sigma)} \left(  (-1)^{m-k}\tilde{A}\wedge G^-_{k+1}\square \alpha -(-1)^{k(m-k)}(-1)^k A\wedge G^-_{m-k+1}\square \alpha\right)\,\,\mod\Z\nonumber\\
&= \int_M \left( (-1)^{m-k} \tilde{A} \wedge \alpha - (-1)^{k(m-k+1)} A\wedge \tilde{\alpha} \right) \,\,\mod\Z,
\end{align}
where we made use of Stoke's theorem and of the property $G^\pm_k\square=\square G^\pm_k= \text{Id}$ for compactly supported forms. Thus, \eqref{eq:2:independence_Sigma_scobservables} recasts \eqref{eq:2:semi_classical_observables} as an integral over the whole spacetime $M$ in a way manifestly independent of the choice of $\Sigma$.
\end{rem}

The assignment:
\begin{gather*}
\mathfrak{D}^k(\arg;\Z):\mathsf{Loc}_m\to \mathsf{Ab}\\
M\mapsto \mathfrak{D}^k(M;\Z)
\end{gather*}
defines a covariant functor. Firstly observe that the assignment of characters group $\C^k(M;\Z)^\star=\Hom(\C^k(M;\Z);\T)$ to each $M$ in $\mathsf{Loc}_m$, without any regularity restriction, is a covariant functor. In fact, consider a morphism $f:M\to M'$. By dualizing with respect to $^\star$ the associated morphism $f^\ast=\C^k(f;\Z)$ provided by the functoriality of the semi-classical configuration space, we get a pushforward map of characters along $f$:
$$
f_\ast=:(f^\ast)^\star: \C^k(M;\Z)^\star \to \C^k(M';\Z)^\star.
$$
Now, $\mathfrak{D}^k(\arg;\Z)$ can be realized as a subfunctor of $\C^k(\arg;\Z)^\star$, the regularity condition \eqref{eq:2:semi_classical_observables} being preserved by $f_\ast$. Let $\varphi\in\mathfrak{D}^k(M;\Z)$ and let $\omega=\ast\tilde{\omega}\in\Omega^k_{sc,\Z}\cap \ast\Omega^{m-k}_{sc,\Z}(M)$ be as in Definition \ref{def:2:semi_classical_observables}. In view of Proposition \ref{prop:2:iso_config_diffchar}, we can push forward $\omega$ and $\tilde{\omega}$ to $f_\ast\omega=\ast f_\ast \tilde{\omega}\in \Omega_{sc,\Z}^k\cap \ast\Omega^{m-k}_{sc,\Z}(M')$ by pushing forward the initial data for the corresponding Cauchy problem from $\Sigma$ to $\Sigma '$, where $\Sigma '$ is obtained extending $f(\overline{U})$ to a Cauchy surface in $M'$, with $U\subseteq \Sigma$ an open relatively compact neighbourhood of $\text{supp}(\omega)\cap \Sigma$ (cfr.\ \cite{BS06}). Then, for each $([A],[\widetilde A]) \in \mathfrak{T}^k(M^\prime)$:
\begin{equation}
\begin{split}
f_\ast\varphi\left((\iota\times\iota)([A],[\tilde{A}]) \right)= &\varphi \left((\iota\times\iota)([f^\ast A],[f^\ast \tilde{A}]) \right)\\
=& \int_\Sigma \left(f^\ast \tilde{A}\wedge \omega -(-1)^{k(m-k)}f^\ast A\wedge \tilde{\omega} \right)\,\,\mod\Z\\
=& \int_{\Sigma'} \left( \tilde{A}\wedge f_\ast \omega -(-1)^{k(m-k)} A\wedge f_\ast\tilde{\omega} \right)\,\,\mod\Z.
\end{split}
\end{equation}

\subsection{Space-like compact gauge fields}
In the spirit of differential characters with compact support, the present section will be devoted to show that the Abelian group of space-like compact gauge fields $\mathfrak{C}^k_{sc}(M;\Z)$ is isomorphic to the group $\mathfrak{D}^k(M;\Z)$ of semi-classical observables. To begin with, let us define the former properly.\\

Let $K\subseteq M$ be a compact subset. Define the \emph{semi-classical gauge fields on $M$ relative to $M\setminus J(K)$} as the Abelian group $\C^k(M,M\setminus J(K);\Z)$ of elements $(h,\tilde{h})\in \hat{H}^{k,m-k}(M,M\setminus J(K);\Z)$  such that $\curv h=\ast \curv \tilde{h}\in \Omega^{k}(M,M\setminus J(K))$, with $\hat{H}^{k,m-k}(M,M\setminus J(K);\Z)$ given by \eqref{eq:2:functor_diff_char_comp_supp}.
\begin{defn}
We call \emph{semi-classical gauge fields with space-like compact support} the colimit:
\begin{equation}\label{eq:2:definition_sc_gauge_gauge_fields}
\C^k_{sc}(M;\Z):=\text{colim} \left( \C^k(M,M\setminus J(\arg);\Z):\mathcal{K}_M \to \mathsf{Ab} \right).
\end{equation}
\end{defn}

\begin{rem}\label{rem:2:colim_sigma}
An alternative but equivalent procedure consists in taking the colimit over the directed set $\mathcal{K}_\Sigma$ rather than over $\mathcal{K}_M$. Denoting by $\mathcal{C}_M$ the directed set of closed subsets of $M$, introduce the map $J:\mathcal{K}_M\to\mathcal{C}_M$, $K\mapsto J(K)$. The functor $\C^k(M,M\setminus J(\arg);\Z):\mathcal{K}_M\to \mathsf{Ab}$ can be thought of as the composition of the functors $\C^k(M,M\setminus \arg;\Z):\mathcal{C}_M\to \mathsf{Ab}$ and $J:\mathcal{K}_M\to \mathcal{C}_M$. $J$ preserves the preorder relation and $\mathcal{K}_\Sigma\subseteq \mathcal{K}_M$ is cofinal with respect to $J$. Then:
\begin{equation}
\C^{k}_{sc}(M;\Z)\simeq \text{colim}\left( \C^k(M,M\setminus J(\arg);\Z): \mathcal{K}_\Sigma \to \mathsf{Ab} \right).
\end{equation}
\end{rem}

The assignment of the Abelian group $\C^k_{sc}(M;\Z)$ to each object $M$ in $\mathsf{Loc}_m$ is a covariant functor:
$$
\C^k_{sc}(\arg;\Z):\mathsf{Loc}_m \to \mathsf{Ab}.
$$
For the explicit construction of the morphism $f_\ast:=\C^k_{sc}(f;\Z):\C^k_{sc}(M;\Z)\to\C^k_{sc}(M';\Z)$ corresponding to $f:M\to M'$ in $\mathsf{Loc}_m$ see \cite[Lemma A.3]{BPhys}.

\begin{rem}
The group homomorphism $I:\hat{H}^k(M,M\setminus J(K);\Z)\to \hat{H}^k(M;\Z)$ defined by \eqref{map:2:I_relative_diff_char} induces a homomorphism 
$$
I:\C^k(M,M\setminus J(K);\Z)\to \C^k(M;\Z),
$$ 
which, in turn, yields a homomorphism $I:\C^k_{sc}(M;\Z)\to \C^k(M;\Z)$ via the colimit. For the sake of clarity, we denote all these homomorphisms by the same letter, $I$; it will be clear from the context what is the domain and what is the target in each case. As before, the map $I$ is neither injective nor surjective, in general, thus entailing that space-like compact gauge fields cannot be represented faithfully as elements in $\C^k(M;\Z)$. The fact that $\C^k_{sc}(M;\Z)$ is \emph{not} a subgroup of $\C^k(M;\Z)$ motivates its sophisticated definition.
\end{rem}

Let $K\subseteq \Sigma$ be a compact set and pick $(B,\tilde{B})\in\Omega^{k,m-k}_\Z(\Sigma, \Sigma \setminus K)$. From the support properties of the Cauchy problem \eqref{eq:2:cauchy_problem_F}, retracing the above reasoning, it is easy to check that there exists a unique solution $F\in\Omega^k_\Z\cap\ast\Omega^{m-k}_\Z(M,M\setminus J(K))$ to \eqref{eq:2:cauchy_problem_F}. Therefore, we obtain a relative version of Proposition \ref{prop:2:iso_config_diffchar}:
\begin{prop}\label{prop:2:cauchy_probl_relative_forms}
Let M be an object in $\mathsf{Loc}_m$, $\Sigma\subseteq M$ a Cauchy surface and $K\subseteq M$ a compact set. The embedding $\iota_\Sigma: \Sigma \to M$ induces an isomorphism of Abelian groups:
\begin{equation}
\Omega_\Z^k\cap \ast \Omega^{m-k}_\Z(M,M\setminus J(K)) \overset{(\iota^\ast_\Sigma, \iota^\ast_\Sigma \ast^{-1})}{\longrightarrow}   \Omega_\Z^{k,m-k}(\Sigma,\Sigma \setminus K),
\end{equation}
whose inverse is the map assigning to each initial data $(B,\tilde{B})\in \Omega_\Z^{k,m-k}(\Sigma,\Sigma\setminus K)$ the unique solution $F\in\Omega^k_\Z\cap\ast\Omega^{m-k}_\Z(M,M\setminus J(K))$ of the Cauchy problem \eqref{eq:2:cauchy_problem_F}.
\end{prop}

Recalling the functoriality of relative differential cohomology and diagram \eqref{diag:2:relative_diff_char}, we obtain that the diagram
{\scriptsize
\begin{equation}\label{diag:2:biline_relative}
\xymatrix@C-=0.47cm{
0 \ar[r] & H^{k-1,m-k-1}(M,M	\setminus J(K);\T) \ar[r]^-{\kappa\times\kappa} \ar[d]^-{\iota^\ast_\Sigma\times\iota^\ast_\Sigma} & \C^k(M,M\setminus J(K);\Z) \ar[r]^-{\curv_1} \ar[d]^-{\iota^\ast_\Sigma\times\iota^\ast_\Sigma} & \Omega_\Z^k\cap \ast\Omega^{m-k}_\Z(M,M\setminus J(K))\ar[r]\ar[d]^-{(\iota^\ast_\Sigma,\iota^\ast_\Sigma \ast^{-1})} &0\\
0 \ar[r] & H^{k-1,m-k-1}(\Sigma,\Sigma\setminus K;\T) \ar[r]^-{\kappa\times\kappa} & \hat{H}^{k,m-k}(\Sigma,\Sigma\setminus K;\Z) \ar[r]^-{\curv\times\curv} & \Omega_\Z^{k,m-k}(\Sigma,\Sigma \setminus K) \ar[r] &0\\
}
\end{equation}
}commutes and has exact rows. The third and the first vertical arrows are isomorphisms by, respectively, Proposition \ref{prop:2:cauchy_probl_relative_forms} and \cite[Lemma A.2]{BPhys}. The reduced Five Lemma then yields:
\begin{thm}
The embedding map $\iota_\Sigma:\Sigma \to M$ induces an isomorphism of Abelian groups:
\begin{equation*}
\xymatrix{
\C^k(M,M\setminus J(K);\Z) \ar[r]^-{\iota_\Sigma^\ast\times \iota_\Sigma^\ast} &\hat{H}^k(\Sigma, \Sigma\setminus K;\Z).
}
\end{equation*}
\end{thm}
Taking the colimit of \eqref{diag:2:biline_relative} over the directed set $\mathcal{K}_\Sigma$ and recalling Remark \ref{rem:2:colim_sigma}, we obtain the commutative diagram:
\begin{equation}\label{diag:2:iso_Csc_H_sigma}
\xymatrix@C-=0.5cm{
0 \ar[r] & H^{k-1,m-k-1}_{sc}(M;\T) \ar[r]^-{\kappa\times\kappa} \ar[d]^-{\iota^\ast_\Sigma\times\iota^\ast_\Sigma} & \C^k_{sc}(M;\Z) \ar[r]^-{\curv_1} \ar[d]^-{\iota^\ast_\Sigma\times\iota^\ast_\Sigma} & \Omega_{sc,\Z}^k\cap \ast\Omega^{m-k}_{sc,\Z}(M)\ar[r]\ar[d]^-{(\iota^\ast_\Sigma,\iota^\ast_\Sigma \ast^{-1})} &0\\
0 \ar[r] & H^{k-1,m-k-1}_c(\Sigma;\T) \ar[r]^-{\kappa\times\kappa} & \hat{H}^{k,m-k}_c(\Sigma;\Z) \ar[r]^-{\curv\times\curv} & \Omega_{c,\Z}^{k,m-k}(\Sigma) \ar[r] &0\\
}
\end{equation}
whose rows are exact and whose vertical arrows are isomorphisms, because colim is an exact functor over diagrams of Abelian groups. What we have shown is that $\C^k_{sc}(M;\Z)$ can be realised as the space of solutions to the initial value problem:
$$
\curv h=\ast \curv \tilde{h}, \qquad \iota^\ast_\Sigma h=h_\Sigma, \qquad \iota^\ast_\Sigma \tilde{h}=\tilde{h}_\Sigma
$$
for $(h,	\tilde{h})\in \hat{H}^{k,m-k}_{sc}(M;\Z)$ with initial data $(h_\Sigma,\tilde{h}_\Sigma)\in\hat{H}^{k,m-k}_c(\Sigma;\Z)$.
\begin{thm}\label{thm:2:isomorphism_i_C_sc}
Let $M$ be an object in $\mathsf{Loc}_m$ and $\Sigma \subseteq M$ a Cauchy surface. The embedding $\iota_\Sigma:\Sigma\to M$ induces an isomorphism of Abelian groups:
\begin{equation}
\xymatrix{
\C^k_{sc}(M;\Z) \ar[r]^-{\iota^\ast_\Sigma \times \iota^\ast_\Sigma} & \hat{H}^{k,m-k}_c(\Sigma;\Z).
}
\end{equation}
\end{thm}

Now, for every Cauchy surface $\Sigma$ in $M$, a $\T$-valued pairing between differential characters $\hat{H}^k(\Sigma;\Z)$ and differential characters with compact support $\hat{H}^{m-k}_c(\Sigma;\Z)$ is defined through \eqref{pairing:2:duality_diff_char}:
$$
\pa{\arg}{\arg}_c: \hat{H}^k(\Sigma;\Z) \times \hat{H}^{m-k}_c(\Sigma;\Z) \to \T.
$$
Exploiting Theorem \ref{thm:2:isomorphism_i_C} and Theorem \ref{thm:2:isomorphism_i_C_sc}, we introduce a $\T$-valued pairing between semi-classical configurations $\C^k(M;\Z)$ and space-like compact gauge fields $\C^k_{sc}(M,\Z)$ by:
\begin{eqnarray}\label{pairing:2:config_gauge_fields}
&\pa{\arg}{\arg}:\C^k(M;\Z) \times \C^k_{sc}(M;\Z) \to \T \nonumber&\\
&\left((h,\tilde{h}),(h',\tilde{h}') \right)\mapsto \pa{\iota^\ast_\Sigma \tilde{h}}{\iota^\ast_\Sigma h'}_c - (-1)^{k(m-k)}\pa{\iota^\ast_\Sigma h}{\iota^\ast_\Sigma \tilde{h}'}_c.&
\end{eqnarray}

\begin{prop}
The pairing \eqref{pairing:2:config_gauge_fields} does not depend on the choice of the Cauchy surface.
\end{prop}
\begin{proof}
Let $\Sigma$ and $\Sigma'$ be two arbitrary Cauchy surfaces of $M$. We prove that the difference between \eqref{pairing:2:config_gauge_fields} evaluated on $\Sigma$ and $\Sigma '$ respectively is vanishing. Pick $(h,\tilde{h})\in\C^k(M;\Z)$ and, for some compact $K\subseteq \Sigma$, let $(h',\tilde{h}')\in\C^k(M,M\setminus J(K);\Z)$ be a representative of $(h',\tilde{h}')\in\C^k_{sc}(M;\Z)$, denoted, with a slight abuse of notation, by the same symbol. Let $\mu\in H_{m-1}(\Sigma,\Sigma\setminus K)$ be the unique homology class which restricts to the orientation of $\Sigma$ for each point of $K$. Analogously, set $K':=\Sigma '\cap J(K)$ and let $\mu '\in H_{m-1}(\Sigma ',\Sigma' \setminus K')$ be the unique homology class which restricts to the orientation of $\Sigma '$ for each point of $K'$. The orientations of both $\Sigma$ and $\Sigma '$ are chosen consistently with the orientation of $M$. For every choice of a Cauchy surface $\tilde{\Sigma}\subseteq M$, global hyperbolicity provides with an embedding $\iota_{\tilde{\Sigma}}:\tilde{\Sigma}\to M$ and with a projection $\pi_{\tilde{\Sigma}}:M\to \tilde{\Sigma}$. We can, therefore, exploit such maps to compare $\mu$ with the homology class $\tilde{\mu}=\pi_{\Sigma\ast}\iota_{\Sigma '\ast}\mu' \in H_{m-1}(\Sigma,\Sigma\setminus K;\Z)$, obtained mapping $\mu '$ to $\Sigma$. As $\tilde{\mu}$ also will restrict to the orientation of $\Sigma$ for every point of $K$, due to uniqueness we obtain $\tilde{\mu}=\mu$. By \cite[Lemma A.2]{BPhys}, it ensues that $\iota_{\Sigma '\ast}\mu'=\iota_{\Sigma \ast} \mu \in\ H_{m-1}(M,M \setminus J(K))$. Choosing representatives $\nu\in Z_{m-1}(\Sigma, \Sigma\setminus K)$ for $\mu$ and $\nu'\in Z_{m-1}(\Sigma ',\Sigma '\setminus K')$ for $\mu'$, their pushforwards differ by a boundary. Hence, there exists $\gamma \in C_{m}(M,M\setminus J(K))$ such that $\iota_{\Sigma\ast}\nu- \iota_{\Sigma '\ast}\nu'=\partial \gamma$. Therefore:
\begin{equation*}
\begin{split}
&\pa{(h,\tilde{h})}{(h',\tilde{h}')}_{\Sigma}-\pa{(h,\tilde{h})}{(h',\tilde{h}')}_{\Sigma '}= \left( \tilde{h}\cdot h' -(-1)^{k(m-k)}h\cdot \tilde{h}' \right)(\partial \gamma)\\
&= \int_\gamma \curv \left( \tilde{h}\cdot h' -(-1)^{k(m-k)}h\cdot \tilde{h}' \right)\,\,\mod\Z\\
&= \int_\gamma\left( \curv\tilde{h}\wedge \ast \curv \tilde{h}'-(-1)^{k(m-k)} \ast \curv \tilde{h}\wedge \curv \tilde{h}'\right) \,\,\mod \Z\\
&=\int_\gamma \left( \curv \tilde{h}\wedge \ast\curv \tilde{h}'-\curv\tilde{h}'\wedge \ast\curv\tilde{h}\right)\,\,\mod\Z=0,
\end{split}
\end{equation*}
where we have used \eqref{eq:2:module_structure} and $\curv h=\ast \curv\tilde{h}$, $\curv h'=\ast \curv\tilde{h}'$. The subscripts $_\Sigma$ and $_{\Sigma'}$ indicate the Cauchy surface with respect to which we are computing the pairing.
\end{proof}

The pairing \eqref{pairing:2:config_gauge_fields} allows to realize elements in $\C^k_{sc}(M;\Z)$ as group characters on $\C^k(M;\Z)$:
\begin{eqnarray*}
\C^k_{sc}(M;\Z)\ni (h',\tilde{h}')\mapsto \pa{\arg}{(h',\tilde{h}')} \in \C^k(M;\Z)^\star.
\end{eqnarray*}
They are actually something more than a generic subgroup of the character group, since, as group characters, they separate points in the semi-classical configuration space.

\begin{prop}\label{prop:2:non_deg_pair_C_Csc}
The pairing \eqref{pairing:2:config_gauge_fields} is weakly non-degenerate.
\end{prop}
\begin{proof}
Since the embedding map $\iota:\Sigma \to M$ induces the isomorphisms $\C^k(M;\Z)\simeq \hat{H}^{k,m-k}(\Sigma;\Z)$ and $\C^k_{sc}(M;\Z)\simeq \hat{H}^{k,m-k}_c(\Sigma,\Z)$ by Theorem \ref{thm:2:isomorphism_i_C} and Theorem \ref{thm:2:isomorphism_i_C_sc} respectively, \eqref{pairing:2:config_gauge_fields} is nothing but the pairing between the intial data of the relevant Cauchy problems:
$$
\pa{\arg}{\arg}:\hat{H}^{k,m-k}(\Sigma;\Z)\times \hat{H}^{k,m-k}_c(\Sigma;\Z)\to \T
$$
$$
\left((\iota^\ast_\Sigma h, \iota^\ast_\Sigma \tilde{h}),(\iota^\ast_\Sigma h',\iota^\ast_\Sigma \tilde{h}')\right)\mapsto \pa{\iota^\ast_\Sigma \tilde{h}}{\iota^\ast_\Sigma h'}_c-(-1)^{k(m-k)}\pa{\iota^\ast_\Sigma h}{\iota^\ast_\Sigma \tilde{h}'}_c
$$
By resorting to Corollary \ref{cor:2:non_deg_pair_c}, we conclude.
\end{proof}

Let us come to the main point of the section: the maps obtained by partial evaluation of $\eqref{pairing:2:config_gauge_fields}$ lie in $\mathfrak{D}^k(M;\Z)$ and every semi-classical observable is of this form.

\begin{prop}
The group homomorphism:
\begin{eqnarray}\label{eq:2:iso_D_C_sc}
&\mathcal{O}:\C^k_{sc}(M;\Z) \to \mathfrak{D}^k(M;\Z) \nonumber &\\
& (h',\tilde{h}')\mapsto \pa{\arg}{(h',\tilde{h}')}	&
\end{eqnarray}
is a natural isomorphism between the functors $\C^k(\arg;\Z)$ and $\mathfrak{D}^k(\arg;\Z)$ from the category $\mathsf{Loc}_m$ to the category $\mathsf{Ab}$.
\end{prop}
\begin{proof}
To begin with, let us prove that the map is surjective. Let $\varphi \in \mathfrak{D}^k(M;\Z)$. In view of Theorem \ref{thm:2:isomorphism_i_C}, we can find a unique element $\varphi_\Sigma\in \hat{H}^{k,m-k}(\Sigma;\Z)^\star_\infty$ such that $\varphi_\Sigma\circ (\iota^\ast_\Sigma \times \iota^\ast_\Sigma)=\varphi$. Then, by invoking Theorem \ref{thm:2:iso_exact_rows}, it is possible to select $(h'_\Sigma,\tilde{h}'_\Sigma)\in \hat{H}^{k,m-k}(\Sigma;\Z)$ such that $\varphi_\Sigma=\langle \arg,(h'_\Sigma,\tilde{h}'_\Sigma)\rangle _\Sigma$. Eventually, Theorem \ref{thm:2:isomorphism_i_C_sc} guarantees the existence of a unique $(h',\tilde{h}')\in\C^k_{sc}(M;\Z)$ whose pull-back to the Cauchy surface $\Sigma$ coincides with $(h'_\Sigma, \tilde{h}'_\Sigma)$. By construction, it holds that $\varphi=\pa{\arg}{(h',\tilde{h}')}$.

As far as injectivity is concerned, firstly we have to show that for every $(h',\tilde{h}')\in \C^k_{sc}(M;\Z)$ the group character $\pa{\arg}{(h',\tilde{h}')}$ satisfies the regularity condition \eqref{eq:2:semi_classical_observables}. For all $([A],[\tilde{A}])\in\mathfrak{T}^k(M;\Z)$ we have:
{\small \begin{equation}
\begin{split}
\pa{(\iota\times\iota)([A],[\tilde{A}])}{(h',\tilde{h}')}= &\left( \iota[\tilde{A}]\cdot h'-(-1)^{k(m-k)} \iota[A]\cdot \tilde{h}' \right) \mu\\
=  & \int_\Sigma \left( \tilde{A}\wedge \curv h'-(-1)^{k(m-k)} A\wedge \curv \tilde{h}' \right) \,\,\mod \Z,
\end{split}
\end{equation}}
where, for the last equality, we made use of \eqref{eq:2:module_structure}. Injectivity now follows at once from Proposition \ref{prop:2:non_deg_pair_C_Csc}.

To prove naturality, take a morphism $f:M\to M'$ in $\mathsf{Loc}_m$. Lemma \cite[Lemma A.4]{BPhys} establishes that the pairing \eqref{pairing:2:config_gauge_fields} is natural. Recalling that $f_\ast:=\mathfrak{D}^k(f,\Z)= (f^\ast)^\star$, we get the following chain of equalities:
$$
f_\ast\pa{\arg}{(h',\tilde{h}')}=\pa{f^\ast\cdot}{(h',\tilde{h}')}=\pa{\arg}{f_\ast(h',\tilde{h}')},
$$
that is to say, the diagram:
\begin{equation}
\xymatrix{
\C^k_{sc}(M;\Z) \ar[d]_-{f_\ast} \ar[r]^-{\mathcal{O}} & \mathfrak{D}^k(M;\Z) \ar[d]^-{f_\ast} \\
\C^k_{sc}(M';\Z) \ar[r]_-{\mathcal{O}} & \mathfrak{D}^k(M';\Z)\\
}
\end{equation}
is commutative and $\mathcal{O}$ is natural.
\end{proof}

\subsection{Covariant field theory and quantization} \label{subsec:2:quantiz}
Brunetti, Fredenhagen and Verch proposed, in 2003, a novel approach to quantum field theory \cite{BFV03}. Instead of looking at the specific properties of each spacetime, they take a category-theoretic perspective and they establish sufficient conditions which allow to discuss the quantization of a given field theory on all spacetimes at once in a coherent way. In this view, they define a covariant functor, called \emph{locally covariant quantum field theory}, from the category $\mathsf{Loc}_m$ to the category of $C^\ast$-algebras $\mathsf{C}^\ast\mathsf{Alg}$, whose morphisms are unit-preserving $C^\ast$-homomorphisms, fulfilling suitable properties called \emph{causality} and \emph{time-slice axiom}. Besides allowing to recover the Haag-Kastler framework \cite{HK64} as a particular case, such an assignment of a $C^\ast$-algebra to each globally hyperbolic spacetime enjoys the remarkable feature of being automatically consistent with the covariance requirements imposed by general relativity, as it implements naturally the general covariance under a group of isometries of the spacetime.

In mathematical terms, we have:

\begin{defn}[Covariant quantum field theory {\cite[Definition 2.1]{BFV03}}] We call \emph{covariant quantum field theory} a functor:
\begin{equation*}
\mathfrak{A}:\mathsf{Loc}_m \to \mathsf{C^\ast Alg}.
\end{equation*}
A locally covariant quantum field theory is a covariant quantum field theory such that all morphisms in $\mathsf{Loc}_m$ induce injective $\ast$-homomorphism in $\mathsf{C^\ast Alg}$.

A covariant quantum field theory $\mathfrak{A}$ is called \emph{causal} if 
\begin{equation}
[\mathfrak{A}(f_1)(\mathfrak{A}(M_1)),\mathfrak{A}(f_2)(\mathfrak{A}(M_2))]=\{0\}
\end{equation}
for all morphisms $f_1:M_1\to M$, $f_2:M_2\to M$ in $\mathsf{Loc}_m$ such that $f_1(M_1)$ and $f_2(M_2)$ are causally separated in $M$.

A covariant quantum field theory $\mathfrak{A}$ satisfies the \emph{time-slice axiom} if
$$
\mathfrak{A}(f): \mathfrak{A}(M) \to \mathfrak{A}(M^\prime)
$$
is an isomorphism for all morphisms $f:M\to M'$ in $\mathsf{Loc}_m$ such that $f(M)$ contains a Cauchy surface $\Sigma'$ of $M'$.
\end{defn}

The first step towards a covariant quantum field theory for differential cohomology is the construction of a pre-symplectic structure on the space of semi-classical observables. Let $\mathsf{PSAb}$ be the category whose objects are pre-symplectic Abelian groups and whose morphisms are the homomorphisms of Abelian groups preserving the pre-symplectic structure. Let us state a preparatory lemma first.

\begin{lem}\label{lem:2:grading_Ih_h'}
The pairing \eqref{pairing:2:comp_supp_char} is graded symmetric, i.e.\ it holds:
$$
\pa{Ih}{h'}_c=(-1)^{k(m-k)}\pa{Ih'}{h}_c,
$$
for all $h\in \hat{H}^k_c(M;\Z)$, $h'\in\hat{H}^{m-k}_c(M;\Z)$.
\end{lem}
\begin{proof}
It is possible to find a sufficiently large compact set $K\subseteq M$ for which $h\in \hat{H}^k(M,M\setminus K;\Z)$ and $h'\in\hat{H}^{m-k}(M,M\setminus K;\Z)$ represent $h\in \hat{H}^k_c(M;\Z)$ and $h'\in\hat{H}^{m-k}_c(M;\Z)$ respectively. It can be checked by technical and tedious computations that the following equalities hold:
$$
Ih\cdot h'=h\cdot h',\qquad Ih'\cdot h=h'\cdot h,
$$
where on the left-hand side $\cdot$ denotes the module structure on $\hat{H}^\ast(M,M\setminus K;\Z)$, while on the right-hand side it denotes the (non-unital) ring structure on relative differential characters. Since the ring structure is graded commutative, we get:
$$
Ih\cdot h'=h\cdot h'=(-1)^{k(m-k)} h'\cdot h= (-1)^{k(m-k)}Ih'\cdot h.
$$
\end{proof}

\begin{prop}
The bilinear map:
\begin{eqnarray}\label{sympl:2:D}
&\tau: \mathfrak{D}^k(M;\Z)\times \mathfrak{D}^k(M;\Z)\to \T \nonumber&\\
& (\varphi, \varphi ')\mapsto \pa{I(\mathcal{O}^{-1}\varphi)}{\mathcal{O}^{-1}\varphi '} &
\end{eqnarray}
endows $\mathfrak{D}^k(M;\Z)$ with a pre-symplectic structure whose radical is $\mathcal{O}(\text{ker}(I))$.
\end{prop}
\begin{proof}
Via the isomorphism \eqref{eq:2:iso_D_C_sc}, \eqref{sympl:2:D} becomes:
\begin{eqnarray}\label{sympl:2:C_sc}
&\sigma: \C^k_{sc}(M;\Z)\times \C^k_{sc}(M;\Z) \to \T &\nonumber\\
&\left((h,\tilde{h}),(h',\tilde{h}') \right)\mapsto \pa{I(h,\tilde{h})}{(h',\tilde{h}')}.    &
\end{eqnarray}
Antisymmetry is a straightforward check:
\begin{equation}
\begin{split}
\pa{I(h,\tilde{h})}{(h',\tilde{h}')}&= \left( I\tilde{h}\cdot h' -(-1)^{k(m-k)} Ih\cdot \tilde{h}'\right)\mu\\
&= \left( (-1)^{k(m-k)} Ih'\cdot \tilde{h} - I\tilde{h}'\cdot h\right) \mu\\
&=- \left(I\tilde{h}'\cdot h - (-1)^{k(m-k)} Ih'\cdot h \right)\mu\\
&=- \pa{I(h',\tilde{h}')}{(h,\tilde{h})}.
\end{split}
\end{equation}
In doing the passages, we have used \eqref{pairing:2:config_gauge_fields} and Lemma \ref{lem:2:grading_Ih_h'}. The pull-back $\iota^\ast_\Sigma$ of differential characters to a Cauchy surface $\Sigma$ has been understood. 

The radical of $\sigma$ coincides with the kernel of $I$, because the pairing \eqref{pairing:2:config_gauge_fields} is weakly non-degenerate (see Corollary \ref{cor:2:non_deg_pair_c}). Then, the radical of $\tau$ is given by $\text{ker}(I\circ\mathcal{O}^{-1})=\mathcal{O}(\text{ker}(I))$.
\end{proof}

\begin{prop}
The pre-symplectic structure $(\mathfrak{D}^k(M;\Z),\tau)$ is natural.
\end{prop}
\begin{proof}
We have to prove that the diagram:
\begin{equation*}
\xymatrix{
{\mathfrak{D}^k(M;\Z)\times \mathfrak{D}^k(M;\Z)} \ar[rd]^-{\tau} \ar[dd]_-{f_\ast \times f_\ast} &  \\
  & \T \\
{\mathfrak{D}^k(M';\Z)\times \mathfrak{D}^k(M';\Z)} \ar[ru]_-{\tau} & \\
}
\end{equation*}
is commutative for every morphism $f:M\to M'$ in $\mathsf{Loc}_m$. Recalling that $\mathcal{O}$ is a natural isomorphism, it is enough to prove that the diagram:
\begin{equation*}
\xymatrix{
\mathfrak{C}^k_{sc}(M;\Z)\times \mathfrak{C}^k_{sc}(M;\Z) \ar[dr]^-{\sigma} \ar[dd]_-{f_\ast \times f_\ast} & \\
 & \T \\
\mathfrak{C}^k_{sc}(M';\Z)\times \mathfrak{C}^k_{sc}(M';\Z) \ar[ur]_-{\sigma} &
}
\end{equation*}
commutes. By the naturality of \eqref{pairing:2:config_gauge_fields}, we obtain:
\begin{equation}
\begin{split}
\sigma(f_\ast(h,\tilde{h}),f_\ast(h',\tilde{h}'))= & \pa{If_\ast(h,\tilde{h})}{f_\ast(h',\tilde{h}')}\\
= & \pa{f^\ast If_\ast(h,\tilde{h})}{(h'\tilde{h}')},
\end{split}
\end{equation}
for all $(h,\tilde{h}),(h',\tilde{h}')\in\C^k_{sc}(M;\Z)$. We need to show that $f^\ast\circ I\circ f_\ast=I$. Let $(\eta,\tilde{\eta})$ be an arbitrary element in $\C^{k}_{sc}(M;\Z)$ and let $(\eta,\tilde{\eta})\in\C^{k}(M,M\setminus J(K);\Z)$ be a representative in the colimit, denoted by the same symbol, for some Cauchy surface $\Sigma\subseteq M$ and some compact set $K\subseteq \Sigma$. Since $I$ is natural, we have:
$$
f^\ast\circ I=I\circ f^\ast: \C^k(M',M'\setminus J(f(K));\Z)\to \C^k(M;\Z).
$$
Hence, recalling the proof of Proposition \ref{prop:2:diff_coho_comp_is_functor}, we obtain:
$$
f^\ast I (f^\ast)^{-1}(\eta,\tilde{\eta})=I f^\ast (f^\ast)^{-1} (\eta,\tilde{\eta})= I(\eta,\tilde{\eta}).
$$
This concludes the proof.
\end{proof}

Due to the above results, we can interpret
\begin{equation*}
(\mathfrak{D}^k(\arg;\Z),\tau): \mathsf{Loc}_m \to \mathsf{PSAb}
\end{equation*}
as a covariant functor assigning to each spacetime manifold $M$ in $\mathsf{Loc}_m$ the pre-symplectic Abelian group $(\mathfrak{D}^k(M;\Z),\tau)$ in $\mathsf{PSAb}$ and to each morphism $f:M\to M'$ in $\mathsf{Loc}_m$ the morphism $(\mathfrak{D}^k(f;\Z),\tau)=:f_\ast:(\mathfrak{D}^k(M;\Z),\tau)\to (\mathfrak{D}^k(M';\Z),\tau)$ in $\mathsf{PSAb}$. Such a functor fulfils analogues of the causality axiom and the time-slice axiom introduced by Brunetti, Fredenhagen and Verch \cite{BFV03}:

\begin{prop}
\begin{enumerate}[(i)]
\item Let $f_1: M_1\to M$ and $f_2:M_2\to M$ be morphisms in $\mathsf{Loc}_m$ such that $J(f_1(M_1))\cap f_2(M_2)=\emptyset$. Then $\tau (\varphi, \varphi ')=0$ for all $(\varphi,\varphi ')\in {f_1}_\ast(\mathfrak{D}^k(M_1;\Z))\times {f_2}_\ast (\mathfrak{D}^k(M_2;\Z))$; 
\item Let $f:M\to M'$ be a morphism in $\mathsf{Loc}_m$ such that there exists a Cauchy surface $\Sigma '\subseteq f(M)\subseteq M'$. Then $\mathfrak{D}^k(f;\Z)$ is an isomorphism.
\end{enumerate}
\end{prop}
\begin{proof}
(i). Again, we resort to the isomorphism $\mathcal{O}$ and we prove the analogous statement for $\C^k_{sc}(M;\Z)$ and $\sigma$. Let $(h,\tilde{h})\in\C^k_{sc}(M_1;\Z)$ and $(h',\tilde{h} ')\in\C^k_{sc}(M_2;\Z)$. We have:
\begin{equation*}
\sigma ({f_1}_\ast (h,\tilde{h}),{f_2}_\ast (h',\tilde{h} '))=\pa{I {f_1}_\ast (h,\tilde{h})}{ {f_2}_\ast (h',\tilde{h} ')}=\pa{f_2^\ast I {f_1}_\ast (h,\tilde{h})}{(h',\tilde{h} ')}.
\end{equation*}
Let $\Sigma$ be a Cauchy surface in $M_1$ and find $K\subseteq \Sigma$ such that $(h,\tilde{h})\in \C^k(M_1,M_1\setminus J(K);\Z)$ is a representative for $(h,\tilde{h})\in\C^k_{sc}(M_1;\Z)$. Besides, let $(f_1^\ast)^{-1}(h,\tilde{h})\in \C^k(M,M\setminus J(f_1(K));\Z)$ be a representative for ${f_1}_\ast (h,\tilde{h})\in \C^k_{sc}(M;\Z)$. By construction, the push-forward ${f_2}_\ast (z,\tilde{z}) \in Z_{k-1,m-k-1}(M)$ gives cycles supported in $f_2(M_2)$, for every pair $(z,\tilde{z})\in Z_{k-1,m-k-1}(M_2)$. By assumption, $f(M_2)\subseteq M \setminus J(M_1) \subseteq M\setminus J(f_1 (K))$. Hence, ${f_2}_\ast (z,\tilde{z})$ is a representative of $0$ in $Z_{k-1,m-k-1}(M,M\setminus J(f_1(K));\Z)$. We deduce that $f_2^\ast I (f_1^\ast)^{-1}(h,\tilde{h})=0$ and, consequently, that $f_2^\ast I {f_1}_\ast (h,\tilde{h})=0$.

(ii). The preimage $f^{-1}(\Sigma ')=:\Sigma$ is a Cauchy surface for $M$. Consider the commutative diagram:
\begin{equation*}
\xymatrix{
\C^k_{sc}(M;\Z) \ar[r]^-{f_\ast} \ar[d]_-{\iota^\ast_\Sigma \times \iota^\ast_\Sigma} & \C_{sc}^k(M';\Z) \ar[d]^-{\iota^\ast_{\Sigma'} \times \iota^\ast_{\Sigma'}}\\
\hat{H}^{k,m-k}_{c}(\Sigma;\Z) \ar[r]_-{{f_\Sigma}_\ast \times {f_\Sigma}_\ast} & \hat{H}^{k,m-k}_{c}(\Sigma ';\Z).
}
\end{equation*}
The vertical arrows are isomorphisms by Theorem \ref{thm:2:isomorphism_i_C_sc}. As for the bottom row, observe that the restriction of $f$ to $\Sigma$, $f_\Sigma: \Sigma \to \Sigma '$, is an isometry and preserves the orientation. The induced push-forward is, therefore, an isomorphism too. This entails that also the top horizontal arrow must be an isomorphism. Resorting to $\mathcal{O}$, we find that $\mathfrak{D}^k(f;\Z):\mathfrak{D}^k(M;\Z)\to \mathfrak{D}^k(M';\Z)$ is an isomorphism, as sought.
\end{proof}

Our goal is, as stated before, to construct a covariant quantum field theory. Following \cite{MAN73} and \cite{BDHS14}, the next step consists in building a covariant functor:
\begin{equation}\label{eq:2:CCR_functor}
\mathfrak{CCR}: \mathsf{PSAb} \to \mathsf{C^\ast Alg}
\end{equation}
called \emph{CCR-functor}, which assigns a $C^\ast$-algebra $\mathfrak{CCR}(\mathfrak{D}^k(M,\Z),\tau)$ to each pre-symplectic Abelian group $(\mathfrak{D}^k(M;\Z),\tau)$. In doing this, particular attention must be devoted to mathematical details, as the groups we are dealing with are not symplectic, but rather pre-symplectic, i.e.\ possibly with degeneracies.

Let $M$ be an object in $\mathsf{Loc}_m$ and let $(\mathfrak{D}^k(M;\Z),	\tau)$ be the associated pre-symplectic structure. Define the $\mathbb{C}$-vector space:
\begin{equation*}
\mathcal{A}_0(\mathfrak{D}^k(M;\Z),\tau)=\text{span}_{\mathbb{C}}\left\{\mathcal{W}(\varphi)\,\,\vline \,\, \varphi\in\mathfrak{D}^k(M;\mathbb{Z})\right\},
\end{equation*}
where $\{\mathcal{W}(\varphi), \varphi \in \mathfrak{D}^k(M;\Z)\}$ are abstract symbols.
$\mathcal{A}_0(\mathfrak{D}^k(M;\Z),\tau)$ can be given an associative, unital algebra structure by:
\begin{equation*}
\mathcal{W}(\varphi)\mathcal{W}(\tilde{\varphi}):=\exp\left(2\pi i \tau (\varphi,\tilde{\varphi})\right) \mathcal{W}(\varphi+\tilde{\varphi}),
\end{equation*}
for all $\varphi,\tilde{\varphi}\in\mathfrak{D}^k(M;\Z)$, with unit element $\mathcal{W}(0)$.
Furthermore, introduce the map $^\ast: \mathcal{A}_0(\mathfrak{D}^k(M;\Z),\tau)\to \mathcal{A}_0(\mathfrak{D}^k(M;\Z),\tau)$ defined by $\mathcal{W}(\varphi)^\ast:=\mathcal{W}(-\varphi)$ for every $\varphi\in\mathfrak{D}^k(M;\Z)$. The automorphism $^\ast$ endows the unital algebra with an involution, turning $\mathcal{A}_0(\mathfrak{D}^k(M;\Z),\tau)$ into a unital $\ast$-algebra. Every element $a\in\mathcal{A}_0(\mathfrak{D}^k(M;\Z),\tau)$ can be written as $a=\sum_i^N \alpha_i \mathcal{W}(\varphi_i)$, for some $N\in\mathbb{N}$, with $\alpha_i \in \mathbb{C}$ and $\varphi_i\in \mathfrak{D}^k(M;\Z)$ for all $i=1,\dots,N$. Without loss of generality, we assume $\varphi_i\neq \varphi_j$ whenever $i\neq j$.

Let $\psi: (\mathfrak{D}^k(M;\Z),\tau) \to (\mathfrak{D}^k(M';\Z),\tau')$ be a morphism in $\mathsf{PSAb}$. Define the homomorphism:
\begin{eqnarray*}
&\mathcal{A}_0(\psi): \mathcal{A}_0(\mathfrak{D}^k(M;\Z),\tau) \to \mathcal{A}_0(\mathfrak{D}^k(M';\Z),\tau') &\\
& a=\sum_{i=1}^N \alpha_i \mathcal{W}(\varphi_i) \mapsto \mathcal{A}_0(\psi)(a):= \sum_{i=1}^N \alpha_i \mathcal{W}'(\psi(\varphi_i)). &
\end{eqnarray*}
As $\psi$ is a group homomorphism preserving the pre-symplectic structure, $\mathcal{A}_0(\psi)$ is a $\ast$-algebra homomorphism. It is straightforward to check that $\mathcal{A}_0(\text{Id}_{(\mathfrak{D}^k(M;\Z),\tau)})=\text{Id}_{(\mathfrak{D}^k(M';\Z),\tau')}$ and that, given a second morphism $\psi': (\mathfrak{D}^k(M';\Z),\tau') \to (\mathfrak{D}^k(M'';\Z),\tau'')$, it holds $\mathcal{A}_0(\psi '\circ \psi)=\mathcal{A}_0(\psi')\circ \mathcal{A}_0(\psi)$. Therefore, we have obtained that:
\begin{equation}
\mathcal{A}_0: \mathsf{PSAb}\to \mathsf{^\ast Alg}
\end{equation}
is a covariant functor.

As an intermediate step, let us endow $\mathcal{A}_0(\mathfrak{D}^k(M;\Z),\tau)$ with the structure of a Banach $\ast$-algebra. Define the norm:
\begin{eqnarray*}
&\norm{\cdot}^{\text{Ban}}:\mathcal{A}_0(\mathfrak{D}^k(M;\Z),\tau) \to \mathbb{R}_{\geq 0}&\\
&a= \sum_{i=1}^N \alpha_i \mathcal{W}(\varphi_i)\mapsto \norm{a}^{\text{Ban}}=\sum_{i=1}^N |\alpha_i|.&
\end{eqnarray*}
The completion $\mathcal{A}_1(\mathfrak{D}^k(M;\Z),\tau):=\overline{\mathcal{A}_0(\mathfrak{D}^k(M;\Z),\tau)}^{\norm{\cdot}^{\text{Ban}}}$ is a unital Banach *-algebra. An arbitrary element in $\mathcal{A}_1(\mathfrak{D}^k(M;\Z),\tau)$ can be written as $a=\sum_{i=1}^\infty \alpha_i \mathcal{W}(\varphi_i)$, with $\alpha_i\in\mathbb{C}$, $\sum_{i=1}^\infty |\alpha_i| < +\infty$. Given a morphism $\psi: (\mathfrak{D}^k(M;\Z),\tau) \to (\mathfrak{D}^k(M';\Z),\tau')$ in $\mathsf{PSAb}$, it holds $\norm{\mathcal{A}_0(\psi)(a)}^{\text{Ban}\,\prime}\leq \norm{\mathcal{A}_0(\psi)}^{\text{Ban}}$. By Hahn-Banach theorem, $\mathcal{A}_0(\psi)$ admits a unique continuous extension:
\begin{equation*}
\mathcal{A}_1(\psi): \mathcal{A}_1(\mathfrak{D}^k(M;\Z),\tau) \to \mathcal{A}_1(\mathfrak{D}^k(M';\Z),\tau').
\end{equation*}
It is a straightforward check that:
\begin{equation*}
\mathcal{A}_1: \mathsf{PSAb} \to \mathsf{B^\ast Alg}
\end{equation*}
fulfils the requirements for it to be a functor from the category $\mathsf{PSAb}$ to the category of Banach $\ast$-algebras $\mathsf{B^\ast Alg}$, whose morphisms are unital $\ast$-Banach algebra homomorphisms.\\
Any positive linear functional on $\mathcal{A}_0(\mathfrak{D}^k(M;\Z),\tau)$ extends, by \cite[Proposition 2.17]{MAN73}, to a continuous positive linear functional on the algebra $\mathcal{A}_1(\mathfrak{D}^k(M;\Z),\tau)$. Consider the functional:
\begin{eqnarray*}
&\omega: \mathcal{A}_0(\mathfrak{D}^k(M;\Z),\tau) \to \mathbb{C}&\\
&\mathcal{W}(\varphi)\mapsto 0, \quad \varphi \neq 0&\\
&\mathcal{W}(0)\mapsto  \omega(\mathcal{W}(0))=\omega(\mathbf{1})=1.&
\end{eqnarray*}
For what just said, $\omega$ can be promoted to a state on $\mathcal{A}_1(\mathfrak{D}^k(M;\Z),\tau)$. Such a state, still denoted by $\omega$, is faithful, viz $\omega(a^\ast a)>0$ for every $a\in \mathcal{A}_1(\mathfrak{D}^k(M;\Z),\tau)$, $a\neq 0$. Therefore, we are granted the existence of at least one faithful state on $\mathcal{A}_1(\mathfrak{D}^k(M;\Z),\tau)$; this entails that the following $C^\ast$-norm is meaningful:

\begin{defn}
Denoting by $\mathcal{F}$ the set of states on $\mathcal{A}_1(\mathfrak{D}^k(M;\Z),\tau)$, i.e.\ of normalized positive continuous linear functionals, we call \emph{minimal regular norm} on $\mathcal{A}_1(\mathfrak{D}^k(M;\Z),\tau)$ the norm defined by:
\begin{eqnarray*}
& \norm{\arg}: \mathcal{A}_1(\mathfrak{D}^k(M;\Z),\tau) \to \mathbb{R}_{\geq 0}&\\
&a \mapsto \norm{a}:= \sup _{\omega \in \mathcal{F}} \sqrt{\omega(a^\ast a)}.&
\end{eqnarray*}
\end{defn}
We eventually define our algebra of observables to be the completion 
$$
\mathfrak{CCR}(\mathfrak{D}^k(M;\Z),\tau):=\overline{\mathcal{A}_1(\mathfrak{D}^k(M;\Z),\tau)}^{\norm{\cdot}}.
$$
It turns out that $\mathfrak{CCR}(\mathfrak{D}^k(M;\Z),\tau)$ is a unital $C^\ast$-algebra and that the assignment
\begin{eqnarray}\label{eq:2:CCR_functor}
&\mathfrak{CCR}: \mathsf{PSAb}\to \mathsf{C^\ast Alg} &\nonumber\\
&(\mathfrak{D}^k(M;\Z),\tau)\mapsto \mathfrak{CCR}(\mathfrak{D}^k(M;\Z),\tau)&
\end{eqnarray}
defines a covariant functor (cfr. \cite[Appendix A]{MAN73, BDHS14}).

We now have all the ingredients for our recipe. We define a family of covariant quantum field theories, labelled by the degree $k\in\{1,\dots, m-1\}$ of the gauge theory, by the composition:
\begin{equation}\label{eq:2:LCQFT}
\mathfrak{A}^k:=\mathfrak{CCR}\circ (\mathfrak{D}^k(\arg;\Z),\tau): \mathsf{Loc}_m \to \mathsf{C^\ast Alg}.
\end{equation}
$\mathfrak{A}^k$ inherits the properties of causality and time-slice axiom from $(\mathfrak{D}^k(\arg,\Z),\tau)$ (see \cite[Theorem 5.1]{BPhys}).

\chapter{States for Differential Cohomology QFT}\label{chapter:3}
This chapter comprises the core of the original part of our work. The concepts and tools introduced in the previous chapters will constitute the basis upon which to lay the foundation of our construction. The issue of how to build Hadamard states for covariant quantum field theory on differential cohomology will be addressed, and a thorough procedure will be provided for spacetimes possessing a compact Cauchy surface. As far as the non-compact case is concerned, some hints will be given along with some comments.\\

In the case of compact Cauchy surfaces, the outline of our construction can be briefly sketched as follows: Firstly, considering the analogous of the commutative diagram of short exact sequences \eqref{diag:2:scconfig} with spacelike compact support, which we interpret physically as the diagram of the semi-classical observables, we endow suitable combinations of the Abelian groups appearing in it with a pre-symplectic structure. This makes it possible to activate the quantization scheme discussed above for the Abelian group $\mathfrak{D}^k(M;\Z)$, thus obtaining a collection of $C^\ast$-algebras of the observables. Secondly, we show that the pre-symplectic structure over $\C^k_{sc}(M;\Z)$ allows for a pre-symplectic decomposition of $\C^k_{sc}(M;\Z)$ in suitable pre-symplectically orthogonal pre-symplectic Abelian groups. Consequently, the relevant $C^\ast$-algebra, i.e.\ the central one, emerges as the tensor product of the $C^\ast$-algebras previously constructed, related to other elements in the diagram, and a state on $\mathfrak{CCR}(\C^k_{sc}(M;\Z),\sigma)$ can be built by assigning a product of states defined on such algebras. Eventually, we construct explicitly these auxiliary states and we prove that their singularities are of Hadamard type.

\section{Pre-symplectic structures}\label{section:3.1_presymplectic_strunctures}
So as to be a little more explicit, represent schematically diagram \eqref{diag:2:scconfig} as follows:
\begin{equation}\label{diag:3:schematics}
\xymatrix{
& 0 \ar[d] & 0\ar[d]  & 0\ar[d] &\\
0 \ar[r] & A \ar[r]\ar[d] & B \ar[r]\ar[d] & C \ar[r]\ar[d] & 0\\
0 \ar[r] & D \ar[r]\ar[d] & E \ar[r]\ar[d] & F \ar[r]\ar[d] & 0\\
0 \ar[r] & G \ar[r]\ar[d] & H \ar[r]\ar[d] & I \ar[r]\ar[d] & 0\\
&0&0&0&\\
}
\end{equation}
We aim at showing that the pre-symplectic information stored in $E$ can actually be fully recostructed out of $G$, $C$ and $A \oplus I$. Since we would like to obtain non-trivial pre-symplectic structures covering all the diagram, derived from or closely related to $(E,\sigma)$ in a sense that will be clear in a while, it turns out that not all of the Abelian groups in diagram \eqref{diag:2:scconfig} can be considered individually. Specifically, $B$ comes with $F$, $A$ comes with $I$ and $D$ comes with $H$. Therefore, the first step will be the definition of non-trivial bilinear, antisymmetric maps on $G$, $ D \times H$, $A \times I$, $B \times F$ and $C$.

For the reader's convenience, we reproduce below the complete diagram analogous to \eqref{diag:2:scconfig} with spacelike compact support:
\begin{equation}\label{diag:3:constr_pre_sympl_str}
{\scriptsize
\xymatrix{
& 0 \ar[d] & 0 \ar[d] & 0 \ar[d] &\\
0\ar[r] & \dfrac{H^{k-1,m-k-1}_{sc}(M;\mathbb{R})}{H^{k-1,m-k-1}_{sc,\free}(M;\mathbb{Z})} \ar[r]^-{\tilde{\kappa}\times \tilde{\kappa}} \ar[d]^-{\tilde{\iota}\times \tilde{\iota}} & \mathfrak{T}_{sc}^k(M;\mathbb{Z}) \ar[r]^-{d_1} \ar[d]^-{\iota \times \iota} & d\Omega_{sc}^{k-1}\cap *d\Omega_{sc}^{m-k-1}(M) \ar[r]\ar[d]^-{\subseteq}& 0\\
0\ar[r]& H^{k-1,m-k-1}_{sc}(M;\mathbb{T}) \ar[r]^-{\kappa \times \kappa} \ar[d]^-{\beta\times\beta}& \mathfrak{C}_{sc}^k(M;\mathbb{Z})  \ar[r]^-{\curv_1} \ar[d]^-{\cha \times \cha} & \Omega^{k}_{sc,\mathbb{Z}}\cap *\Omega^{m-k}_{sc,\mathbb{Z}}(M) \ar[r] \ar[d]^-{([\cdot],[*^{-1}\cdot])}& 0\\
0 \ar[r] & H^{k,m-k}_{sc,\tor}(M;\mathbb{Z}) \ar[r]^-{i\times i} \ar[d] & H^{k,m-k}_{sc}(M;\mathbb{Z}) \ar[r]^-{q\times q}\ar[d] & H^{k,m-k}_{sc,\free}(M;\mathbb{Z}) \ar[r] \ar[d] &0\\
& 0 & 0 & 0 & \\
}}
\end{equation}
The idea that guides us consists in exploiting the exactness of rows and columns in diagram \eqref{diag:3:constr_pre_sympl_str} to map or to compute pre-images of the elements of the various Abelian groups to $\C^k_{sc}(M;\Z)$ and subsequently in resorting to the pre-symplectic structure there defined to induce pre-symplectic structures elsewhere.

\begin{rem}
In the following, we will refer to $\C^k_{sc}(M;\Z)$ or $\mathfrak{D}^k(M;\Z)$ indifferently, in view of the isomorphism \eqref{eq:2:iso_D_C_sc}. Besides, we will occasionally work on an arbitrary Cauchy surface instead of considering the whole spacetime, making use of the isomorphism provided by Theorem \ref{thm:2:isomorphism_i_C_sc}.
\end{rem}

\begin{rem}
In order to make our computations more intelligible, let us introduce some lighter notation. Set:
\begin{align*}
& \i{h}:=\iota ^\ast_\Sigma h,\\
&\I{h}:=Ih,\\
& \ii{h}:= \iota^\ast_\Sigma I h.
\end{align*}
In this way, the pre-symplectic product writes:
\begin{equation*}
\sigma ((h,\tilde{h}),(h',\tilde{h}'))=\left( \ii{\tilde{h}}\cdot \i{h'} -(-1)^{k(m-k)} \ii{h}\cdot \i{\tilde{h}'} \right) \mu.
\end{equation*}
\end{rem}

\subsection*{Pre-symplectic structure on $d\Omega_{sc}^{k-1}\cap *d\Omega_{sc}^{m-k-1}(M)$}
\begin{prop}
Let $M$ be an object in $\mathsf{Loc}_m$ and $\Sigma\subseteq M$ a Cauchy surface. The bilinear map 
\begin{eqnarray}\label{tau:3:u}
&\tau_{u}:\left(d\Omega^{k-1}_{sc}\cap *d\Omega^{m-k-1}_{sc}(M)\right)\times \left(d\Omega^{k-1}_{sc}\cap *d\Omega^{m-k-1}_{sc}(M)\right) \rightarrow \mathbb{T}\nonumber &\\
&(dA=*d\tilde{A},dA'=*d\tilde{A'})\mapsto \int_\Sigma \tilde{A}\wedge dA' -(-1)^{k(m-k)}A\wedge d\tilde{A}'\,\, \emph{\mod}\mathbb{Z}&
\end{eqnarray}
defines a pre-symplectic structure on $d\Omega^{k-1}_{sc}\cap *d\Omega^{m-k-1}_{sc}(M)$.
\end{prop}

\begin{proof}
Given $dA=*d\tilde{A},dA'=*d\tilde{A'}\in d\Omega^{k-1}_{sc}\cap *d\Omega^{m-k-1}_{sc}(M)$, by resorting to diagram \eqref{diag:3:constr_pre_sympl_str} we sort out pre-images $([A],[\tilde{A}]),([A'],[\tilde{A}'])\in \mathfrak{T}^k_{sc}(M;\mathbb{Z})$. Let $(h,\tilde{h}),(h',\tilde{h}')\in \mathfrak{C}^k_{sc}(M;\mathbb{Z})$ such that $(h,\tilde{h}):=(\iota[A],\iota[\tilde{A}])$ and $(h',\tilde{h}'):=(\iota[A'],\iota[\tilde{A}'])$. Exploiting the pre-symplectic structure set on $\mathfrak{C}^k_{sc}(M;\mathbb{Z})$, define:
\begin{equation*}
\begin{split}
\tau_{u}(dA,dA'):=&\sigma((h,\tilde{h}),(h',\tilde{h}'))=\sigma ((\iota[A],\iota[\tilde{A}]),(h',\tilde{h}'))\\
=&\int_\Sigma \tilde{A} \wedge \curv h'-(-1)^{k(m-k)} A \wedge \curv \tilde{h}'\,\,\mod \Z\\
=&\int_\Sigma \tilde{A} \wedge \curv (\iota[A'])-(-1)^{k(m-k)} A \wedge \curv (\iota[\tilde{A}'])\,\,\mod \Z\\
=&\int_\Sigma \tilde{A} \wedge dA '-(-1)^{k(m-k)} A \wedge d\tilde{A}'\,\,\mod\mathbb{Z},
\end{split}
\end{equation*}
where in the second line we have used \eqref{eq:2:module_structure} and in the last line the commutativity of the upper-right square of diagram \eqref{diag:3:constr_pre_sympl_str}.

Let us check the antisymmetry property. Recalling that, given $\omega \in \Omega^p(M)$ and $\tilde{\omega}	\in\Omega_c^q(M)$, $d(\omega\wedge \tilde{\omega})=d\omega \wedge \tilde{\omega}+ (-1)^{p} \omega \wedge d\tilde{\omega}$, it ensues that:
$$
\int_{\Lambda} d\omega \wedge \tilde{\omega}=(-1)^{p+1}\int \omega\wedge d\tilde{\omega}
$$
and
$$
\int_{\Lambda} \omega \wedge d\tilde{\omega}=(-1)^{-p+1}\int d\omega\wedge \tilde{\omega}
$$
for every smooth $(p+q+1)$-dimensional submanifold without boundary $\Lambda\subset M$. Therefore:
\begin{equation*}
\begin{split}
\tau_{u}(dA,dA')=&\int_{\Sigma} (-1)^{(m-k)}d\tilde{A}\wedge A'-(-1)^{k(m-k)}(-1)^{k}dA\wedge \tilde{A}'\,\,\mod \Z\\
=&\int_{\Sigma} (-1)^{(m-k)} (-1)^{(k-1)(m-k)}A'\wedge d\tilde{A}+\\
&-(-1)^{k(m-k)}(-1)^{k+k(m-k-1)}\tilde{A}'\wedge dA\,\,\mod \Z\\
=&\int_{\Sigma} (-1)^{k(m-k)}A'\wedge d\tilde{A} - (-1)^{2k(m-k)}\tilde{A}'\wedge dA\,\,\mod \Z\\
=&-\int_{\Sigma} \tilde{A}'\wedge dA - (-1)^{k(m-k)}A'\wedge d\tilde{A}\,\,\mod \Z\\
=& -\tau(dA',dA).
\end{split}
\end{equation*}
What is left to prove is the independence of the above construction from the choice of the representatives in the cosets. Since the freedom in the choice of representatives corresponds to closed forms and the symplectic product formula involves an exterior derivative which can act equivalently - up to a sign - on either members of the wedge product via an integration by parts, we immediately conclude. 
 \end{proof}

\begin{rem}
We would like to point out that the map $\tau_u$ is \emph{always} weakly non-degenerate, also in the case of non-compact Cauchy surface; therefore, it gives rise to a \emph{symplectic} structure. This fact provides a further justification for the approach we adopted in our attempt at separating the topological degrees of freedom from the dynamical ones:  as the latter always show a good behaviour, the possible degeneracy of the pre-symplectic form $\sigma$ in the general case is only due to the former.
\end{rem}

\subsection*{Pre-symplectic structure on {\small$\mathfrak{T}_{sc}^k(M;\mathbb{Z})\times \Omega^{k}_{sc,\mathbb{Z}}\cap *\Omega^{m-k}_{sc,\mathbb{Z}}(M)$}}
As prospected, $\mathfrak{T}_{sc}^k(M;\mathbb{Z})$ and $\Omega^{k}_{sc,\mathbb{Z}}\cap *\Omega^{m-k}_{sc,\mathbb{Z}}(M)$ have to be dealt with together. It is their intertwining that gives rise to an interesting theory.
 
\begin{prop}
Let $M$ be an object in $\mathsf{Loc}_m$ and let $\Sigma \subseteq M$ be a Cauchy surface. The bilinear map
\begin{align}\label{tau:3:ur}
\tau_{ur}: \left( \mathfrak{T}^k_{sc}(M) \times  \Omega^{k}_{sc, \mathbb{Z}}\cap *\Omega^{m-k}_{sc,\mathbb{Z}}(M) \right)\times \left( \mathfrak{T}^k_{sc}(M)  \times  \Omega^{k}_{sc, \mathbb{Z}}\cap *\Omega^{m-k}_{sc,\mathbb{Z}}(M) \right) \to \mathbb{T} \nonumber \\
(([A],[\tilde{A}],\omega=*\tilde{\omega}),([A'],[\tilde{A}'],\omega '=*\tilde{\omega}'))\mapsto \int_\Sigma \left( \tilde{A} \wedge \omega ' -(-1)^{k(m-k)} A\wedge \tilde{\omega}'\right. \nonumber\\
\left. -\tilde{A}'\wedge \omega+(-1)^{k(m-k)}A'\wedge \tilde{\omega} \right) \,\,\emph{\mod} \Z
\end{align}
defines a pre-symplectic structure on $\mathfrak{T}^k_{sc}(M) \times  \Omega^{k}_{sc, \mathbb{Z}}\cap *\Omega^{m-k}_{sc,\mathbb{Z}}(M)$.
\end{prop}
\begin{proof}
Let $$([A],[\tilde{A}],\omega=*\tilde{\omega}),([A'],[\tilde{A}'],\omega '=*\tilde{\omega}')\in \mathfrak{T}^k_{sc}(M) \times  \Omega^{k}_{sc, \mathbb{Z}}\cap *\Omega^{m-k}_{sc,\mathbb{Z}}(M).$$ By diagram \eqref{diag:3:constr_pre_sympl_str}, there exist $(h,\tilde{h}),(h',\tilde{h}')\in \C^k_{sc}(M;\mathbb{Z})$ such that
\begin{align*}
&\curv h=\omega=*\tilde{\omega}=*\curv \tilde{h},\\
&\curv h'=\omega'=*\tilde{\omega}'=*\curv \tilde{h}'.
\end{align*}
Define:
\begin{align*}
\tau_{ur}(([A],[\tilde{A}],\omega= & *\tilde{\omega}),([A'], [\tilde{A}'], \omega '=*\tilde{\omega}')):=\\
=&\sigma((\iota[A],\iota[\tilde{A}]),(h',\tilde{h}'))+\sigma((h,\tilde{h}), (\iota [A'],\iota [\tilde{A}']))\\
=& \left( \ii{\iota [\tilde{A}]}\cdot \i{h'} -(-1)^{k(m-k)}\ii{\iota[A]}\cdot \i{\tilde{h}'}+\right.\\
&\left.+ \ii{\tilde{h}}\cdot \i{\iota [A']} -(-1)^{k(m-k)} \ii{h}\cdot \i{\iota [\tilde{A}']}\right) \mu\\
=& \left( \ii{\iota [\tilde{A}]}\cdot \i{h'} -(-1)^{k(m-k)}\ii{\iota[A]}\cdot \i{\tilde{h}'} +\right. \\
&\left. -(\i{\iota [\tilde{A}']}\cdot \ii{h} -(-1)^{k(m-k)} \i{\iota [A']}\cdot \ii{\tilde{h}} )\right)\mu\\
=& \int_\Sigma \left(\tilde{A}\wedge \curv\i{h'}-(-1)^{k(m-k)}A\wedge\curv\i{\tilde{h}'} +\right.\\
&-\left. \tilde{A}'\wedge \curv\ii{h}+(-1)^{k(m-k)}A'\wedge \curv\ii{\tilde{h}} \right)\,\,\mod \Z\\
=& \int_\Sigma \left(\tilde{A}\wedge \omega'-(-1)^{k(m-k)}A\wedge\tilde{\omega}'-\tilde{A}'\wedge \omega\right.\\
&\left.+(-1)^{k(m-k)}A'\wedge \tilde{\omega} \right)\,\,\mod \Z.\\
\end{align*}
Notice that for the differential forms the homomorphism $I$ is nothing but an inclusion, which has been understood. A swap of the arguments of $\tau_{ur}$ is equivalent to separate swaps of the arguments of each $\sigma$; therefore, antisymmetry ensues.

We will now prove that the definition is independent of the choices we have made. $(h,\tilde{h}),(h',\tilde{h}')$ are defined up to elements in the image of $\kappa\times \kappa$. We are allowed to replace them with:
\begin{equation*}
(h,\tilde{h})+\kappa\times \kappa (z,\tilde{z}), \quad \qquad  (h',\tilde{h}')+ \kappa\times\kappa (z',\tilde{z}'),
\end{equation*}
for some $(z,\tilde{z}),(z',\tilde{z}')\in H_{sc}^{k-1,m-k-1}(M;\mathbb{T})$. However, this brings about the appearance of vanishing extra-terms. For instance, as far as the the first summand is concerned we obtain:
\begin{equation*}
\begin{split}
\ii{\iota[\tilde{A}]}\cdot (\i{h'}+\i{\kappa(z')})&=\ii{\iota[\tilde{A}]}\cdot \i{h'}+ \ii{\iota[\tilde{A}]}\cdot\i{\kappa(z)}\\
&=\ii{\iota[\tilde{A}]}\cdot \i{h'}+ \i{\iota(\tilde{A} \wedge \curv\circ\kappa(z))}\\
&= \ii{\iota[\tilde{A}]}\cdot \i{h'}.
\end{split}
\end{equation*}
where in the last passage we used $\curv \circ \kappa = 0$ and \eqref{eq:2:module_structure}. 
\end{proof}

\subsection*{Pre-symplectic structure on {\small$\frac{H^{k-1,m-k-1}_{sc}(M;\mathbb{R})}{H^{k-1,m-k-1}_{sc,\mathrm{free}}(M;\mathbb{Z})}\times H^{k,m-k}_{sc,\mathrm{free}}(M;\mathbb{Z})$}}
As before, $\frac{H^{k-1,m-k-1}_{sc}(M;\mathbb{R})}{H^{k-1,m-k-1}_{sc,\mathrm{free}}(M;\mathbb{Z})}$ and $H^{k,m-k}_{sc,\mathrm{free}}(M;\mathbb{Z})$ have to be considered together. Let $M$ be an object in $\mathsf{Loc}_m$ and $\Sigma\subseteq M$ a Cauchy surface. We aim at obtaining a pre-symplectic map:

{\scriptsize
\begin{equation*}
\tau_{lr}: \left(\frac{H^{k-1,m-k-1}_{sc}(M;\mathbb{R})}{H^{k-1,m-k-1}_{sc,\free}(M;\mathbb{Z})}\times H^{k,m-k}_{\free}(M;\mathbb{Z}) \right)\times \left( \frac{H^{k-1,m-k-1}_{sc}(M;\mathbb{R})}{H^{k-1,m-k-1}_{sc,\free}(M;\mathbb{Z})}\times H^{k,m-k}_{\free}(M;\mathbb{Z}) \right) \to \mathbb{T}
\end{equation*}
} \vspace{-4mm}
\begin{equation*}
 \left((u,\tilde{u},v,\tilde{v}),(u',\tilde{u}',v',\tilde{v}') \right) \mapsto \tau_{lr}((u,\tilde{u},v,\tilde{v}),(u',\tilde{u}',v',\tilde{v}')).
\end{equation*}

To this end, let $(u,\tilde{u},v,\tilde{v}),(u',\tilde{u}',v',\tilde{v}')\in \frac{H^{k-1,m-k-1}_{sc}(M;\mathbb{R})}{H^{k-1,m-k-1}_{sc,\free}(M;\mathbb{Z})}\times H^{k,m-k}_{\free}(M;\mathbb{Z})$. Realising $\frac{H^{k-1,m-k-1}_{sc}(M;\mathbb{R})}{H^{k-1,m-k-1}_{sc,\free}(M;\mathbb{Z})}\subset \frac{\Omega_{sc}^{k-1,m-k-1}}{\Omega_{sc,\mathbb{Z}}^{k-1,m-k-1}}(M)$ via de Rham theorem, define $([U],[\tilde{U}]),([U'],[\tilde{U}'])\in \frac{\Omega_{sc}^{k-1,m-k-1}}{\Omega_{sc,\mathbb{Z}}^{k-1,m-k-1}}(M)$ by:
\begin{equation*}
\begin{split}
&([U],[\tilde{U}]):=\tilde{\kappa}\times \tilde{\kappa}(u,\tilde{u})\\
&([U'],[\tilde{U}']):=\tilde{\kappa}\times \tilde{\kappa}(u',\tilde{u}').
\end{split}
\end{equation*}
Then, resorting to diagram \eqref{diag:3:constr_pre_sympl_str}, choose $V=*\tilde{V},V'=*\tilde{V}'\in \Omega_{sc,\mathbb{Z}}^{k}\cap *\Omega_{sc,\mathbb{Z}}^{m-k}(M)$ such that:
\begin{equation*}
\begin{split}
&([V],[*^{-1}V])=(v,\tilde{v})\\
&([V'],[*^{-1}V'])=(v',\tilde{v}').
\end{split}
\end{equation*}
This way, with reference to diagram \eqref{diag:3:schematics}, we have assigned a pair in $B \times F$ to a pair in $A\times I$. Making use of the results of the previous section, we define:
\begin{equation}\label{tau:3:lr}
\begin{split}
\tau_{lr}((u,\tilde{u},v,\tilde{v}),(u',\tilde{u}',v',\tilde{v}')):=\int_\Sigma & \left( \tilde{U} \wedge V'-(-1)^{k(m-k)}U \wedge \tilde{V} '\right.\\
& \left. -\tilde{U}'\wedge V +(-1)^{k(m-k)} U'\wedge \tilde{V} \right) \,\,\mod \mathbb{Z}.
\end{split}
\end{equation}
Whereas elements in $B$ are uniquely identified, a freedom is available in the choice of elements in $F$. Nonetheless, the construction is independent of such choices. In fact, even though we can in principle add to $V=*\tilde{V},V'=*\tilde{V}'$ elements in $d\Omega_{sc}^{k-1}\cap *d\Omega_{sc}^{m-k-1}(M)$, the extra-terms vanish upon integration by parts. Consider, for instance, the first summand and set:
$$
\hat{V'}=V'+d\psi,
$$
for some $\psi\in \Omega_{sc}^{k-1}(M)$.
Then:
\begin{align*}
\int_{\Sigma} \tilde{U}\wedge \hat{V'}&=\int_\Sigma \tilde{U}\wedge V' + \int_{\Sigma} \tilde{U}\wedge d\psi\\
&=\int_\Sigma \tilde{U}\wedge V' +(-1)^{m-k} \int_{\Sigma} d\tilde{U}\wedge \psi\\
&=\int_\Sigma \tilde{U}\wedge V' +(-1)^{m-k} \int_{\Sigma} d(\tilde{k}(u))\wedge \psi\\
&=\int_\Sigma \tilde{U}\wedge V'.
\end{align*}
In the last passage we made use of $d\circ \tilde{\kappa}=0$. We end up, therefore, with a well-defined pre-symplectic structure.

\subsection*{Pre-symplectic structure on {\footnotesize $H^{k-1,m-k-1}_{sc}(M;\mathbb{T})\times H^{k,m-k}_{sc}(M;\mathbb{Z})$}}
\begin{prop}
Let $M$ be an object in $\mathsf{Loc}_m$ and $\Sigma\subseteq M$ a Cauchy surface. The bilinear map:
{\footnotesize
\begin{align*}
\tau_{lb}: \left(H^{k-1,m-k-1}_{sc}(M;\mathbb{T})\times H^{k,m-k}_{sc}(M;\mathbb{Z})\right) \times \left( H^{k-1,m-k-1}_{sc}(M;\mathbb{T})\times H^{k,m-k}_{sc}(M;\mathbb{Z}) \right)\to \mathbb{T} 
\end{align*}}
defined by:
\begin{equation}\label{tau:3:lb}
\begin{split}
\tau_{lb}((a,\tilde{a},z,\tilde{z}),(a',\tilde{a}',z',\tilde{z}') ):=\kappa &\left( \ii{\tilde{a}}\smile \i{z'} -(-1)^{k(m-k)} \ii{a} \smile \i{\tilde{z}'} \right.\\
 & \left. - \ii{\tilde{a}'} \smile \i{z} +(-1)^{k(m-k)} \ii{a'} \smile \i{\tilde{z}}  \right) \mu
\end{split}
\end{equation}
for all $(a,\tilde{a},z,\tilde{z}),(a',\tilde{a}',z',\tilde{z}')\in H^{k-1,m-k-1}_{sc}(M;\mathbb{T})\times H^{k,m-k}_{sc}(M;\mathbb{Z})$, establishes a pre-symplectic structure on $H^{k-1,m-k-1}_{sc}(M;\mathbb{T})\times H^{k,m-k}_{sc}(M;\mathbb{Z})$.
\end{prop}

\begin{proof}
In view of the surjectivity of $\cha \times \cha$ in diagram \eqref{diag:3:constr_pre_sympl_str}, it is possible to find $(h,\tilde{h}),(h',\tilde{h}')\in \mathfrak{C}^k(M;\mathbb{Z})$ such that:
\begin{equation}\label{eq:3:passage1}
\begin{split}
&(z,\tilde{z})=\cha\times \cha(h,\tilde{h})\\
&(z',\tilde{z}')=\cha\times \cha(h',\tilde{h}').
\end{split}
\end{equation}
Then define:
\begin{align*}
&\tau_{lb}((a,\tilde{a},z,\tilde{z}),(a',\tilde{a}',z',\tilde{z}')):= \sigma((\kappa a, \kappa \tilde{a}),(h',\tilde{h}')) + \sigma((h,\tilde{h}),(\kappa a', \kappa \tilde{a}'))\\
=& \left[ \ii{\kappa \tilde{a}} \cdot \i{h'}-(-1)^{k(m-k)}\ii{\kappa a} \cdot \i{\tilde{h}'} + \ii{\tilde{h}} \cdot \i{\kappa a'} -(-1)^{k(m-k)}\ii{h}\cdot \i{\kappa \tilde{a}'}      \right]\mu\\
=& \left[ {\kappa(\ii{\tilde{a}} \smile \i{\cha h'} )} -(-1)^{k(m-k)}{\kappa ( \ii{a} \smile \i{\cha \tilde{h}'}	)}\right.\\
&\left. -( {\kappa(\ii{\tilde{a}'} \smile \i{\cha h})} -(-1)^{k(m-k)} {\kappa( \ii{a'} \smile \i{\cha \tilde{h}})} \right]\mu\\
=& {\kappa ( \ii{\tilde{a}}\smile \i{z'} -(-1)^{k(m-k)} \ii{a} \smile \i{\tilde{z}'} - \ii{\tilde{a}'} \smile \i{z} +(-1)^{k(m-k)} \ii{a'} \smile \i{\tilde{z}}  )}\mu.
\end{align*}

A swap of the arguments of $\tau$ is equivalent to separate swaps of the arguments of each $\sigma$;  therefore, antisymmetry ensues as a consequence of the antisymmetry of $\sigma$.

As far as the independence of the construction on the made choices is concerned, observe that the differential characters $(h,\tilde{h}),(h',\tilde{h}')$ given by \eqref{eq:3:passage1} are fixed up to elements in the image of $\iota\times \iota$. We are allowed to replace them with:
\begin{equation*}
(h+\iota[A], \tilde{h}+\iota[\tilde{A}]), \qquad (h'+\iota[A'], \tilde{h}'+\iota[\tilde{A}']),
\end{equation*}
for some $([A],[\tilde{A}]),([A'],[\tilde{A}'])\in \mathfrak{T}^k_{sc}(M;\mathbb{Z})$.
Once more, the extra-terms are vanishing. Consider, for example, the first summand:
\begin{equation*}
\begin{split}
\ii{\kappa \tilde{a}} \cdot \i{(h'+\iota[A'])}=& \ii{\kappa \tilde{a}} \cdot \i{h'} + \ii{\kappa \tilde{a}} \cdot \i{\iota[A']}\\
=& \ii{\kappa \tilde{a}} \cdot \i{h'} + (-1)^{m-k}\i{\iota ( \curv\circ \kappa (\tilde{a}) \wedge A' )}\\
=& \ii{\kappa \tilde{a}} \cdot \i{h'} .
\end{split}
\end{equation*}
In the last passages we resorted to \eqref{eq:2:module_structure} and to $\curv \circ \kappa=0$.
\end{proof}

\subsection*{Pre-symplectic structure on {\small$H_{sc,\tor}^{k,m-k}(M;\mathbb{Z})$}}
Finally, let us consider the torsion subgroup $H^{k,m-k}_{sc,\tor}(M;\Z)$. We procede as follows: Let $(t,\tilde{t}),(t',\tilde{t}')\in H^{k,m-k}_{sc,\tor}(M;\Z)$.  Since the map $\beta \times \beta$ in diagram \eqref{diag:3:schematics} is surjective, we can find $(z,\tilde{z}),(z',\tilde{z}')\in H_{sc}^{k-1,m-k-1}(M;\mathbb{T})$ such that:
\begin{equation*}
\begin{split}
& \beta \times \beta (z,\tilde{z})=(t,\tilde{t})\\
& \beta \times \beta (z',\tilde{z}')=(t',\tilde{t}').
\end{split}
\end{equation*}
Define the bilinear map:
\begin{eqnarray}\label{tau:3:d}
& \tau_{d}: H_{sc,tor}^{k,m-k}(M;\mathbb{Z})\times  H_{sc,tor}^{k,m-k}(M;\mathbb{Z})\to \mathbb{T} &\nonumber\\
&(t,\tilde{t}),(t',\tilde{t}') \mapsto \left( \ii{\kappa \tilde{z}}\cdot \i{\kappa z'}-(-1)^{k(m-k)}\ii{\kappa z}\cdot \i{\kappa \tilde{z}'}\right) \mu &.
\end{eqnarray}
It enjoys the antisymmetry property, as a straightforward consequence of the fact that:
$$
\tau_d ( (t,\tilde{t}),(t',\tilde{t}') )=\sigma(\kappa \times \kappa (t, \tilde{t}),\kappa \times \kappa (t', \tilde{t}')).
$$

Furthermore, the bilinear map $\tau_d$ is well-defined. Firstly, observe that, thanks to \eqref{eq:2:module_structure}, the following equality holds:
\begin{equation}\label{eq:3:ring_struct_31}
\begin{split}
\tau_d ( (t,\tilde{t}),(t',\tilde{t}'))&=\left[ \i{\kappa(\tilde{z}\smile \cha\, \kappa z')}-(-1)^{k(m-k)}\i{\kappa (z\smile \cha\, \kappa \tilde{z}')} \right]\mu\\
&=\left[ (-1)^{m-k}\i{\kappa(\cha \kappa \tilde{z}\smile  z')}-(-1)^{km} \i{\kappa (\cha \kappa z\smile  \tilde{z}')} \right]\mu.
\end{split}
\end{equation}
Secondly, the pre-images $(z,\tilde{z})$ and $(z',\tilde{z}')$ via $\beta \times \beta$ are fixed up to elements in $\frac{H^{k-1,m-k-1}_{sc}(M;\mathbb{R})}{H^{k-1,m-k-1}_{sc,\free}(M;\Z)}$. We can modify them as follows:
\begin{equation*}
(z+\tilde{\iota} u,\tilde{z}+\tilde{\iota} \tilde{u}),\quad\qquad (z'+\tilde{\iota} u',\tilde{z}'+\tilde{\iota} \tilde{u}'),
\end{equation*}
for some $(u,\tilde{u}),(u',\tilde{u}')\in \frac{H^{k-1,m-k-1}_{sc}(M;\mathbb{R})}{H^{k-1,m-k-1}_{sc,\free}(M;\Z)}$.
Furthermore, introducing elements $([A],[\tilde{A}]),([A'],[\tilde{A}'])\in \mathfrak{T}^k_{sc}(M;\Z)$ by:
\begin{equation*}
\begin{split}
&([A],[\tilde{A}]):=\tilde{\kappa}\times \tilde{\kappa}(u,\tilde{u})\\
&([A'],[\tilde{A}]'):=\tilde{\kappa}\times \tilde{\kappa}(u',\tilde{u}'),
\end{split}
\end{equation*}
the commutativity of the upper-left square in diagram \eqref{diag:3:schematics} yields:
\begin{equation*}
\begin{split}
&(\kappa\times\kappa)(\tilde{\iota}\times \tilde{\iota})(u,\tilde{u})=(\iota \times \iota)([A],[\tilde{A}]),\\
&(\kappa\times\kappa)(\tilde{\iota}\times \tilde{\iota})(u',\tilde{u}')=(\iota \times \iota)([A'],[\tilde{A}']).
\end{split}
\end{equation*}
Hence, recalling equations \eqref{eq:3:ring_struct_31} and \eqref{eq:2:module_structure}, the extra-terms vanish:
\begin{equation*}
\begin{split}
&\ii{(\kappa z' + \iota [A'])}\cdot \i{(\kappa \tilde{z}+\iota[\tilde{A}])}=\\
=& \ii{\kappa z'} \cdot \i{\kappa \tilde{z}}+\ii{\iota [A']}\cdot \i{\kappa \tilde{z}} + \ii{\kappa z'} \cdot \i{\iota[\tilde{A}]} + \ii{\iota[A']}\cdot \i{\iota[\tilde{A}]}\\
=&\ii{\kappa z'} \cdot \i{\kappa \tilde{z}} + (-1)^k {\kappa (\ii{\cha \iota [A']}\smile \i{\tilde{z}})} \\
&+ {\kappa (\ii{z'} \smile \i{\cha \iota[\tilde{A}]})}+ \ii{\iota[A']}\cdot \i{\iota[\tilde{A}]}\\
=&\ii{\kappa z'} \cdot \i{\kappa \tilde{z}}  + {\iota(\ii{A'} \wedge \i{\curv\kappa(\alpha \tilde{u} )})}\\
=&\ii{\kappa z'} \cdot \i{\kappa \tilde{z}},
\end{split}
\end{equation*}where we used $\cha\circ \iota=0$ and $\curv \circ \kappa=0$.

\section{Pre-symplectic decomposition}\label{section:3.2_Presymplectic_decomposition}
Before activating the quantization machinery introduced in Section \ref{subsec:2:quantiz}, let us show that the pre-symplectic Abelian group $(\C^k_{sc}(M;\Z),\sigma)$ decomposes as the direct sum of three pre-symplectically orthogonal pre-symplectic Abelian groups. Due to computational convenience, we change perspective and, by means of the isomorphisms provided by \eqref{diag:2:iso_Csc_H_sigma} and \cite[Diagram (5.30)]{BMath}, we perform our construction on the diagram for differential cohomology with compact support over an arbitrary Cauchy surface $\Sigma\subseteq M$. The isomorphisms induce pre-symplectic structures on the Abelian groups of such a diagram in a natural way. Notice that this does not entail any conceptual novelty, as we can pass indifferently from one diagram to the other under the action of the isomorphism maps.\\

In doing this, we need to map the groups in the corners into the central one. To this end, let us recall the following definition.
\begin{defn}[Split short exact sequence]
A short exact sequence in the category of Abelian groups $\mathsf{Ab}$
\begin{equation}
\xymatrix{
0 \ar[r] &A \ar[r]^-{f} &B \ar[r]^-{g}& C\ar[r]&0\\
}
\end{equation}
is called \emph{split} if and only if one of the following equivalent statements is fulfilled:
\begin{enumerates}[(i)]
\item \emph{Left split}: there exists a morphism $\xymatrix{B\ar[r]^-{\psi}& A}$ such that $\psi\circ f=\mathrm{Id}_A$;
\item \emph{Right split}: there exists a morphism $\xymatrix{C\ar[r]^-{\phi}& B}$ such that $g\circ \phi=\mathrm{Id}_C$;
\item \emph{Direct sum isomorphism}: there is an isomorphism of short exact sequences
$$
\xymatrix{
0 \ar[r] &A \ar[d] \ar[r]^-f & B \ar[d]^\simeq \ar[r]^-g & C \ar[d] \ar[r]& 0\\
0 \ar[r] &A \ar[r]^-i & A\oplus C \ar[r]^-\pi& C \ar[r]&0
}
$$
where $i:A\to A\oplus C$ is the canonical inclusion and $\pi:A\oplus C\to C$ is the canonical projection.
\end{enumerates}
\end{defn}

\begin{rem}
Observe that the above conditions implicitly affirm that the following decompositions ensue:
\begin{align*}
& B=\mathrm{Im}(\phi)\oplus \mathrm{ker}(g)\\
& B=\mathrm{Im}(f) \oplus \mathrm{ker}(\psi).
\end{align*}

The result concerning the equivalence of the conditions (i)-(iii) is known as \emph{splitting lemma}. We refer the interested reader to \cite{HAT02,ML95} for details and proof.
\end{rem}

All of the sequences in diagram \eqref{diag:3:constr_pre_sympl_str} are split. A proof is given by Becker, Schenkel and Szabo in \cite[Appendix A]{BSS14}, where only slight refinements are necessary.
The splitting maps are, in general, non canonical and they involve arbitrary choices. We know that, in our case, splittings exist; nonetheless, we have no guarantees that it is possible to select them so that they fulfil the further requirement of preserving the pre-symplectic products. In fact, our goal is to obtain the following decomposition:
\begin{equation}
\begin{split}
\left(\mathfrak{C}^k_{sc}(M;\mathbb{Z}), \sigma\right)=&\left( d\Omega_{sc}^{k-1}\cap *d\Omega^{m-k-1}_{sc}, \tau_u \right) \oplus \left( H_{sc,\tor}^{k}(M,\mathbb{Z}),\tau_d \right) \oplus\\
& \oplus \left(\dfrac{H_{sc}^{k-1,m-k-1}(M;\mathbb{R})}{H_{sc,\free}^{k-1,m-k-1}(M;\mathbb{Z})}\oplus H_{sc,\free}^{k,m-k}(M;\mathbb{Z}),\tau_{lr} \right),
\end{split}
\end{equation}
or an equivalent one up to isomorphism. We will show that, if we restrict the topologies of the class of admissible spacetimes, this is always possible.

\begin{prop}\label{prop:3:existence_of_splitting}
Let $M$ be an object in $\mathsf{Loc}_m$ with compact Cauchy surface $\Sigma$. With reference to diagram \eqref{diag:3:splitting}, there exist splitting homomorphisms:
\begin{subequations}\label{eq:3:splitting}
\begin{align}
& a: d\Omega^{k-1,m-k-1}(\Sigma) \to \dfrac{\Omega^{k-1,m-k-1}(\Sigma)}{\Omega_{\mathbb{Z}}^{k-1,m-k-1}(\Sigma)}, \label{eq:3:splitting_a}\\
& b : H^{k,m-k}_{\tor}(\Sigma;\mathbb{Z}) \to H^{k-1,m-k-1}(\Sigma;\mathbb{T}), \label{eq:3:splitting_b}\\
& x : H^{k,m-k}_{\free}(\Sigma;\mathbb{Z}) \to \hat{H} ^{k,m-k}(\Sigma;\mathbb{Z}), \label{eq:3:splitting_x}
\end{align}
\end{subequations}
satisfying the following conditions:
\begin{subequations}\label{eq:3:splitting_condition}
\begin{align}
& \sigma(x(v,\tilde{v}),x(v',\tilde{v}'))=0, \label{eq:3:splitting_condition_xx}\\
& \sigma((\iota\times \iota) a(dA,d\tilde{A}),x(v,\tilde{v}))=0, \label{eq:3:splitting_condition_ax}\\
& \sigma((\kappa\times \kappa) b(t,\tilde{t}), x(v,\tilde{v}))=0 \label{eq:3:splitting_condition_bx},
\end{align}
\end{subequations}
for all $(v,\tilde{v}),(v',\tilde{v}')\in H^{k,m-k}_{\free}(\Sigma;\mathbb{Z})$, $(dA,d\tilde{A})\in d\Omega^{k-1,m-k-1}(\Sigma)$ and $(t,\tilde{t})\in H^{k,m-k}_{\tor}(\Sigma;\mathbb{Z})$.
\end{prop}

{\footnotesize
\begin{equation}\label{diag:3:splitting}
\xymatrix{
{\color{white}{0}}& 0 \ar[d] & 0 \ar[d] & 0 \ar[d] & {\color{white}{0}}\\
0\ar[r] & \dfrac{H^{k-1,m-k-1}(\Sigma;\mathbb{R})}{H^{k-1,m-k-1}_{\free}(\Sigma;\mathbb{Z})} \ar[r]^-{\tilde{\kappa}\times \tilde{\kappa}} \ar[d]^-{\tilde{\iota}\times \tilde{\iota}} & \dfrac{\Omega^{k-1,m-k-1}(\Sigma)}{\Omega_{\mathbb{Z}}^{k-1,m-k-1}(\Sigma)} \ar[r]^-{d\times d} \ar[d]^-{\iota \times \iota} & d\Omega^{k-1,m-k-1}(\Sigma) \ar[r]\ar[d]^-{\subseteq \times \subseteq} \ar@/^1em/[l]^{a}  & 0\\
0\ar[r]& H^{k-1,m-k-1}(\Sigma;\mathbb{T}) \ar[r]^-{\kappa \times \kappa} \ar[d]^-{\beta\times\beta}& \hat{H} ^{k,m-k} (\Sigma;\mathbb{Z})  \ar[r]^-{\curv \times\curv} \ar[d]^-{\cha \times \cha} & \Omega^{k,m-k}_{\mathbb{Z}}(\Sigma) \ar[r] \ar[d]^-{([\cdot]\times[\cdot])}& 0\\
0 \ar[r] & H^{k,m-k}_{\tor}(\Sigma;\mathbb{Z})  \ar@/^1em/[u]^{b} \ar[r]^-{i\times i} \ar[d] & H^{k,m-k}(\Sigma;\mathbb{Z}) \ar[r]^-{q\times q}\ar[d] & H^{k,m-k}_{\free}(\Sigma;\mathbb{Z}) \ar[lu]_{x} \ar[r] \ar[d] &0\\
{\color{white}{0}}& 0 & 0 & 0 & {\color{white}{0}}\\
}
\end{equation}
}

\begin{proof}
Being $a$, $b$ and $x$ splitting homomorphisms entails that:
\begin{subequations}\label{eq:3:def_split}
\begin{align}
&(d\times d)\circ a =\text{id}_{d\Omega^{k-1,m-k-1}(\Sigma)}, \label{eq:3:def_split_a}\\
&(\beta\times\beta)\circ b=\text{id}_{H^{k,m-k}_{\tor}(\Sigma;\mathbb{Z})}, \label{eq:3:def_split_b}\\
&([\cdot]\times[\cdot])\circ (\curv\times \curv) \circ x=\text{id}_{H^{k,m-k}_{\free}(\Sigma;\mathbb{Z})} \label{eq:3:def_split_x}.
\end{align}
\end{subequations}
To begin with, let us consider the existence of $x$. The pairing \eqref{pairing:2:duality_cohomology} induces a non-degenerate pairing:
\begin{eqnarray} \label{pair:3:f}
\langle \cdot,\cdot \rangle_f: \dfrac{H^{k-1}(\Sigma;\mathbb{R})}{H_{\free}^{k-1}(\Sigma;\mathbb{Z})}\times H_{\free}^{m-k}(\Sigma;\mathbb{Z})\to \mathbb{T},\nonumber \\
([\theta],\omega)\mapsto \left(\theta \smile \omega \right)\mu.
\end{eqnarray}
Since $H^k_{\free}(\Sigma;\mathbb{Z})$ and $H^{m-k}_{\free}(\Sigma;\mathbb{Z})$ are free Abelian groups, there exist $n,\tilde{n}\in \mathbb{N}$ such that $H_{\free}^k(\Sigma;\mathbb{Z})\simeq \mathbb{Z}^n$ and $H^{m-k}_{\free}(\Sigma;\mathbb{Z})\simeq \mathbb{Z}^{\tilde{n}}$. Pick out sets of generators $\{z_i,i\in\{1,\dots,n\}\}$, $\{\tilde{z}_j, j=\{1,\dots,\tilde{n}\}\}$ for $\mathbb{Z}^n$ and $\mathbb{Z}^{\tilde{n}}$ respectively. With abuse of notation, $z_i$ and $\tilde{z}_j$ will denote also the corresponding generators in the cohomology groups via the isomorphism. Therefore, a set of generators for $H_{\free}^{k,m-k}(\Sigma;\mathbb{Z})$ is given by:
$$
\{(z_i,0),(0,\tilde{z}_j)\,|\,i\in\{1,\dots,n\}, j\in\{1,\dots,\tilde{n}\}\}.
$$
Interpreting $H^\ast_{\free}(\Sigma;\mathbb{Z})$ as a lattice in $H^\ast(\Sigma;\mathbb{R})$ , a basis for the former is also a basis for the latter. Then, consider the isomorphisms $H^{m-k-1}(\Sigma;\mathbb{R})\simeq \mathbb{R}^n$ and $H_c^{k-1}(\Sigma;\mathbb{R})\simeq\mathbb{R}^{\tilde{n}}$ and sort out sets of generators dual to the previous ones with respect to the pairing \eqref{pairing:2:duality_cohomology}, once we have replaced $\T$ and $\Z$ with $\R$:
\begin{align*}
& \{\tilde{r}_i, i\in\{1,\dots,{n}\}\,|\, \tilde{r}_i\smile z_{i'}=\delta_{ii'}\,\, \forall i'\in\{1,\dots,n\} \}, \\
& \{r_j, j\in\{1,\dots,\tilde{n}\}\,|\, r_j\smile \tilde{z}_{j'}=\delta_{jj'}\,\, \forall j' \in\{1,\dots,\tilde{n}\}\}.
\end{align*}
For all $i\in\{1,\dots,n\}$ choose $h_i$ such that $[\curv h_i]=z_i$ and for all $j\in\{1,\dots,\tilde{n}\}$ choose $\tilde{h}'_j$ such that $[\curv \tilde{h}'_j]=\tilde{z}_j$.  This is certainly possible due to the surjectivity of $\curv$ and $[\arg]$.

The next step consists in modifying $\tilde{h}_j'$ so that the new differential character $\tilde{h}_j$ satisfies:
$$
\langle I \tilde{h}_j, h_i\rangle_c=0 	\qquad \forall j\in\{1,\dots,\tilde{n}\}, \forall i \in\{1,\dots,n\}.
$$
For all $j\in\{1,\dots,\tilde{n} \}$ set:
\begin{equation*}
\tilde{h}_j:=\tilde{h}_j'+\iota (\tilde{\kappa} \tilde{u}_j),
\end{equation*}
for some $\tilde{u}_j \in \frac{H^{m-k-1}(\Sigma;\mathbb{R})}{H_{\free}^{m-k-1}(\Sigma;\mathbb{Z})}$ to be determined.
Now:
\begin{equation*}
\langle I\tilde{h}_j,h_i \rangle_c=\langle I\tilde{h}'_j,h_i\rangle_c + \langle I \iota (\tilde{\kappa} \tilde{u}_j), h_i\rangle_c= \langle I \tilde{h}'_j,h_i\rangle_c + \langle I \tilde{u}_j, z_i\rangle_f.
\end{equation*}
Our constraint on $\tilde{u}_j$ takes the form:
$$
\langle I \tilde{u}_j, z_i\rangle_f=-\langle I \tilde{h}'_j,h_i\rangle_c.
$$
Let us introduce an array of real numbers $(c_{ij})\in \mathbb{M}(n,\tilde{n};\mathbb{R})$ defined as follows:
\begin{equation*}
c_{ij}\,\,\mod\Z=-\langle I\tilde{h}'_j,h_i\rangle_c.
\end{equation*}
Furthermore, define:
\begin{equation*}
\tilde{s}_j:=\sum_{i'=1}^n c_{i'j}\tilde{r}_{i'} \in H^{m-k-1}_{c}(\Sigma;\mathbb{R}).
\end{equation*}
We claim that $\widetilde{u}_j$  obtained as the image of $\widetilde{s}_j$ via the quotient map fulfils our requirement. As a matter of fact:
\begin{equation*}
\begin{split}
(\tilde{s}_j\smile z_j)\mu\,\,\mod\Z= & \left(\sum_{i'=1}^n c_{i'j}\tilde{r}_{i'} \smile z_i \right)\mu\,\,\mod\Z= \\
= &\sum_{i'=1}^n c_{i'j} \delta_{i'i} \,\,\mod\Z= c_{ij}\,\,\mod\Z=\\
=& -\langle I\tilde{h}'_j\cdot h_i\rangle_c.
\end{split}
\end{equation*}
We are now in the position to define $x$ through its action on the generators:
\begin{gather*}
x: 
\begin{cases}
(z_i,0)\mapsto (h_i,0) \qquad \forall i\in\{1,\dots,n \},\\
(0,\tilde{z}_j)\mapsto (0,\tilde{h}_j)\qquad \forall j\in\{1, \dots, \tilde{n}\}.
\end{cases}
\end{gather*}
Notice that the homomorphism $x$ splits the map from the central group to the down-right one. This is true by construction when we have $h_i$ and $\tilde{h}_j'$. The replacement of $\tilde{h}_j'$ with $\tilde{h}_j$ does not affect this property.

Let us check that \eqref{eq:3:splitting_condition_xx} holds for generators:
\begin{equation*}
\sigma(x(z_i,0),x(z_{i'},0))=\sigma((h_i,0),(h_{i'},0))=0,
\end{equation*}
\begin{equation*}
\sigma(x(0,\tilde{z}_j),x(0,\tilde{z}_{j'}))=\sigma((0,\tilde{h}_j),(0,\tilde{h}_{j'}))=0,
\end{equation*}
\begin{equation*}
\begin{split}
\sigma((z_i,0),(0,\tilde{z}_j))=&\sigma((h_i,0),(0,\tilde{h}_j))\\
=&0-(-1)^{k(m-k)}\langle I h_i, \tilde{h}_j\rangle_c\\
=&-\langle I \tilde{h}_j, h_i\rangle_c=0,
\end{split}
\end{equation*}
where the first two equalities are true by definition of $\sigma$, whereas the last one is true by construction. Then \eqref{eq:3:splitting_condition_xx} holds for every $(v,\tilde{v}), (v',\tilde{v}')\in H^{k,m-k}_{c,\free}(\Sigma;\mathbb{Z})$: the existence of $x$ is proved.\\


Let us move to $a$. As the upper horizontal sequence is split, we can find a homomorphism 
$$
a':  d\Omega_{c}^{k-1,m-k-1}(\Sigma) \to \dfrac{\Omega_c^{k-1,m-k-1}(\Sigma)}{\Omega_{c,\mathbb{Z}}^{k-1,m-k-1}(\Sigma)}
$$
such that $(d\times d)\circ a'=\mathrm{Id}_{d\Omega_{c}^{k-1,m-k-1}(\Sigma)}$. For the general properties of split sequences, other splittings of the same sequence can be obtained by adding an element of the form:
$$
\Delta_a \in \text{Hom}\left(d\Omega_{c}^{k-1,m-k-1}(\Sigma);\frac{H^{k-1,m-k-1}_{c}(\Sigma;\mathbb{R})}{H^{k-1,m-k-1}_{c,\free}(\Sigma;\mathbb{Z})} \right).$$ 
The strategy is to set:
\begin{equation*}
a:=a'+(\tilde{\kappa}\times \tilde{\kappa})\circ \Delta_a
\end{equation*}
and to search for a $\Delta_a$ such that $a$ fulfils \eqref{eq:3:splitting_condition_ax}. 

Introduce the following notation: for every map $f: X\to Y \times Z$ denote by $f_1:=\pi_1\circ f$ and $f_2:=\pi_2\circ f$ the maps to the first and to the second component of the image respectively, where $\pi_1: Y\times Z\to Y$ and $\pi_2:Y\times Z\to Z$ are the canonical projections.

The constraint imposed by \eqref{eq:3:splitting_condition_ax} reads:
\begin{equation*}
\begin{split}
&\sigma ((\iota\times\iota)a(dA,d\tilde{A}),x(v,\tilde{v}))=\\
&= \sigma( (\iota\times\iota)a'(dA,d\tilde{A}),x(v,\tilde{v})) +\sigma((\iota\times\iota)(\tilde{\kappa}\times \tilde{\kappa})\Delta_a(dA,d\tilde{A}) , x(v,\tilde{v}))\\
&= \left[\I{\iota a'_2(dA,d\tilde{A})}\cdot x_1(v,\tilde{v})-(-1)^{k(m-k)}\I{\iota a_1'(dA,d\tilde{A})}\cdot x_2(v,\tilde{v})+\right. \\
&\left. + \I{\iota \tilde{\kappa}\Delta_{a2}(dA,d\tilde{A})}\cdot x_1(v,\tilde{v}) -(-1)^{k(m-k)}\I{\iota \tilde{\kappa}\Delta_{a1}(dA,d\tilde{A})}\cdot x_2(v,\tilde{v})\right]\mu\\
&= \left[\I{\iota a'_2(dA,d\tilde{A})}\cdot x_1(v,\tilde{v})-(-1)^{k(m-k)}\I{\iota a_1'(dA,d\tilde{A})}\cdot x_2(v,\tilde{v})\right]\mu\\
& + \left[\I{\Delta_{a2}(dA,d\tilde{A})}\smile  [\curv x_1(v,\tilde{v})] +\right.\\
&\left. -(-1)^{k(m-k)} \I{\Delta_{a1}(dA,d\tilde{A})}\smile [\curv x_2(v,\tilde{v})]\right]\mu\\
& \overset{\downarrow}{=}0.
\end{split}
\end{equation*}
The symbol $\overset{\downarrow}{=}$ means that we are \emph{imposing} the left-hand side and the right-hand side to be equal.
Let us evaluate the second entry of the pre-symplectic form on the generators alternately.

$(z_i,0),\, i\in\{1,\dots,n\}:$
\begin{equation*}
\begin{split}
&\left[\I{\iota a'_2(dA,d\tilde{A})}\cdot x_1(z_i,0) -(-1)^{k(m-k)}\I{\iota a_1'(dA,d\tilde{A})}\cdot x_2(z_i,0)\right]\mu+\\
&+ 	\left[\I{\Delta_{a2}(dA,d\tilde{A})}\smile  [\curv x_1(z_i,0)]\right.+\\
&\left. -(-1)^{k(m-k)} \I{\Delta_{a1}(dA,d\tilde{A})}\smile [\curv x_2(z_i,0)]\right]\mu=\\
&=\left[ 	\I{\iota a'_2(dA,d\tilde{A})}\cdot h_i\right]\mu + \left[\I{\Delta_{a2}(dA,d\tilde{A})}\smile  z_i \right]\mu.
\end{split}
\end{equation*}

$(0,\tilde{z}_j),\,j\in\{1,\dots,\tilde{n}\}:$
\begin{equation*}
\begin{split}
&\left[ \I{\iota a'_2(dA,d\tilde{A})}\cdot x_1(0,\tilde{z}_j) -(-1)^{k(m-k)}\I{\iota a_1'(dA,d\tilde{A})}\cdot x_2(0,\tilde{z}_j)\right]\mu+\\
& + \left[\I{\Delta_{a2}(dA,d\tilde{A})}\smile  [\curv x_1(0,\tilde{z}_j)]+\right. \\
& \left. -(-1)^{k(m-k)} \I{\Delta_{a1}(dA,d\tilde{A})}\smile [\curv x_2(0,\tilde{z}_j)]\right]\mu\\
&= \left[-(-1)^{k(m-k)}\I{\iota a_1'(dA,d\tilde{A})}\cdot \tilde{h}_j\right]\mu -(-1)^{k(m-k)} \left[\I{\Delta_{a1}(dA,d\tilde{A})}\smile \tilde{z}_j\right]\mu.\\
\end{split}
\end{equation*}
Hence we have to impose:
\begin{align*}
&\langle I\Delta_{a1}(dA,d\tilde{A}),\tilde{z}_j\rangle_f=-\langle I\iota a_1'(dA,d\tilde{A}),\tilde{h}_j \rangle_c \qquad \forall j\in\{1,\dots,\tilde{n}\},\\
&\langle I\Delta_{a2}(dA,d\tilde{A}),z_i \rangle_f=-\langle I\iota a'_2(dA,d\tilde{A}), h_i \rangle_c \qquad \forall i\in\{1,\dots,n\}.
\end{align*}

Making use of the above ingredients, we construct the homomorphism:
\begin{gather*}
\Delta_a: d\Omega^{k-1,m-k-1}(\Sigma)\to \frac{H^{k-1,m-k-1}(\Sigma;\mathbb{R})}{H^{k-1,m-k-1}_{\free}(\Sigma;\mathbb{Z})} \\
(dA,d\tilde{A})\mapsto 
\begin{pmatrix}
(z_i,0)\mapsto -\langle I\iota a'_2(dA,d\tilde{A}), h_i \rangle_c\\
(0,\tilde{z}_j)\mapsto -\langle I\iota a_1'(dA,d\tilde{A}),\tilde{h}_j \rangle_c
\end{pmatrix}.
\end{gather*}
The right-hand side is an element in $\left(H_{\free}^{k(m-k)}(\Sigma;\mathbb{Z}) \right)^\star\simeq  \frac{H^{k-1,m-k-1}(\Sigma;\mathbb{R})}{H_{\free}^{k-1,m-k-1}(\Sigma;\mathbb{Z})}$. The map $\Delta_a$ is consequently well-defined. The existence of $a$ ensues.\\


As far as the existence of the homomorphism $b$ is concerned, let $b':H_{\tor}^{k,m-k}(\Sigma;\mathbb{Z})\to H^{k-1,m-k-1}(\Sigma;\mathbb{T})$ be an arbitrary splitting homomorphism, whose existence is given by the left column in \eqref{diag:3:splitting} being split. Define:
\begin{equation*}
b:=b'+\tilde{\iota}\times\tilde{\iota}\circ\Delta_b,
\end{equation*}
where $\Delta_b:H^{k,m-k}_{\tor}(\Sigma;\mathbb{Z})\to \frac{H^{k-1,m-k-1}(\Sigma;\mathbb{R})}{H^{k-1,m-k-1}_{\free}(\Sigma;\mathbb{Z})}$ is a homomorphism to be determined. As above, our aim is to construct $\Delta_b$ such that $b$ fulfils \eqref{eq:3:splitting_condition_bx}. Observe that $b$ is automatically a splitting homomorphism, in view of $\beta\circ \tilde{\iota}=0$.\\
The constraint \eqref{eq:3:splitting_condition_bx} writes:
\begin{equation*}
\begin{split}
&\sigma  ( (\kappa\times \kappa)b(t,\tilde{t}), x(v,\tilde{v})) =\\
&= \sigma( (\kappa\times\kappa)b'(t,\tilde{t}),x(v,\tilde{v})) + \sigma( (\kappa\times\kappa)(\tilde{\iota}\times\tilde{\iota})\Delta_b(t,\tilde{t}), x(v,\tilde{v}) ) \\
&= \left[ \I{\kappa b'_2(t,\tilde{t})}\cdot x_1(v,\tilde{v}) -(-1)^{k(m-k)}\I{\kappa b'_1(t,\tilde{t})}\cdot x_2(v,\tilde{v}) \right.\\
&+\left. \I{\kappa (\tilde{\iota} \Delta_{b2}(t,\tilde{t}))}\cdot x_1(v,\tilde{v}) -(-1)^{k(m-k)}\I{\kappa(\tilde{\iota} \Delta_{b1}(t,\tilde{t}))}\cdot x_2(v,\tilde{v})\right]\mu\\
&= \left[ \I{\kappa b'_2(t,\tilde{t})}\cdot x_1(v,\tilde{v}) -(-1)^{k(m-k)}\I{\kappa b'_1(t,\tilde{t}}\cdot x_2(v,\tilde{v})\right]\mu\\
&+ \left[\I{\Delta_{b2}(t,\tilde{t})}\smile [\curv x_1(v,\tilde{v})]+\right.\\
&\left.  -(-1)^{k(m-k)} \I{\Delta_{b1}(t,\tilde{t})}\smile [\curv x_2(v,\tilde{v})]\right]\mu\overset{\downarrow}{=}0.\\
\end{split}
\end{equation*}
Replacing $(v,\tilde{v})$ with the generators alternately, we get:
\begin{align*}
& \langle I\Delta_{b2}(t,\tilde{t}),z_i\rangle_f=- \langle I\kappa b_2'(t,\tilde{t}),h_i \rangle_c\qquad \forall i\in\{1,\dots,n\},\\
& \langle I\Delta_{b1}(t,\tilde{t}),\tilde{z}_j\rangle_f= -\langle I\kappa b_1'(t,\tilde{t}),\tilde{h}_j\rangle_c\qquad \forall j\in\{1,\dots,\tilde{n}\}.
\end{align*}
As a consequence, we define $\Delta_b$ by:
\begin{gather*}
\Delta_b:H^{k,m-k}_{\tor}(\Sigma;\mathbb{Z})\to \frac{H^{k-1,m-k-1}(\Sigma;\mathbb{R})}{H^{k-1,m-k-1}_{\free}(\Sigma;\mathbb{Z})}\\
(t,\tilde{t})\mapsto
\begin{pmatrix}
(z_i,0)\mapsto -\langle I \kappa b_2'(t,\tilde{t}), h_i \rangle_c\\
(0,\tilde{z}_j)\mapsto -\langle \kappa I b_1'(t,\tilde{t}),\tilde{h}_j \rangle_c
\end{pmatrix}.
\end{gather*}
The right-hand side is an element in $\left( H_{\free}^{k(m-k)}(\Sigma;\mathbb{Z}) \right)^\star\simeq  \frac{H^{k-1,m-k-1}(\Sigma;\mathbb{Z})}{H^{k-1,m-k-1}_{\free}(\Sigma;\mathbb{Z})}$. $\Delta_b$ is well-defined and the existence of $b$ then ensues.
\end{proof}

\begin{rem}
Denoting by $\mathfrak{U}^k(\Sigma;\Z)$ the pushout of the diagram:
\begin{equation}
\xymatrix{
\dfrac{H^{k-1,m-k-1}(\Sigma;\mathbb{R})}{H_{\free}^{k-1,m-k-1}(\Sigma;\mathbb{Z})} \ar[r]^-{\tilde{\kappa}\times \tilde{\kappa}} \ar[d]_-{\tilde{\iota}\times\tilde{\iota}} & \dfrac{\Omega^{k-1,m-k-1}(\Sigma)}{\Omega_\Z^{k-1,m-k-1}(\Sigma)}\\
 H^{k-1,m-k-1}(\Sigma;\T)  &
}
\end{equation}
the sequence:
\begin{equation}\label{sequence:3:x}
\xymatrix{
0 \ar[r] & \mathfrak{U}^k(\Sigma;\Z) \ar[r] & \hat{H}^{k,m-k}(\Sigma;\Z) \ar[r] & H^{k,m-k}_\free(\Sigma;\Z) \ar[r] & 0
}
\end{equation}
is short exact by construction and split on account of the properties of diagram \eqref{diag:3:splitting}. Notice that the sequence split by $x$ in the previous proposition is exactly \eqref{sequence:3:x}.\\
\end{rem}

\begin{thm} \label{thm:3:existence_decomposition_from_splittings}
Under the assumptions of Proposition \ref{prop:3:existence_of_splitting}, $a$, $b$ and $x$ yield a decomposition of $\hat{H}^{k,m-k}(\Sigma;\Z)$ as the direct sum of pre-symplectic Abelian groups, which is orthogonal with respect to the pre-symplectic structure:
\begin{equation}
\begin{split}
\left(\hat{H} ^{k,m-k}(\Sigma;\mathbb{Z}),\sigma \right)\simeq& \left( \frac{H^{k-1,m-k-1}(\Sigma;\mathbb{R})}{H^{k-1,m-k-1}_{\free}(\Sigma;\mathbb{Z})} \oplus H^{k,m-k}_{\free}(\Sigma;\mathbb{Z}), \tau_{lr} \right) \oplus \\
& \oplus \left(H^{k,m-k}_{\tor}(\Sigma;\mathbb{Z}),\tau_d \right) \oplus  \left(d\Omega^{k-1,m-k-1}(\Sigma),\tau_u \right).
\end{split}
\end{equation} 
\end{thm}
\begin{proof}
Consider the homomorphism:
\begin{eqnarray*}
& F: \frac{H^{k-1,m-k-1}(\Sigma;\mathbb{R})}{H^{k-1,m-k-1}_{\free}(\Sigma;\mathbb{Z})} \times H^{k,m-k}_{\free}(\Sigma;\mathbb{Z}) \to \hat{H} ^{k,m-k}(\Sigma;\mathbb{Z})	\nonumber &\\
& (u,\tilde{u},v,\tilde{v})\mapsto (\iota (\tilde{\kappa} u),\iota (\tilde{\kappa}\tilde{u}))+x(v,\tilde{v}).&
\end{eqnarray*}
$F$ preserves the pre-symplectic structure. In fact:
\begin{equation}
\begin{split}
&\sigma\left(F(u,\tilde{u},v,\tilde{v}),F(u',\tilde{u}',v',\tilde{v}')\right)\\
&= \sigma \left( (\iota (\tilde{\kappa} u),\iota (\tilde{\kappa}\tilde{u})),(\iota (\tilde{\kappa} u'),\iota (\tilde{\kappa}\tilde{u}'))\right)+ \sigma \left(x(v,\tilde{v}),x(v',\tilde{v}')\right)+\\
& + \sigma \left( (\iota (\tilde{\kappa} u),\iota (\tilde{\kappa}\tilde{u})), x(v', \tilde{v}') \right)+\sigma\left( x(v,\tilde{v}), (\iota (\tilde{\kappa} u'),\iota (\tilde{\kappa}\tilde{u}'))\right)=\\
& = 0 + 0 + \left[\I{\iota (\tilde{\kappa} \tilde{u})} \cdot x_1(v',\tilde{v}')-(-1)^{k(m-k)}\I{\iota(\tilde{\kappa} u)}\cdot x_2(v', \tilde{v}')+\right.\\
&+\left. \I{x_2(v,\tilde{v})}\cdot \iota (\tilde{\kappa}u')-(-1)^{k(m-k)} \I{x_1(v,\tilde{v})}\cdot \iota (\tilde{\kappa}\tilde{u}')\right]\mu=\\
&= \iota \left[\tilde{\kappa}\tilde{u}\wedge \curv x_1(v',\tilde{v}')-(-1)^{k(m-k)} \tilde{\kappa}u \wedge \curv x_2(v',\tilde{v}')\right.\\ 
&\left. -\tilde{\kappa}\tilde{u}'\wedge \curv x_1(v,\tilde{v})+(-1)^{k(m-k)}\tilde{\kappa} u'\wedge \curv x_1 (v',\tilde{v}')\right]\mu\\
\end{split}
\end{equation}
where in the first passage the first term is vanishing as one observes that $\iota\circ\tilde{\kappa}=\kappa\circ\tilde{\iota}$ and $\cha\circ\iota=0$, whilst the second one is zero in view of \eqref{eq:3:splitting_condition_xx}. It is a straightforward check that what we obtained is nothing but $\tau_{lr}$:
$$
F^*\sigma=\tau_{lr}.
$$

As far as $d\Omega^{k-1,m-k-1}(\Sigma)$ is concerned, define the homomorphism of Abelian groups:
\begin{eqnarray*}
&D: d\Omega^{k-1,m-k-1}(\Sigma)\to \hat{H} ^{k,m-k}(\Sigma;\mathbb{Z}) \nonumber &\\
&(dA,d\tilde{A})\mapsto (\iota\times \iota) a (dA,d\tilde{A}).&
\end{eqnarray*}
Let us show that it preserves the pre-symplectic structure:
\begin{equation*}
\begin{split}
\sigma &(D(dA,d\tilde{A}),D(dA',d\tilde{A}'))= \sigma ((\iota\times \iota) a (dA,d\tilde{A}), (\iota\times \iota) a(dA', d\tilde{A}'))\\
=&\left[ \I{\iota a_2 (dA,d\tilde{A})}\cdot \iota a_1(dA',d\tilde{A}') - (-1)^{k(m-k)} \I{\iota a_1(dA,d\tilde{A})}\cdot \iota a_2(dA',d\tilde{A}')\right]\mu\\
=& \iota \left[ a_2 (dA,d\tilde{A})\wedge \curv \iota a_1(dA',d\tilde{A}')+\right.\\
&\left. - (-1)^{k(m-k)} a_1(dA,d\tilde{A})\wedge \curv \iota a_2(dA',d\tilde{A}')\right]\mu\\
=& \iota \left[ a_2 (dA,d\tilde{A})\wedge \subseteq d a_1(dA',d\tilde{A}')+\right.\\
& \left.- (-1)^{k(m-k)} a_1(dA,d\tilde{A})\wedge \subseteq d a_2(dA',d\tilde{A}')\right]\mu\\
=& \iota \left[ a_2 (dA,d\tilde{A})\wedge dA' - (-1)^{k(m-k)} a_1(dA,d\tilde{A})\wedge d\tilde{A}'\right]\mu\\
=&\tau_u ((dA,d\tilde{A}),(dA',d\tilde{A}') ).
\end{split}
\end{equation*}
We have obtained that:
$$
D^\ast \sigma=\tau_{u}.
$$

The reasoning goes the same way for the torsion group. Define:
\begin{eqnarray*}
& T:H^{k,m-k}_{\tor}(\Sigma;\mathbb{Z}) \to \hat{H} ^{k,m-k}(\Sigma;\mathbb{Z})\nonumber &\\
&(t,	\tilde{t})\mapsto (\kappa\circ\kappa) b (t,\tilde{t}).&
\end{eqnarray*}
It satisfies:
$$
T^\ast\sigma=\tau_d.
$$
In fact:
\begin{equation*}
\begin{split}
&\sigma\left(T(t,\tilde{t}),T(t',\tilde{t}')\right)= \\
&= \left[\kappa b_2(t,\tilde{t})\cdot \kappa b_1(t',\tilde{t}')-(-1)^{k(m-k)}\kappa b_1(t,\tilde{t})\cdot\kappa b_2 (t',\tilde{t}')\right]\mu\\
&= \kappa\left[ b_2(t,\tilde{t})\smile \cha\kappa b_1(t',\tilde{t}')-(-1)^{k(m-k)} b_1(t,\tilde{t})\smile \cha\kappa b_2 (t',\tilde{t}')\right] \mu\\
&= \kappa\left[ b_2(t,\tilde{t})\smile i \beta b_1(t',\tilde{t}')-(-1)^{k(m-k)} b_1(t,\tilde{t})\smile i \beta b_2 (t',\tilde{t}')\right] \mu\\
&= \kappa\left[ b_2(t,\tilde{t})\smile i t' -(-1)^{k(m-k)} b_1(t,\tilde{t})\smile i \tilde{t}' \right] \mu\\
&= \tau_d \left((t,\tilde{t}),(t',\tilde{t}') \right).
\end{split}
\end{equation*}

Let us move to orthogonality. Recalling that $\curv \circ \kappa=0$, consider $d\Omega^{k-1,m-k-1}(\Sigma)$ and $H^{k,m-k}_{\tor}(\Sigma;\mathbb{Z})$:
\begin{equation*}
\begin{split}
&\sigma( D(dA,d\tilde{A}), T(t,\tilde{t}) )= \sigma ( (\iota\times \iota) a (dA,d\tilde{A}) ,(\kappa\circ\kappa) b (t,\tilde{t}))\\
&= \left[\I{\iota a_2(dA,d\tilde{A})}\cdot \kappa b_1(t,\tilde{t}) -(-1)^{k(m-k)}\I{\iota a_1(dA,d\tilde{A})}\cdot \kappa b_2(t,\tilde{t})\right]\mu\\
&= \iota \left[a_2(dA,d\tilde{A})\wedge \curv \kappa b_1(t,\tilde{t}) -(-1)^{k(m-k)}a_1(dA,d\tilde{A})\wedge \curv \kappa b_2(t,\tilde{t})\right]\mu\\
&=0.
\end{split}
\end{equation*}
Secondly, consider $d\Omega^{k-1,m-k-1}(\Sigma)$ and $\frac{H^{k-1,m-k-1}(\Sigma;\mathbb{R})}{H^{k-1,m-k-1}_{\free}(\Sigma;\mathbb{Z})} \oplus H^{k,m-k}_{\free}(\Sigma;\mathbb{Z})$:
\begin{equation*}
\begin{split}
&\sigma( D(dA,d\tilde{A}),F(u,\tilde{u},v,\tilde{v}) )=\sigma((\iota\times\iota)a(dA,d\tilde{A}),(\iota\tilde{\kappa}u,\iota \tilde{\kappa})+x(v,\tilde{v}))\\
&= \sigma((\iota\times\iota)a(dA,d\tilde{A}), (\iota\tilde{\kappa}u,\iota \tilde{\kappa})) +\sigma((\iota\times\iota)a(dA,d\tilde{A}), x(v,\tilde{v}))\\
&= 0 + \sigma((\iota\times\iota)a(dA,d\tilde{A}), x(v,\tilde{v}))=0,
\end{split}
\end{equation*}
where in the last passage we made use of \eqref{eq:3:splitting_condition_ax}.

Lastly, let us prove that $H^{k,m-k}_{\tor}(\Sigma;\mathbb{Z})$ and $\frac{H^{k-1,m-k-1}(\Sigma;\mathbb{R})}{H^{k-1,m-k-1}_{\free}(\Sigma;\mathbb{Z})} \oplus H^{k,m-k}_{\free}(\Sigma;\mathbb{Z})$ are orthogonal:
\begin{equation*}
\begin{split}
&\sigma(T(t,\tilde{t}),F(u,\tilde{u},v,\tilde{v}))= \sigma((\kappa\circ\kappa) b (t,\tilde{t}), (\iota\tilde{\kappa}u,\iota \tilde{\kappa})+x(v,\tilde{v}))\\
=& \sigma((\kappa\circ\kappa) b (t,\tilde{t}),(\iota\tilde{\kappa}u,\iota \tilde{\kappa}))+\sigma((\kappa\circ\kappa) b (t,\tilde{t}),x(v,\tilde{v}))\\
=& 0 + \sigma((\kappa\circ\kappa) b (t,\tilde{t}),x(v,\tilde{v}))=0,
\end{split}
\end{equation*}
where in the last passage we made use of \eqref{eq:3:splitting_condition_bx}. The proof is therefore completed.
\end{proof}

The quantization procedure can now be put forward, with a further remarkable feature: as the space of observables is presented as the orthogonal direct sum of pre-symplectic Abelian groups, the Weyl algebra on it is isomorphic to the tensor product of the Weyl algebras on the three direct summands. As a matter of fact, every $h\in \hat{H}^{k,m-k}(\Sigma;\Z)$ can be presented as $a_h \oplus b_h \oplus c_h$, with $a_h \in H^{k,m-k}_{\tor}(\Sigma;\mathbb{Z})$, $b_h\in \frac{H^{k-1,m-k-1}(\Sigma;\mathbb{R})}{H^{k-1,m-k-1}_{\free}(\Sigma;\mathbb{Z})} \oplus H^{k,m-k}_{\free}(\Sigma;\mathbb{Z})$ and $c_h\in d\Omega^{k-1,m-k-1}(\Sigma)$. 

Define:
\begin{equation*}
\begin{split}
\tilde{\mathcal{A}}_0:=&\mathcal{A}_0\left(H^{k,m-k}_{\tor}(\Sigma;\mathbb{Z}),\tau_d\right)\\  \otimes& \mathcal{A}_0\left(\frac{H^{k-1,m-k-1}(\Sigma;\mathbb{R})}{H^{k-1,m-k-1}_{\free}(\Sigma;\mathbb{Z})} \oplus H^{k,m-k}_{\free}(\Sigma;\mathbb{Z}),\tau_{lr}\right)\\
 \otimes &\mathcal{A}_0\left(d\Omega^{k-1,m-k-1}(\Sigma),\tau_u\right)
\end{split}
\end{equation*}

Consider the following homomorphism, defined by its action on the generators:
\begin{eqnarray}
& \mathcal{I}: \mathcal{A}_0\left( \hat{H}^{k,m-k}(\Sigma;\Z),\sigma\right) \to \tilde{\mathcal{A}}_0 &\nonumber\\
&\mathcal{W}(h)\mapsto \mathcal{W}(a_h)\otimes \mathcal{W}(b_h) \otimes \mathcal{W}(c_h).&
\end{eqnarray}
Set on $\tilde{\mathcal{A}}_0$ the algebra structure given by:
\begin{multline*}
[\mathcal{W}(a_h)\otimes \mathcal{W}(b_h) \otimes \mathcal{W}(c_h)][\mathcal{W}(a_{h'})\otimes \mathcal{W}(b_{h'}) \otimes \mathcal{W}(c_{h'})]:=\\
:=\mathcal{W}(a_h)\mathcal{W}(a_{h'})\otimes \mathcal{W}(b_h)\mathcal{W}(b_{h'}) \otimes \mathcal{W}(c_h)\mathcal{W}(c_{h'}).
\end{multline*}
We have then on the one hand:
\begin{equation*}
\begin{split}
\mathcal{I}\left(\mathcal{W}(h)\mathcal{W}(h')\right)=\mathcal{W}(a_h)\mathcal{W}(a_{h'})\otimes \mathcal{W}(b_h)\mathcal{W}(b_{h'}) \otimes \mathcal{W}(c_h)\mathcal{W}(c_{h'}).
\end{split}
\end{equation*}
On the other hand:
\begin{equation*}
\begin{split}
&\exp(2\pi i \sigma(h,h'))\mathcal{I}\mathcal{W}(h+h')=\\
&=\exp[2\pi i(\tau_d (a_h,a_{h'})+\tau_{lr}(b_h,b_{h'})+\tau_{u}(c_h,c_{h'}))]\,\mathcal{I}\mathcal{W}(h+h')\\
&=\exp(2\pi i \tau_d (a_h,a_{h'}))\,\mathcal{W}(a_h+a_{h'})\otimes\exp(2\pi i \tau_{lr} (b_h,b_{h'}))\,\mathcal{W}(b_h+b_{h'}) \otimes\\
&\otimes \exp(2\pi i \tau_u (c_h,c_{h'}))\,\mathcal{W}(c_h+c_{h'})
\end{split}
\end{equation*}
$\mathcal{I}$ is an algebra homomorphism, it preserves the unit and it is by construction injective and surjective, i.e.\ $\mathcal{I}$ is an algebra isomorphism. Therefore, we can choose $\tilde{\mathcal{A}}_0$ as the starting point for the quantization.\\

This property turns out to be extremely useful in the construction of states; in fact, it allows us to exhibit a state as a product of three contributions, each one corresponding to a tensor factor of the relevant $C^\ast$-algebra, as arising from the decomposition discussed above.
\section{The 2-dimensional case}\label{section:3.3_2dim_case}
The most basic, yet very interesting, case is $M=\R\times \S ^1$ with $k=1$. Far from being trivial, this example can be carried out in a fully explicit way. Whereas Schubert proved \cite{SHU13} the non-existence of ground states for 2-dimensional scalar massless Klein-Gordon field in general, a noteworthy feature of our model is that it admits a 2-dimensional Hadamard state which is ground.

\subsection{Elements in $\C^k_{sc}(\R\times \S^1;\Z)$} \label{section:3:differential_characters}
Let us endow $\R\times \S^1$ with the ultrastatic metric $g=-dt\otimes dt + d\theta\otimes d\theta$. Observing that $\S^1\times\{0\}$ is a compact Cauchy surface, the restriction to spacelike-compact supports is always automatically implemented. Furthermore, $\R\times \S^1$ is homotopy equivalent to $\S^1$. 
Since both homology and cohomology are invariants of homotopy type, we have $H_\ast(\mathbb{S}^1\times \mathbb{R};\mathbb{Z})\simeq H_\ast(\mathbb{S}^1;\mathbb{Z})$ and $H^\ast(\mathbb{S}^1\times \mathbb{R};\mathbb{Z})\cong  H^\ast(\mathbb{S}^1;\mathbb{Z})$. Explicitly:
\begin{align*}
& H^0(\mathbb{S}^1\times \mathbb{R};\mathbb{Z})=\mathbb{Z},\\
& H^1(\mathbb{S}^1\times \mathbb{R};\mathbb{Z})=\mathbb{Z},\\
& H^i(\mathbb{S}^1\times \mathbb{R};\mathbb{Z})=0 \quad \text{for}\,\,i>1,
\end{align*}
and
\begin{align*}
& H_0(\mathbb{S}^1\times \R;\mathbb{Z})=\mathbb{Z},\\
& H_1(\mathbb{S}^1\times \R;\mathbb{Z})=\mathbb{Z},\\
& H_i(\mathbb{S}^1\times R;\mathbb{Z})=0\quad \text{for}\,i>1.
\end{align*}
The central element in the last row is then $H^{1,1}(\R\times \S^1;\Z)\simeq \Z\times\Z$. The torsion subgroup is vanishing; consequently, $\frac{H^{0,0}_{sc}(\R\times \S^1;\R)}{H^{0,0}_{sc}(\R\times \S^1;\Z)}\simeq H^{0,0}_{sc}(\R\times \S^1;\T)\simeq (H^{1,1}(\R\times \S^1;\Z))^\star\simeq\T\times \T$.\\

Diagram \eqref{diag:3:constr_pre_sympl_str} becomes:
{\footnotesize
\begin{equation}\label{diag:3:R_S1}
\xymatrix{
& 0 \ar[d] & 0 \ar[d] & 0 \ar[d] &\\
0\ar[r] & \mathbb{T}\times \mathbb{T} \ar[r] \ar[d]& \mathfrak{T}^1(\R\times \S^1;\Z) \ar[r] \ar[d] & dC^\infty(M)\cap *dC^\infty(M) \ar[d] \ar[r]& 0\\
0\ar[r]& \mathbb{T}\times \mathbb{T} \ar[r] \ar[d]& \C^1(\R\times \S^1;\Z) \ar[r] \ar[d] & \Omega^{1}_{\mathbb{Z}}\cap *\Omega^{1}_{\mathbb{Z}}(M) \ar[r] \ar[d]& 0\\
0 \ar[r] &0 \ar[r] \ar[d] &\mathbb{Z}\times \mathbb{Z} \ar[r]\ar[d] & \mathbb{Z}\times \mathbb{Z} \ar[r] \ar[d] &0\\
& 0 & 0 & 0 & \\
}
\end{equation}
} Exploiting the Fourier decomposition, we claim that a generic element $(h,\tilde{h})\in \C^1(\R\times \S^1;\Z)$ can be presented as:
{\small \begin{equation}\label{eq:3:h}
\begin{split}
h(t,\theta)=h_0 +n\theta-\Bigg(\tilde{n}t-& \sum_{k=1}^\infty \left\{ - b_k ^{-} \cos[2\pi k(t-\theta)] -b_k^{+} \cos[2\pi k(t+\theta)]\right. \\
& \left. +a_k^{-} \sin[2\pi k(t-\theta)] + a_k^{+} \sin[2\pi k(t+\theta)]\right\} \,\,\mod\mathbb{Z} \Bigg) 
\end{split}
\end{equation}
\begin{equation}\label{eq:3:h_tilde}
\begin{split}
\tilde{h}(t,\theta)=\tilde{h}_0+\tilde{n}\theta-\Bigg(nt- & \sum_{k=1}^\infty  \left\{ - b_k ^{-} \cos[2\pi k(t-\theta)] +b_k^{+} \cos[2\pi k(t+\theta)]\right. \\
& \left. + a_k^{-} \sin[2\pi k(t-\theta)] - a_k^{+} \sin[2\pi k(t+\theta)]\right\}\,\,\mod \mathbb{Z}\Bigg) 
\end{split}
\end{equation}}
where $(h_0,\tilde{h}_0) \in \mathbb{T}\times \T$, $(n,\tilde{n})\in\mathbb{Z}\times \Z$. Basic calculus guarantees that the above expression for $h$ is the most general one. Observe that all of the series of the present section are uniformly convergent, as they are the Fourier expansions of smooth functions. The coefficients of $\tilde{h}$ are arranged so that $\curv h=\ast \curv \tilde{h}$. In fact:
\begin{align}\label{eq:3:curvh}
\curv h &=   +nd\theta-\tilde{n}dt\nonumber\\
+2\pi \sum_{k=1}^\infty & \left\{ + b_k ^{-} k \sin[2\pi k(t-\theta)] \,dt -b_k^{-} k\sin[2\pi k(t-\theta)]\,d\theta \right. \nonumber\\
&+ b_k ^{+} k \sin[2\pi k(t+\theta)] \,dt +b_k^{+} k \sin[2\pi k(t+\theta)]\,d\theta\nonumber\\
&+a_k^{-} k \cos[2\pi k(t-\theta)]\,dt - a_k^{-} k \cos[ 2\pi k(t-\theta)]\,d\theta\nonumber\\
&\left. +a_k^{+} k \cos[2\pi k(t+\theta)]\,dt + a_k^{+} k \cos[2\pi k(t+\theta)]\,d\theta \right\},
\end{align}
\begin{equation}\label{eq:3:curvh_tilde}
\begin{split}
\curv \tilde{h}&=   +\tilde{n}d\theta-ndt\\
+ 2\pi\sum_{k=1}^\infty & \left\{ + b_k ^{-} k \sin[2\pi k(t-\theta)] \,dt -b_k^{-} k \sin[2\pi k(t-\theta)]\,d\theta \right. \\
& - b_k ^{+} k \sin[2\pi k(t+\theta)] \,dt -b_k^{+} k \sin[2\pi k(t+\theta)]\,d\theta\\
& +a_k^{-} k \cos[2\pi k(t-\theta)]\,dt - a_k^{-} k \cos[2\pi k(t-\theta)]\,d\theta\\
& \left. -a_k^{+} k \cos[2\pi k(t+\theta)]\,dt - a_k^{+} k \cos[2\pi k(t+\theta)]\,d\theta \right\}.
\end{split}
\end{equation}
As far as the Hodge dual of the generators of $\Omega^1(	\R\times \S^1)$ is concerned, we have:
\begin{gather*}
\psi_t\, dt\wedge * d\theta+\psi_\theta\, d\theta\wedge * d\theta=\psi \wedge * d\theta=\langle \psi,d\theta\rangle * 1= g^{\theta\theta}\psi_\theta \,dt\wedge d\theta\\
\Rightarrow * d\theta=-dt,
\end{gather*}
and
\begin{gather*}
\psi_t \,dt\wedge *dt+\psi_\theta\, d\theta \wedge *dt=\psi \wedge *dt= \langle \psi,dt \rangle *1= g^{tt}\psi_t\,dt\wedge d\theta\\ 
\Rightarrow  *dt=-d	\theta,
\end{gather*}
where $\psi=\psi_t\, dt + \psi_\theta\, d\theta$ is a generic one-form.
A mere substitution in \eqref{eq:3:curvh} and \eqref{eq:3:curvh_tilde} yields the desired result.\\

From its expression, we can identify the different contributions to $(h,\tilde{h})$. First of all, $(h_0,\tilde{h}_0)$ lies in the image of $\kappa\times \kappa$, as it vanishes under the action of $\curv \times \curv$.  Defining: 
\begin{equation}\label{eq:3:h_phi}
\begin{split}
\varphi:=  \sum_{k=1}^\infty  \left\{ - b_k ^{-} \cos[2\pi k(t-\theta)] \right.&-b_k^{+} \cos[2\pi k(t+\theta)] \\
&\left. +a_k^{-} \sin[2\pi k(t-\theta)] + a_k^{+} \sin[2\pi k(t+\theta)] \right\},
\end{split}
\end{equation}
\begin{equation}\label{eq:3:h_phi_tilde}
\begin{split}
\tilde{\varphi}:= \sum_{k=1}^\infty  \left\{ - b_k ^{-} \cos[2\pi k(t-\theta)]\right. &+b_k^{+} \cos[2\pi k(t+\theta)]\\
&\left. + a_k^{-} \sin[2\pi k(t-\theta)] - a_k^{+} \sin[2\pi k(t+\theta)] \right\},
\end{split}
\end{equation}
we realize that:
\begin{equation*}
\curv\times \curv (h,\tilde{h})=(nd\,\theta-\tilde{n}\,dt + d\varphi, \tilde{n}\,d\theta-n\,dt + d\tilde{\varphi}).
\end{equation*}
Furthermore, the characteristic class is:
$$
\cha\times\cha (h,\tilde{h})=(n,\tilde{n}).
$$
At last, $d\varphi=*d\tilde{\varphi}$ is the contribution due to $dC^\infty\cap *dC^\infty (\R\times \S^1)$. Notice that the topologically trivial contribution to $(h,  \tilde{h})$ is given by $(h_0 + \varphi \,\,\mod\Z, \tilde{h}_0 +\tilde{\varphi}\,\,\mod\Z)\in \mathfrak{T}^1(\R\times \S^1;\Z)$. 

\subsection{Computation of the pre-symplectic product}\label{subsection:3.3.2_comput_presympl_prod}
Since for the case in hand the homomorphism $I$ is the identity map, the expression for $\sigma$ simplifies considerably. We can write:
\begin{eqnarray*}
&\sigma: \C^1(M;\Z)\times \C^1(M;\Z)\to \T&\\
&\left((h,\tilde{h}),(h',\tilde{h}') \right)\mapsto\left(\iota^\ast_{\S^1\times\{0\}}\tilde{h}\cdot \iota^\ast_{\S^1\times\{0\}} h'+ \iota^\ast_{\S^1\times\{0\}} h\cdot \iota^\ast_{\S^1\times\{0\}} \tilde{h}'  \right)\mu. &
\end{eqnarray*}
With the above notation, we define:
\begin{align*}
\i{h}:=\iota^*_{\mathbb{S}^1\times \{0\}}h(\theta)=h_0+n\theta+ \sum_{k=1}^\infty &\left\{ (- b_k ^{-} -b_k^{+}) \cos(2\pi k\theta) \right.\\
&\left. +(a_k^{+} -a_k^{-}) \sin(2\pi k\theta)\right\}  \,\, \mod\Z,\\
\i{\tilde{h}}:=\iota^*_{\mathbb{S}^1\times \{0\}}\tilde{h}(\theta)=\tilde{h}_0+\tilde{n}\theta+ \sum_{k=1}^\infty &\left\{ (b_k ^{+} -b_k^{-}) \cos(2\pi k\theta)\right.\\
&\left. +(-a_k^{+} -a_k^{-}) \sin(2\pi k\theta)\right\}  \,\,\mod\Z,\\
\i{h'}:=\iota^*_{\mathbb{S}^1\times \{0\}} h'(\theta)=h_0'+n'\theta+ \sum_{k=1}^\infty &\left\{ (- b_k^{-\prime} -b_k^{+\prime}) \cos(2\pi k\theta) \right.\\
&\left. +(a_k^{+\prime} -a_k^{-\prime}) \sin(2\pi k\theta)\right\} \,\,\mod\Z,\\
\i{\tilde{h}'}:=\iota^*_{\mathbb{S}^1\times \{0\}}\tilde{h}'(\theta)=\tilde{h}_0'+\tilde{n}'\theta+ \sum_{k=1}^\infty &\left\{ (b_k ^{+\prime} -b_k^{-\prime}) \cos(2\pi k\theta) \right.\\
&\left. +(-a_k^{+\prime} -a_k^{-\prime}) \sin(2\pi k\theta)\right\}  \,\,\mod\Z.\\
\end{align*}
The same way, we set:
{\small \begin{equation}\label{eq:3:phi_S1}
\begin{split}
&\phi:=\varphi(0,\theta)=\sum_{k=1}^\infty \left\{ (- b_k ^{-} -b_k^{+}) \cos(2\pi k\theta) +(a_k^{+} -a_k^{-}) \sin(2\pi k\theta)\right\},\\
&\tilde{\phi}:=\tilde{\varphi}(0,\theta)=\sum_{k=1}^\infty \left\{ (b_k ^{+} -b_k^{-}) \cos(2\pi k\theta) +(-a_k^{+} -a_k^{-}) \sin(2\pi k\theta)\right\},\\
&\phi ':=\varphi'(0,\theta)=\sum_{k=1}^\infty \left\{ (- b_k^{-\prime} -b_k^{+\prime}) \cos(2\pi k\theta) +(a_k^{+\prime} -a_k^{-\prime}) \sin(2\pi k\theta) \right\},\\
&\tilde{\phi}':= \tilde{\varphi}(0,\theta)=\sum_{k=1}^\infty \left\{ (b_k ^{+\prime} -b_k^{-\prime}) \cos(2\pi k\theta) +(-a_k^{+\prime} -a_k^{-\prime}) \sin(2\pi k\theta)\right\} .
\end{split}
\end{equation}}

Let us compute the two summands contributing to $\sigma$ separately. The first one reads:
{\small \begin{equation*}
\begin{split}
\i{\tilde{h}}\cdot \i{h'}&= \kappa(\tilde{h}_0)\cdot \i{h'} + \i{\tilde{h}}\cdot \kappa(h'_0)+ \tilde{n}\theta \cdot n'\theta + \tilde{n}\theta \cdot \iota (\phi') + \iota(\tilde{\phi})\cdot n'\theta + \iota(\tilde{\phi})\cdot \iota(\phi')\\
&=\kappa (\tilde{h}_0\smile \i{\cha {h'}})-\kappa (h_0'\smile \i{\cha \tilde{h}})-\iota (\phi'\tilde{n}\,d\theta)+\iota(\tilde{\phi}n'\,d\theta)+\iota(\tilde{\phi}\,d\phi ')\\
&=\kappa(n'\tilde{h}_0-\tilde{n}h_0')+\iota(-\tilde{n}\phi'\,d\theta+n'\tilde{\phi}\,d\theta+\tilde{\phi}\,d\phi')
\end{split}
\end{equation*}}
As far as the second one is concerned, we have:
{\small \begin{equation*}
\begin{split}
\i{h}\cdot \i{\tilde{h}'}&=(\kappa(h_0)+n\theta +\iota(\phi))\cdot (\kappa(\tilde{h}_0')+\tilde{n}'\theta+\iota(\tilde{\phi}'))\\
&=\kappa(h_0)\cdot \i{\tilde{h}'}-\kappa(\tilde{h_0}')\cdot \i{h}+n\theta\cdot \tilde{n}'\theta+n\theta \cdot \iota(\tilde{\phi}')+\iota(\phi)\cdot \tilde{n}'\theta+\iota(\phi)\cdot \iota(\tilde{\phi}')\\
&=\kappa(h_0\smile \i{\cha\tilde{h}'})-\kappa(\tilde{h}_0\smile \i{\cha h})+\iota(\tilde{\phi}'n\,d\theta)+\iota(\phi\tilde{n}'\,d\theta)+\iota(\phi \,d\tilde{\phi}') \\
&=\kappa(\tilde{n}'h_0-n\tilde{h}_0')+\iota(-n\tilde{\phi}'\,d\theta+\tilde{n}'\phi \,d\theta+\phi \,d\tilde{\phi}').
\end{split}
\end{equation*}}
What is left is just the evaluation on $\mu$, which, in our case, amounts to integrating over $\mathbb{S}^1$. Recall that, according to the adopted convention, the radius of $\mathbb{S}^1$ equals $(2\pi)^{-1}$. Retaining the non-vanishing terms only, we obtain:
{\small \begin{align}\label{eq:3:pre-symplectic_product_R_S1}
\sigma((h,\tilde{h}),(h',\tilde{h}')) &=\tilde{n}'h_0-n\tilde{h}_0'+n'\tilde{h}_0-\tilde{n}h_0' + \int_{\mathbb{S}^1} \phi \,d\tilde{\phi}' +\tilde{\phi}\,d\phi' \nonumber\\
 &= \tilde{n}'h_0-n\tilde{h}_0'+n'\tilde{h}_0-\tilde{n}h_0'+\nonumber\\
 &+ \bigg( \sum_{k=1}^\infty \pi k \left\{ (b_k ^{+} -b_k^{-})(a_k^{+\prime} -a_k^{-\prime})-(a_k^{+} +a_k^{-})(b_k^{-\prime} +b_k^{+\prime}) \right.\nonumber\\
 &\left. + (b_k ^{-}+b_k^{+})(a_k^{+\prime} + a_k^{-\prime}) - (a_k^{+} -a_k^{-})(b_k ^{+\prime} -b_k^{-\prime}) \right\}\bigg) \,\,\mod\Z\nonumber \\
 &= \tilde{n}'h_0-n\tilde{h}_0'+n'\tilde{h}_0-\tilde{n}h_0'+\nonumber \\
 & + \bigg( \sum_{k=1}^\infty 2\pi k \left\{ b_k^+ a_k^{+\prime} + b_k^- a_k^{-\prime} - a_k^+ b_k^{+\prime} - a_k^- b_k^{-\prime} \right\} \bigg)\,\,\mod\Z.
\end{align}}

\subsection{Hadamard State for $dC^\infty\cap \ast dC^\infty(\R\times \S^1)$} \label{section:3:Hadamard_2D}
We are now in the position to construct a quasifree Hadamard state. As we have no torsion, we have to build two states only, one for the topological sector $(\T\times\T)\times (\Z\times\Z)$ and one for the upper-right corner of diagram \eqref{diag:3:R_S1}, $dC^\infty \cap \ast dC^\infty(\R\times \S^1)$. Let us start from the latter.\\

In pursuing our goal, we will follow the prescriptions given in \cite[Section 4.3]{WAL94}, where a general procedure is described to build quasifree states in the case of stationary spacetimes. Since both in the case in hand and in the following examples we work with ultrastatic metrics, stationarity is always granted. The same procedure is discussed, from a more abstract and mathematical perspective, in \cite{ARA71a,ARA71b}. We also refer the reader to \cite{BR03} for additional comments and insights.\\

Let $\mathcal{S}:=dC^\infty \cap \ast dC^\infty(\R\times \S^1)$ and let $\mathcal{S}_\mathbb{C}=\mathcal{S}\otimes_\R \mathbb{C}$ be its complexification. Introduce the map:
\begin{eqnarray*}
&\tau_\mathbb{C}: \mathcal{S}_\mathbb{C} \times \mathcal{S}_\mathbb{C} \to \mathbb{C}\nonumber &\\
&(d\varphi_1, d\varphi_2)\mapsto \tau_u^\mathbb{C}(\overline{d\varphi_1},d\varphi_2), &
\end{eqnarray*}
where the overline denotes the complex conjugation. $\tau_u^\mathbb{C}$ is the sesquilinear extension of $\tau_u^\R$, i.e.\ the map \eqref{tau:3:u} where the $\mod\Z$ operation has been waived.  An arbitrary element $d\varphi\in\mathcal{S}$ can be decomposed as:
{\small \begin{align*}
d\varphi=\sum_{k=1}^\infty  & \bigg[ \left( a_k^{(0)} e^{2\pi ik(t+\theta)} + b_k^{(0)} e^{-2\pi ik(t+\theta)} + c_k^{(0)} e^{2\pi ik(t-\theta)} +d_k^{(0)} e^{-2\pi ik(t-\theta)}\right) dt \\
& + \left( a_k^{(1)} e^{2\pi ik(t+\theta)} + b_k^{(1)} e^{-2\pi ik(t+\theta)} + c_k^{(1)} e^{2\pi ik(t-\theta)} +d_k^{(1)} e^{-2\pi ik(t-\theta)} \right) d\theta \bigg].
\end{align*}}
Notice that not all the coefficients are independent, as a consequence of the fact that $d\varphi$ must be in $\mathcal{S}$. We define the map:
\begin{eqnarray}
& \mathcal{P}: \mathcal{S} \to \mathcal{S}_\mathbb{C}&\nonumber \\
& d\varphi \mapsto d\varphi^+, &
\end{eqnarray}
where:
{\small 
$$
d\varphi^+:=  \sum_{k=1}^\infty e^{-2\pi ikt}\bigg[ \left( b_k^{(0)} e^{-2\pi ik\theta} + d_k^{(0)} e^{2\pi ik\theta} \right) dt + \left( b_k^{(1)} e^{-2\pi ik\theta} + d_k^{(1)} e^{2\pi ik\theta} \right) d\theta \bigg].
$$}
The map $\mathcal{P}$ can be physically interpreted as the projection on the ``positive frequency" subspace of $\mathcal{S}_\mathbb{C}$. In fact, it suppresses the coefficients of $e^{ikt}$ while it preserves the part in $e^{-ikt}$. The map $\mu$ defined by:
\begin{eqnarray*}
& \mu: \mathcal{S}\times \mathcal{S}\to \mathbb{R}&\\
& (d\varphi_1, d\varphi_2)\mapsto \Im\, \tau_\mathbb{C}({\mathcal{P} d\varphi_1},\mathcal{P}d\varphi_2)&
\end{eqnarray*}
is a positive, symmetric bilinear map and it fulfils the inequality:
\begin{equation}\label{eq:3:constraint_mu}
\dfrac{1}{2}|\tau_u^\R(d\varphi_1,d\varphi_2)|\leq |\mu(d\varphi_1,d\varphi_1)|^{\frac{1}{2}} |\mu(d\varphi_2,d\varphi_2)|^{\frac{1}{2}}.
\end{equation}
In view of \cite[Equation (3.24)]{KW91}, the quasifree state $\omega_\mu$ associated with $\mu$ is obtained by the prescription:
\begin{eqnarray*}
&\omega_\mu : \mathfrak{CCR}(\mathcal{S},\tau_u) \to \mathbb{C}\nonumber &\\
&\mathcal{W}(d\varphi)\mapsto \omega_\mu(\mathcal{W}(d\varphi))= e^{-\frac{1}{2}\mu(d\varphi ,d\varphi)}. &
\end{eqnarray*}

Let us carry out the computation explicitly. With the notation adopted in Section \ref{section:3:differential_characters}, replace the trigonometric functions in \eqref{eq:3:h_phi} and \eqref{eq:3:h_phi_tilde} with their expression in terms of complex exponentials:
\begin{equation*}
\begin{split}
\varphi= \sum_{k=1}^\infty \dfrac{1}{2} \bigg\{ & e^{2\pi ikt} \bigg[ \left(-b_k^+ - i a_k^+ \right) e^{2\pi ik\theta} + \left( -b_k^- -i a_k^- \right) e^{-2\pi ik\theta}\bigg] \\
 +  & e^{-2\pi ikt} \Big[ \left( -b_k^- +ia_k^- \right)e^{2\pi ik\theta} + \left(-b_k^+ + ia_k^+ \right) e^{-2\pi ik\theta}\bigg] \bigg\},
\end{split}
\end{equation*}
\begin{equation*}
\begin{split}
\tilde{\varphi}= \sum_{k=1}^\infty \dfrac{1}{2} \bigg\{ & e^{2\pi ikt} \bigg[ \left(b_k^+ + i a_k^+ \right) e^{2\pi ik\theta} + \left( -b_k^- -i a_k^- \right) e^{-2\pi ik\theta}\bigg] \\
 +  & e^{-2\pi ikt} \Big[ \left( -b_k^- +ia_k^- \right)e^{2\pi ik\theta} + \left(b_k^+ - ia_k^+ \right) e^{-2\pi ik\theta}\bigg] \bigg\}.
\end{split}
\end{equation*}
The exterior derivative yields:
\begin{equation*}
\begin{split}
d\varphi= \sum_{k=1}^\infty i\pi k \bigg\{ 
& e^{2\pi ikt} \left[ \left(-b_k^- -i a_k^- \right) e^{-2\pi ik\theta} + \left( -b_k^+ -ia_k^+ \right)e^{2\pi ik\theta }\right] dt \\
&+ e^{2\pi ikt} \left[ \left( b_k^-+ia_k^- \right) e^{-2\pi ik\theta}+ \left(-b_k^+ -ia_k^+ \right) e^{2\pi ik\theta} \right] d\theta\\
&+ e^{-2\pi ikt} \left[ \left(b_k^+ -i a_k^+ \right) e^{-2\pi ik\theta} + \left( b_k^- -ia_k^- \right)e^{2\pi ik\theta} \right] dt \\
&+ e^{-2\pi ikt}  \left[ \left( b_k^+-ia_k^+ \right) e^{-2\pi ik\theta}+ \left(-b_k^- +ia_k^- \right) e^{2\pi ik\theta} \right] d\theta \bigg\}, 
\end{split}
\end{equation*}
\begin{equation*}
\begin{split}
d\tilde{\varphi}= \sum_{k=1}^\infty i\pi k \bigg\{ 
& e^{2\pi ikt} \left[ \left(-b_k^- -i a_k^- \right) e^{-2\pi ik\theta} + \left( +b_k^+ +ia_k^+ \right)e^{2\pi ik\theta }\right] dt \\
&+ e^{2\pi ikt} \left[ \left( b_k^-+ia_k^- \right) e^{-2\pi ik\theta}+ \left(b_k^+ +ia_k^+ \right) e^{2\pi ik\theta} \right] d\theta\\
&+ e^{-2\pi ikt} \left[ \left(-b_k^+ +i a_k^+ \right) e^{-2\pi ik\theta} + \left( b_k^- -ia_k^- \right)e^{2\pi ik\theta} \right] dt \\
&+ e^{-2\pi ikt}  \left[ \left( -b_k^++ia_k^+ \right) e^{-2\pi ik\theta}+ \left(-b_k^- +ia_k^- \right) e^{2\pi ik\theta} \right] d\theta \bigg\}.
\end{split}
\end{equation*}
It ensues that:
\begin{align*}
& \tau_\mathbb{C}(\mathcal{P}(d\varphi=\ast d\tilde{\varphi}), \mathcal{P}(d\varphi '=\ast d\tilde{\varphi}'))=\\
& = \int_{\{0\}\times\S^1} \overline{{\mathcal{P}\tilde{\varphi}}}\wedge {\mathcal{P}d\varphi'} + \overline{{\mathcal{P}\varphi}}\wedge {\mathcal{P} d\tilde{\varphi}'}\\
& = \sum_{k=1}^\infty i\dfrac{\pi k}{2} \bigg\{ \left( ia_k^-+b_k^- \right) \left( -ia_k^{-\prime} +b_k^{-\prime}\right) + \left( ia_k^+ + b_k^+ \right) \left( -ia_k^{+\prime} +b_k^{+\prime} \right)\\
&\qquad \quad+ \left(ia_k^- + b_k^- \right)\left( -ia_k^{-\prime} + b_k^{-\prime}\right) + \left( -ia_k^+ -b_k^+\right)\left( ia_k^{+\prime} -b_k^{+\prime} \right) \bigg\}\\
&=\sum_{k=1}^\infty \pi k \bigg\{ i \left( a_k^+a_k^{+\prime} + b_k^+b_k^{+\prime} +a_k^-a_k^{-\prime}+b_k^-b_k^{-\prime} \right)\\
&\qquad \quad - \left( a_k^-b_k^{-\prime} -a_k^{-\prime}b_k^- + a_k^+b_k^{+\prime}- a_k^{+\prime}b_k^+ \right)\bigg\}\\
&= \sum_{k=1}^\infty \pi k
\begin{pmatrix}
b_k^+ & a_k^+ & b_k^- & a_k^-
\end{pmatrix}
\begin{pmatrix}
i & 1 & 0 & 0 \\
-1 & i & 0 & 0 \\
0&0& i & 1 \\
0&0& -1 & i 
\end{pmatrix}
\begin{pmatrix}
b_k^{+\prime}\\
a_k^{+\prime}\\
b_k^{-\prime}\\
a_k^{-\prime}
\end{pmatrix}.
\end{align*}
Hence:
\begin{align}\label{eq:3:mu_R_S1}
\mu(d\varphi, d\varphi')&= \Im\, \tau_\mathbb{C}(\mathcal{P}(d\varphi=\ast d\tilde{\varphi}), \mathcal{P}(d\varphi '=\ast d\tilde{\varphi}'))\nonumber\\
&=\sum_{k=1}^\infty \pi k \bigg\{ a_k^+a_k^{+\prime} + b_k^+b_k^{+\prime} +a_k^-a_k^{-\prime}+b_k^-b_k^{-\prime} \bigg\}.
\end{align}
The state $\omega_\mu$ takes the form:
\begin{equation}\label{eq:3:state_2D}
\omega_\mu(\mathcal{W}(d\varphi))=\exp \left( -\dfrac{1}{2}\sum_{k=1}^\infty \pi k\bigg\{ (a_k^+)^2 + (b_k^+)^2 + (a_k^-)^2 +(b_k^-)^2 \bigg\}  \right).
\end{equation}

\begin{rem}
The state $\omega_\mu$ obtained this way enjoys the further property of being pure, namely extremal in the convex set of states; in fact, it can be proved that $\mu$ satisfies:
$$
\mu(d\varphi_1,d\varphi_1)=\dfrac{1}{4}\sup _{d\varphi_2 \neq 0} \dfrac{|\tau_u^\R(d\varphi_1,d\varphi_2)|^2}{\mu(d\varphi_2,d\varphi_2)},
$$
which is one of the equivalent criteria to assess whether a quasifree state is pure (cfr. \cite[Equations (3.34)-(3.35)]{KW91}). As a consequence, the associated GNS representation is irreducible.
\end{rem}
Define:
\begin{equation}\label{eq:3:ansatz_two_point}
\omega_{2}(d\varphi, d\varphi')= \mu(d\varphi, d\varphi) + \dfrac{i}{2} \tau_u^\R(d\varphi, d\varphi ').
\end{equation}
We say that $\omega_\mu$ is a Hadamard state in a weak sense if there exists a well-defined two-point function $\tilde{\omega}_2$ with singularity structure of Hadamard form that restricts to $\omega_{2}$ on $\mathcal{S}\times \mathcal{S}$.\\
Introduce the continuous map:
\begin{eqnarray}
&\mathcal{F}: \Omega^1_{c}(M) \to \Omega^1_{c}(M)&\nonumber\\
&\psi' \mapsto \psi=\psi' - \int_{\S^1} \psi' .&
\end{eqnarray}
Define the pre-symplectic map:
\begin{eqnarray}
&\tilde{\sigma}: G_1\left[ \mathcal{F}(\Omega_c^1(M))\right] \times G_1\left[ \mathcal{F}(\Omega_c^1(M))\right]\to \mathbb{R}&\nonumber\\
&(G_1\tilde{\psi}, G_1\psi) \mapsto \int_M \tilde{\psi}\wedge \ast G_1 \psi,&
\end{eqnarray}
where $G_1$ is the causal propagator for the d'Alembert-de Rham operator on $1$-forms. Our choice of the ultrastatic metric entails that $G_1$ is block-diagonal, i.e.\
\begin{eqnarray*}
&G_1: \Omega_0^1(M) \mapsto \Omega^1(M)&\\
&\psi=\psi_0\, dt + \psi_1\, d\theta\mapsto (G_0 \psi_0)\,dt+ (G_0\psi_1)\,d\theta,&\\
\end{eqnarray*}
with
$$
G_0(t,\theta)=t + \sum_{k=1}^\infty \dfrac{\sin(2\pi k t)\cos (2\pi k\theta)}{\pi k}.
$$
The effect of $\mathcal{F}$ is to wipe out the zero mode, in order to prevent the infrared divergences otherwise present in the two-dimensional case \cite{SHU13}. 

Following \cite{WAL94}, we can repeat the same passages as above with $\tilde{\sigma}$ in place of $\tau_u$. Writing $\psi=\psi_0\,dt + \psi_1 \,d\theta \in \Omega_c^1(M)$, we have:
\begin{align}\label{eq:3:G_1psi}
G_1\psi= \sum_{k=1}^\infty \dfrac{1}{4\pi i k}\bigg\{ & e^{2\pi i k t}\left[ c_{k,0}^+e^{2\pi i k\theta} + c_{k,0}^- e^{-2\pi i k \theta} \right] \nonumber\\
+&  e^{-2\pi i k t} \left[ d_{k,0}^+e^{2\pi i k\theta} + d_{k,0}^- e^{-2\pi i k \theta} \right]\bigg\} \,dt \nonumber\\
+\sum_{k=1}^\infty \dfrac{1}{4\pi i k} \bigg\{ & e^{2\pi i k t}\left[ c_{k,1}^+e^{2\pi i k\theta} + c_{k,1}^- e^{-2\pi i k \theta} \right] \nonumber\\
+&  e^{-2\pi i k t} \left[ d_{k,1}^+e^{2\pi i k\theta} + d_{k,1}^- e^{-2\pi i k \theta} \right]\bigg\} \,d\theta,
\end{align}
where
\begin{align*}
& c_{k,s}^+= \int_\R dt \int_{\S^1} d\theta\, e^{-2\pi i k t'} e^{-2\pi i k\theta'} \psi_s(t',\theta '),\\
& c_{k,s}^-= \int_\R dt \int_{\S^1} d\theta\, e^{-2\pi i k t'} e^{+2\pi i k\theta'} \psi_s(t',\theta '),\\
& d_{k,s}^+= -\int_\R dt \int_{\S^1} d\theta\, e^{+2\pi i k t'} e^{-2\pi i k\theta'} \psi_s(t',\theta '),\\
& d_{k,s}^-= -\int_\R dt \int_{\S^1} d\theta\, e^{+2\pi i k t'} e^{+2\pi i k\theta'} \psi_s(t',\theta ').
\end{align*}
One realises that the coefficients are not independent, but they fulfil the following relations:
$$
d_{k,0}^+=-\overline{c_{k,0}^-}, \qquad d_{k,0}^-=-\overline{c_{k,0}^+}, \qquad d_{k,1}^+=-\overline{c_{k,1}^-}, \qquad d_{k,1}^-=-\overline{c_{k,1}^+}.
$$
Carrying out the computation, we find:
\begin{equation}\label{eq:4:mu_4D}
\mu(G_1\tilde{\psi},G_1\psi)= \sum_{k=1}^\infty \dfrac{1}{4\pi k} \left\{\overline{\tilde{d}_{k,0}^+}d_{k,0}^+  + \overline{\tilde{d}_{k,0}^-}d_{k,0}^- + \overline{\tilde{d}_{k,1}^+}d_{k,1}^+ + \overline{\tilde{d}_{k,1}^-}d_{k,1}^- \right\}.
\end{equation}
We define:
$$
\tilde{\omega}_2(G_1\tilde{\psi},G_1\psi):=\tilde{\mu}(G_1\tilde{\psi}, G_1\psi)+ \dfrac{i}{2}\tilde{\sigma}(G_1\tilde{\psi}, G_1\psi),
$$
and we call $\tilde{\omega}_\mu$ the quasifree state having $\tilde{\omega}_2$ as two-point function. As $\tilde{\omega}_2(G_1\arg,G_1\arg)$ is a well-defined bi-distribution, it makes sense to investigate the structure of its singularities.

\begin{prop}\label{prop:3:omega_2D_is_Hadamard}
$\tilde{\omega}_\mu$ is a Hadamard state.
\end{prop}
\begin{proof}
In order to prove the statement, we show that $\tilde{\omega}_\mu$ is a ground state. 

Following \cite{SV00}, we call \emph{ground state} for $(\mathcal{A}, \{\alpha_t\}_{t\in\R})$, where $\mathcal{A}$ is a $C^\ast$-algebra and $\{\alpha_t\}_{t\in\R}$ a one-parameter group of automorphisms of $\mathcal{A}$, a state $\omega$ such that, for each $\mathcal{W}(a),\mathcal{W}(b)\in\mathcal{A}$, the function:
$$
\R \ni t\mapsto \omega \left(\mathcal{W}(a)\alpha_t(\mathcal{W}(b)) \right) 
$$
is bounded and it holds:
\begin{equation}\label{eq:3:integral_Sahlmann_Verch_ground}
\int_{-\infty}^{+\infty} \hat{f}(t)\,\omega\left(\mathcal{W}(a)\alpha_t(\mathcal{W}(b)) \right)=0
\end{equation}
for all $f\in C^\infty_c((-\infty,0))$.\\

Let $\mathcal{A}=\mathfrak{CCR}(G_1[\mathcal{F}(\Omega^1_{c}(M))],\tilde{\sigma})$ and introduce the map:
\begin{eqnarray*}
&\beta_t: G_1[\mathcal{F}(\Omega^1_{c}(M))] \to G_1[\mathcal{F}(\Omega^1_{c}(M))]&\\
&G_1\psi\mapsto \left( (s,\theta)\mapsto G_1\psi (s+t,\theta) \right).&
\end{eqnarray*}
The one-parameter group of automorphism is defined by:
$$
\alpha_t(\mathcal{W}(G_1 \psi)):=\mathcal{W}(\beta_t (G_1 \psi)).
$$

Let us check the boundedness property by direct computation. In view of the algebra structure, we have:
\begin{equation*}
\tilde{\omega}_\mu(\mathcal{W}(G_1\tilde{\psi})\alpha_t(\mathcal{W}(G_1\psi)))=e^{2\pi i \tilde{\sigma}(G_1\tilde{\psi},\beta_t(G_1\psi))} \tilde{\omega}_\mu(\mathcal{W}(G_1\tilde{\psi}+\beta_t(G_1\psi))).
\end{equation*}
Let us consider the second factor in the right-hand side; the complex exponential can then be handled the same way. Looking at \eqref{eq:3:G_1psi}, we realise that the Fourier coefficients for $G_1\tilde{\psi}+\beta_t(G_1\psi)$ are:
\begin{align*}
D_{k,j}^{\pm}=\tilde{d}_{k,j}^{\pm}+e^{-2\pi i k t}d_{k,j}^{\pm}, \qquad\quad j=0,1.
\end{align*}
In order to evaluate the state, we need the square modulus of the coefficients:
\begin{align*}
\overline{D_{k,j}^\pm}D_{k,j}^\pm=& \overline{\tilde{d}_{k,j}^\pm}\tilde{d}_{k,j}^\pm + \overline{d_{k,j}^\pm}d_{k,j}^\pm + e^{-2\pi i k t}\overline{\tilde{d}_{k,j}^\pm}d_{k,j}^\pm + e^{2\pi i k t}\overline{d_{k,j}^\pm}\tilde{d}_{k,j}^\pm\\
=&\overline{\tilde{d}_{k,j}^\pm}\tilde{d}_{k,j}^\pm + \overline{d_{k,j}^\pm}d_{k,j}^\pm + i \sin(2\pi k t)\left[\overline{d_{k,j}^\pm}\tilde{d}_{k,j}^\pm-\overline{\tilde{d}_{k,j}^\pm}d_{k,j}^\pm\right]\\ 
&+ \cos(2\pi k t)\left[\overline{d_{k,j}^\pm}\tilde{d}_{k,j}^\pm+\overline{\tilde{d}_{k,j}^\pm}d_{k,j}^\pm\right].
\end{align*}
On account of \eqref{eq:4:mu_4D}, it ensues that:
\begin{align*}
&\tilde{\omega}_\mu(\mathcal{W}(G_1\tilde{\psi}+\beta_t(G_1\psi)))=\tilde{\omega}_\mu(\mathcal{W}(G_1\tilde{\psi}))\,\,\tilde{\omega}_\mu(\mathcal{W}(G_1\psi))\times\\
&\times \exp\left(-i\sum_{k=1}^\infty\sum_{j=0}^1\frac{1}{8\pi k}\sin(2\pi k t) \left[\overline{d_{k,j}^+}\tilde{d}_{k,j}^+-\overline{\tilde{d}_{k,j}^+} d_{k,j}^+ + \overline{d_{k,j}^-}\tilde{d}_{k,j}^- -\overline{\tilde{d}_{k,j}^-}d_{k,j}^- \right] \right)\\
&\times \exp \left( -\sum_{k=1}^\infty\sum_{j=0}^1\frac{1}{8\pi k}\cos(2\pi k t) \left[\overline{d_{k,j}^+}\tilde{d}_{k,j}^+ +\overline{\tilde{d}_{k,j}^+} d_{k,j}^+ + \overline{d_{k,j}^-}\tilde{d}_{k,j}^- +\overline{\tilde{d}_{k,j}^-}d_{k,j}^- \right] \right).
\end{align*}
The part independent of $t$ can be left aside, as it does not affect boundedness. As for the exponentials containing the sine and the cosine functions, they can be thought of as a real function which is continuous and periodic in $t$. Consequently, it must be bounded. Observe that the series at the exponent is convergent, since the multiplication by the sine or by the cosine does not worsen the convergence properties.\\


As far as the integral \eqref{eq:3:integral_Sahlmann_Verch_ground} is concerned, in view of the state being quasifree, it is enough to prove that it is vanishing when we replace $\tilde{\omega}_\mu$ with $\tilde{\omega}_2$. For every $\psi=\mathcal{F}\psi', \tilde{\psi}=\mathcal{F}\tilde{\psi}' \in \Omega^1_c(M)$, we have:
\begin{align*}
&\int_{-\infty}^{+\infty} \hat{f}(t)\, \tilde{\omega}_2(G_1 \tilde{\psi}, \beta_t(G_1\psi)) \,dt = \int_{-\infty}^{+\infty} dt\, \left(  \int_{-\infty}^{+\infty} dE\, e^{-2\pi iEt} f(E) \right) \,\times \\
&\times \left( \int_{M\times M} \tilde{\omega}_2^{ii}(\theta-\theta',s-s')\tilde{\psi}_i(\theta,s)\psi_i(\theta',s'+t)\, d\mu(M)d\mu'(M)\right) \\
=&\int_{-\infty}^{+\infty}dt \int_{-\infty}^{+\infty} dE\, e^{-2\pi iEt}f(E) \int_{M\times M} d\theta d\theta' ds d\eta \, \times\\
&\times \left( \int_{-\infty}^{+\infty}dk\, e^{2\pi i k (s-\eta+t)} \dot{\hat{\tilde{\omega}}}^{ii}_2(\theta-\theta',k) \right)\,\tilde{\psi}_i(\theta,s)\psi_i(\theta',\eta)  \\
=& \int_{-\infty}^{+\infty}dE f(E)\int_{M\times M} d\theta d\theta' ds d\eta\times\\
&\times \int_{-\infty}^{+\infty}dk\,   e^{2\pi ik(s-\eta)}  \dot{\hat{\tilde{\omega}}}^{ii}_2(\theta-\theta',k) \tilde{\psi}_i(\theta,s)\psi_i(\theta',\eta) \int_{-\infty}^{+\infty}dt\, e^{-2\pi i(E-k)t}\\
=& \int_{-\infty}^{+\infty} dE\, f(E) \delta(E-k) \int_{M\times M} d\theta d\theta' ds d\eta\times\\
&\times \int_{-\infty}^{+\infty}dk\,  e^{2\pi ik(s-\eta)}  \dot{\hat{\tilde{\omega}}}^{ii}_2(\theta-\theta',k) \tilde{\psi}_i(\theta,s)\psi_i(\theta',\eta)\\
=& \int_{-\infty}^{+\infty} dE\, f(E) \int_{M\times M} d\theta d\theta' \,\dot{\hat{\tilde{\omega}}}^{ii}_2(\theta-\theta',E) \dot{\hat{\tilde{\psi}}}(\theta,-E)\dot{\hat{\psi}}(\theta',E) =0
\end{align*}
where $\theta$ and $\theta'$ are local coordinates on $\S^1$, $s$ and $s'$ coordinates on $\R$, $d\mu(M)$ is the standard measure on $M$ and where the Einstein notation has been adopted. The intex $i$ can take the value $0$ or $1$ and it refers to the components of the one-forms. The dot over the Fourier transform symbol points out that we have Fourier-transformed the time component only. In the last passage our conclusion is due to the fact that in constructing the state we selected the positive frequencies only, whereas $f$ is supported in $(-\infty,0)$: the action of $\delta(E-k)$ annihilates the integral.

We then conclude by using Sahlmann and Verch's result stating that each ground state in a globally hyperbolic, stationary spacetime descending from a well-defined bi-distribution fulfils the passivity condition and is, consequently, Hadamard \cite{SV00}.
\end{proof}

\begin{rem}
Notice that we could have proved that $\tilde{\omega}_\mu$ is Hadamard by computing directly the wavefront set of its two-point function. In fact, putting together the above ingredients, a straightforward calculation gives that the two-point function reads:
{\small \begin{equation}
\begin{split}
\tilde{\omega}_2(\tilde{\psi},\psi)=&\sum_{k=1}^\infty \dfrac{1}{2\pi k}\int dt\, d\theta\int dt'\,d\theta' \, e^{-2\pi i k(t-t')}\cos\left( 2\pi k(\theta - \theta ') \right) \tilde{\psi}_0(t,\theta)\psi_0(t',\theta')\\
+&\sum_{k=1}^\infty \dfrac{1}{2\pi k}\int dt\, d\theta\int dt'\,d\theta' \, e^{-2\pi i k(t-t')}\cos\left( 2\pi k(\theta - \theta ') \right) \tilde{\psi}_1(t,\theta)\psi_1(t',\theta')
\end{split}
\end{equation}}
which explicitly satisfies the microlocal spectrum condition.
\end{rem}

\subsection{State for the topological sector}\label{subsection:3.3.4_state_topological_sector}
Finally, let us construct a state on 
\begin{equation}
\mathcal{A}_{lr}:=\mathfrak{CCR}\left(\frac{H^{k-1,m-k-1}(M;\R)}{H_{\free}^{k-1,m-k-1}(M;\Z)}\times H_{\free}^{k,m-k}(M;\Z),\tau_{lr}\right).
\end{equation}
To this end, there are several possibilities. We will provide, with quite a general construction, two different examples: a faithful state and a non-faithful one. The latter will turn out to possess interesting properties at the level of the associated GNS Hilbert space.

\begin{prop}
The continuous linear functional specified by:
\begin{eqnarray}\label{eq:3:state_topol_0}
& \omega^0_t: \mathcal{A}_{lr} \to \mathbb{C}& \nonumber \\
& \mathcal{W}(u,\tilde{u},v,\tilde{v}) \mapsto
\begin{cases}
1\qquad \mathrm{if}\,\,u=\tilde{u}=0, v=\tilde{v}=0\\
0 \qquad \mathrm{otherwise}
\end{cases} &
\end{eqnarray}
is a faithful state.
\end{prop}
\begin{proof}
The functional $\omega^0_t$ is normalized, being $\mathcal{W}(0,0,0,0)$ the unit of the algebra. Furthermore, it is positive; in fact, for every finite linear combination $a=\sum_i \alpha_i \mathcal{W}(u_i,\tilde{u}_i, v_i, \tilde{v}_i)\in \mathcal{A}_{lr}$, we have:
\begin{align*}
\omega_t^0(a^\ast a)=\omega_t^0\left(  \sum_{ij} \overline{\alpha_i}\alpha_j \mathcal{W}(u_i,\tilde{u}_i, v_i, \tilde{v}_i)^\ast  \mathcal{W}(u_j,\tilde{u_j},v_j, \tilde{v}_j) \right)= \sum_i \overline{\alpha_i} \alpha_i\geq 0,
\end{align*}
where we assumed that $(u_i,\tilde{u}_i, v_i, \tilde{v}_i)\neq(u_j,\tilde{u}_j, v_j, \tilde{v}_j)$ whenever $i\neq j$. It ensues that $\omega_t^0(a^\ast a)=0$ if and only if $\alpha_i=0\,\,\forall i$, i.e.\ the state is faithful.
\end{proof}

\begin{prop}\label{prop:3:omega_t}
The continuous linear functional specified by:
\begin{eqnarray}\label{eq:3:state_topol}
& \omega_t:  \mathcal{A}_{lr} \to \mathbb{C} &\nonumber\\
& \mathcal{W}(u,\tilde{u},v,\tilde{v}) \mapsto
\begin{cases}
1\qquad \text{if}\,\,v=\tilde{v}=0\\
0 \qquad \text{otherwise}
\end{cases} &
\end{eqnarray}
is a state.
\end{prop}
\begin{proof}
Let $\mathcal{I}$ be an index set of finite cardinality and let 
$$
a=\sum_{i\in \mathcal{I}} \alpha_i \mathcal{W}(u_i, \tilde{u}_i, v_i, \tilde{v}_i)
$$ 
be an element in $\mathcal{A}_{lr}$, where we assume $(u_i, \tilde{u}_i, v_i, \tilde{v}_i)\neq (u_j, \tilde{u}_j, v_j, \tilde{v}_j)$ for $i\neq j$, without loss of generality. 

As above, $\omega_t$ is normalized, due to $\mathcal{W}(0,0,0,0)$ being the unit of the algebra. 

As far as positivity is concerned, introduce in $\mathcal{I}$ the equivalence relation $\sim$ defined by: $i\sim j$ if and only if $v_i=v_j$ and $\tilde{v}_i=\tilde{v}_j$. Define $\tilde{\mathcal{I}}:=\mathcal{I}\slash\sim$. We have:
\begin{align*}
&\omega_t(a^\ast a)=\\
=& \omega_t\Bigg( \sum_{\substack{\tilde{i}\in \tilde{\mathcal{I}} \\ \tilde{j}\in \tilde{\mathcal{I}}}} \sum_{\substack{i\in \tilde{i} \\ j\in\tilde{j}}} \overline{\alpha_i} \alpha_j e^{-2\pi i (\pa{\tilde{u}_i}{v_j}_f -(-1)^{k(m-k)} \pa{u_i}{\tilde{v}_j}_f - \pa{\tilde{u}_j}{v_i}_f+(-1)^{k(m-k)}\pa{u_j}{ \tilde{v}_i}_f)}  \,\,\times\\
&\quad \times  \mathcal{W}(u_j-u_i, \tilde{u}_j-\tilde{u}_i, v_j-v_i, \tilde{v}_j-\tilde{v}_i) \Bigg)\\
=& \sum_{\substack{\tilde{i}\in \tilde{\mathcal{I}} \\ \tilde{j}\in \tilde{\mathcal{I}}}} \sum_{\substack{i\in \tilde{i} \\ j\in\tilde{j}}} \overline{\alpha_i} \alpha_j e^{-2\pi i (\pa{\tilde{u}_i}{v_j}_f -(-1)^{k(m-k)} \pa{u_i}{\tilde{v}_j}_f - \pa{\tilde{u}_j}{v_i}_f+(-1)^{k(m-k)}\pa{u_j}{ \tilde{v}_i}_f)} \,\,\delta_{\tilde{i}\tilde{j}}\\
=& \sum_{\substack{\tilde{i}\in \tilde{\mathcal{I}}}} \sum_{\substack{i,j \in \tilde{i}}} \overline{\alpha_i} \alpha_j e^{-2\pi i (\pa{\tilde{u}_i}{v_i}_f -(-1)^{k(m-k)} \pa{u_i}{\tilde{v}_i}_f - \pa{\tilde{u}_j}{v_j}_f+(-1)^{k(m-k)}\pa{u_j}{ \tilde{v}_j}_f)} \\
=& \sum_{\tilde{i}\in\tilde{\mathcal{I}}}\, \left| \sum_{i\in\tilde{i}} \alpha_i e^{2\pi i(\pa{\tilde{u_i}}{v_i}_f-(-1)^{k(m-k)} \pa{u_i}{\tilde{v_i}}_f)} \right|^2 \geq 0.
\end{align*}
\end{proof}

\begin{rem}
Notice that the prescription to build the above states is general: four-dimensional spacetimes with compact Cauchy surface are a particular instances.
\end{rem}

The state $\omega_t$ is not faithful. Let us try to characterize to what extent faithfulness fails. Let $(0,0,v,\tilde{v}),(u,\tilde{u},v,\tilde{v})\in \frac{H^{k-1,m-k-1}(M;\R)}{H_{\free}^{k-1,m-k-1}(M;\Z)}\times H_{\free}^{k,m-k}(M;\Z)$, let 
\begin{equation}\label{eq:3:elements_gelfand_ideal}
b=\mathcal{W}(0,0,v,\tilde{v})- e^{-2\pi i (\pa{\tilde{u}}{v}_f -(-1)^{k(m-k)} \pa{u}{\tilde{v}}_f)} \mathcal{W}(u,\tilde{u},v,\tilde{v})
\end{equation}
and consider the following computation:
\begin{align*}
 &\omega_t\left( b^\ast b\right)= \omega_t\left( \mathcal{W}(0,0,v,\tilde{v})^\ast \mathcal{W}(0,0,v,\tilde{v}) \right) + \omega_t\left( \mathcal{W}(u,\tilde{u},v,\tilde{v})^\ast \mathcal{W}(u,\tilde{u},v,\tilde{v})  \right) \\
&-  e^{-2\pi i (\pa{\tilde{u}}{v}_f -(-1)^{k(m-k)} \pa{u}{\tilde{v}}_f)}  \omega_t \left(e^{2\pi i (\pa{\tilde{u}}{v}_f -(-1)^{k(m-k)} \pa{u}{\tilde{v}}_f)} \mathcal{W}(u,\tilde{u},0,0) \right) \\
&- e^{2\pi i (\pa{\tilde{u}}{v}_f -(-1)^{k(m-k)} \pa{u}{\tilde{v}}_f)} \omega_t\left( e^{-2\pi i (\pa{\tilde{u}}{v}_f -(-1)^{k(m-k)} \pa{u}{\tilde{v}}_f)}\mathcal{W}(-u,-\tilde{u},0,0) \right)\\
&= 1+1-1-1=0.
\end{align*}
This result suggests that, from the point of view of the state, the elements of the algebra are identified as follows:
\begin{equation}
\mathcal{W}(u,\tilde{u},v,\tilde{v}) \sim e^{2\pi i (\pa{\tilde{u}}{v}_f -(-1)^{k(m-k)} \pa{u}{\tilde{v}}_f)}\mathcal{W}(0,0,v,\tilde{v}).
\end{equation}
We will return extensively to this point later.
\section{The 4-dimensional case}\label{section:3.4_4D_case}
The four-dimensional case is definitely the most interesting one from the perspective of physics, since four is the degree of the spacetimes thoroughly investigated by the theory of general relativity. In this section we will analyse in full detail the case of manifolds with compact Cauchy surface, making also reference to specific examples. Having in mind a duality argument, we set to two the dimension $k$ of the field theory. In conclusion, some remarks on the non-compact case will be given.\\

Considering that the existence of a compact Cauchy surface $\Sigma$ entails that $M=J(\Sigma)$, the assumption of spacelike compactness made for the observables is redundant and can be safely waived. Again, all of the manifolds appearing in the present section will be assumed to be endowed with an ultrastatic metric, unless otherwise stated. This, in particular, guarantees that all the three-dimensional spatial sections, namely the \emph{folia} of the spacetime, are Riemannian manifolds.\\

In all the relevant examples, as we shall show, the torsion subgroup is vanishing. We can therefore define a state for the covariant quantum field theory by constructing a state for the topological sector and a state for the upper-right corner.

\subsection{Hadamard state for $d\Omega^{k-1}(M)\cap \ast d\Omega^{m-k-1}(M)$}\label{subsection:4.3.1_Hadamard_state_general}
A considerable part of the construction we are going to present prescinds from the dimension of the spacetime; therefore, we will remain general as far as we can.\\

Let $M$ be an ultrastatic object of $\mathsf{Loc}_m$, with Cauchy surface $\Sigma$ and metric $g=-dt\otimes dt + h$, where $h$ is a Riemannian metric independend of $t$. Hodge theory provides us with the following results:

\begin{thm}[Hodge Theorem for forms {\cite[Theorem 1.30]{Ros97}}]\label{thm:3:Hodge_theorem}
Let $N$ be a compact, connected, oriented Riemannian manifold. There exists an orthonormal basis of the space of $L^2$ $k$-forms consisting of eigenforms of the Laplacian on $k$-forms:
$$
\Delta= d\delta+\delta d: \Omega^k(N)\to \Omega^k(N),
$$
where $d$ is the exterior derivative and $\delta$ the codifferential.
All the eigenvalues are non-negative. Each eigenvalue has finite multiplicity and the eigenvalues accumulate only at infinity.
\end{thm}

\begin{thm}[Hodge decomposition {\cite[Theorem 1.37]{Ros97}}]
Let $N$ be a compact, connected, oriented Riemannian manifold. Define the space of \emph{harmonic $k$-forms}:
$$
\mathcal{H}^k_\Delta:=\{\alpha \in \Omega^k(N)\,\,|\,\, \Delta \alpha=0 \}=\{ \alpha\in \Omega^k(N)\,\,|\,\, \alpha\in \mathrm{ker}(d)\cap \mathrm{ker}(\delta) \}.
$$
The following decomposition holds:
\begin{equation}\label{eq:3:hodge_decomposition}
\Omega^k(N)=\mathcal{H}^k_\Delta(N)\oplus d\Omega^{k-1}(N)\oplus \delta\Omega^{k+1}(N).
\end{equation}
The decomposition is orthogonal with respect to the natural metric-induced $L^2$ scalar product.
\end{thm}

We aim at exploiting these theorems to obtain a decomposition for a generic element $dA=\ast d\tilde{A}\in d\Omega^{k-1}\cap \ast d\Omega^{m-k-1}(M)$ in terms of an orthonormal basis. Let us state and prove a technical lemma first.

\begin{lem}\label{lem:3:hodge_dual_M_Sigma}
Let $\omega$ be a purely spatial $k$-form. Then:
\begin{enumerate}[(i)]
\item $\ast (dt\wedge \omega) = -\ast_\Sigma \omega$;
\item $\ast \omega=(-1)^k dt\wedge \ast_\Sigma \omega$,
\end{enumerate}
where $\ast_\Sigma$ denotes the Hodge dual on $\Sigma$.
\end{lem}
\begin{proof}
(i). By the definition of the Hodge dual, we have:
$$
 \omega \wedge \ast_\Sigma \omega = h(\omega,\omega)\mathrm{dVol}_\Sigma
$$
$$
dt\wedge \omega \wedge \ast (dt\wedge \omega)=-h(\omega,\omega) dt\wedge \mathrm{dVol}_\Sigma.
$$
It ensues that:
$$
\ast(dt\wedge \omega)=-\ast_\Sigma \omega.
$$
(ii). Again by the definition of the Hodge dual, we write:
\begin{align*}
\omega\wedge \ast\omega=& h(\omega,\omega) dt\wedge \mathrm{dVol}_\Sigma\\
=& dt\wedge \omega \wedge \ast_\Sigma \omega= (-1)^k \omega\wedge dt\wedge \ast_\Sigma\omega.
\end{align*}
\end{proof}

To avoid confusion, we will denote by $d_\Sigma$ the exterior derivative on the Cauchy surface $\Sigma$ and by $d_M$ the exterior derivative on $M$. We want to obtain an orthonormal basis for $d_{\Sigma}\Omega^{k-1}(\Sigma)$. Resorting to the Hodge decomposition \eqref{eq:3:hodge_decomposition}, we have:
$$
d_\Sigma\Omega^{k-1}(\Sigma)=d_\Sigma(\delta_\Sigma \Omega^{k}(\Sigma)).
$$
Via Theorem \ref{thm:3:Hodge_theorem}, choose an orthonormal set $\mathcal{U}\subset\delta_\Sigma\Omega^k(\Sigma)$, i.e.\ such that:
$$
\int_\Sigma {u}_i\wedge\ast {u}_j= \delta_{ij}
$$
and
$$
\Delta_\Sigma {u}_i=\lambda_i^2 {u}_i, \qquad \lambda_i>0.
$$
Then:
$$
d_\Sigma\Omega^{k-1}(\Sigma)=\overline{\mathrm{span}\{d_\Sigma u_i\,:\, u_i\in \mathcal{U}\}}.
$$
Furthermore, we observe that, by the properties of the Hodge decomposition, defining $\mathcal{V}:=\{\ast u_i\,:\,u_i\in\mathcal{U}\}$, it holds $\mathcal{V}\subset d_\Sigma\Omega^{m-k-1}(\Sigma)$ and:
$$
d_\Sigma\Omega^{m-k-1}(\Sigma)=\overline{\mathrm{span}\{ v_i\,:\,v_i\in\mathcal{V} \}}.
$$

\begin{prop}\label{prop:3:solution_four_1}
A solution $(d_MA_i,d_M\tilde{A}_i)\in d\Omega^{k-1,m-k-1}(M)$ to:
\begin{equation}\label{eq:3:cauchy_4d_1}
d_M A_i=\ast d_M \tilde{A}_i
\end{equation}
with initial data $(\iota^\ast_\Sigma d_M A_i= d_\Sigma u_i, \iota^\ast_\Sigma d_M\tilde{A}_i=0)$ is given by:
\begin{align*}
& d_MA_i=d_M \left( \cos(\lambda_i t) u_i \right),\\
& d_M\tilde{A}_i= d_M \left( (-1)^{km+1}\lambda_i^{-1}\sin(\lambda_i t)\ast_\Sigma d_\Sigma u_i \right).
\end{align*}
\end{prop}
\begin{proof}
The left-hand side in \eqref{eq:3:cauchy_4d_1} is given by:
\begin{align*}
d_M (\cos(\lambda_i t)= -\lambda_i \sin(\lambda_i u_i) dt\wedge u_i + \cos(\lambda_i t) du_i,
\end{align*}
while the right-hand side is given by:
\begin{align*}
&\ast_M d_M((-1)^{km+1}\lambda_i^{-1}\sin(\lambda_i t)\ast_\Sigma d_\Sigma u_i ) \\
&= \ast_M( (-1)^{km+1}\cos(\lambda_i t) dt\wedge \ast_\Sigma d_\Sigma u_i+ (-1)^{km+1}\lambda_i^{-1} \sin(\lambda_i t)d_\Sigma \ast_\Sigma d_\Sigma u_i)\\
&= - (-1)^{km+1}\cos(\lambda_i t) \ast_\Sigma \ast_\Sigma d_\Sigma u_i + (-1) \lambda_i^{-1} \sin(\lambda_i t) dt\wedge \Delta_\Sigma u_i\\
&= \cos(\lambda_i t) d_\Sigma u_i - \lambda_i \sin(\lambda_i t) dt\wedge u_i.
\end{align*}
The statement then follows.
\end{proof}

\begin{prop}\label{prop:3:solution_four_2}
A solution $(d_MA_j,d_M\tilde{A}_j)\in d\Omega^{k-1,m-k-1}(M)$ to:
\begin{equation}\label{eq:3:cauchy_4d_2}
d_M A_j=\ast_M d_M \tilde{A}_j
\end{equation}
with initial data $(\iota^\ast_\Sigma d_M A_j=0,\iota^\ast_\Sigma d_M\tilde{A}=\ast_\Sigma u_j)$ is given by:
\begin{align*}
& d_MA_j=d_M \left( (-1)^{k(m-k)}\lambda_j^{-1}\sin (\lambda_j t) u_j \right),\\
& d_M\tilde{A}_j= d_M \left( (-1)^k \lambda_j^{-2} \cos(\lambda_j t) \ast_\Sigma d_\Sigma u_j \right).
\end{align*}
\end{prop}
\begin{proof}
The left-hand side in \eqref{eq:3:cauchy_4d_2} is given by:
\begin{align*}
&d_M\left( (-1)^{k(m-k)}\lambda_j^{-1} \sin(\lambda_j t) u_j \right)= \\
&=(-1)^{k(m-k)}\cos(\lambda_j t) dt\wedge u_j + (-1)^{k(m-k)}\lambda_j^{-1}\sin(\lambda_j t) d_\Sigma u_j,
\end{align*}
whereas the right-hand side is given by:
\begin{align*}
& \ast_M d_M \left( (-1)^k \lambda_j^{-2} \cos(\lambda_j t)\ast_\Sigma d_\Sigma u_j \right)=\\
&= \ast_M \left( -(-1)^{k} \lambda_j^{-1}\sin(\lambda_j t) dt	\wedge \ast_\Sigma d_\Sigma u_j + (-1)^k \lambda_j^{-2} \cos(\lambda_j t) d_\Sigma \ast_\Sigma d_\Sigma u_j \right)\\
&= (-1)^k \lambda_j^{-1} \sin(\lambda_j t) \ast_\Sigma \ast_\Sigma d_\Sigma u_j + (-1)^{m+1} \lambda_j^{-2} \cos(\lambda_j t) dt\wedge \ast_\Sigma d_\Sigma \ast_\Sigma d_\Sigma u_j\\
&= (-1)^{km-k^2}\lambda_j^{-1} \sin(\lambda_j t) d_\Sigma u_j + (-1)^{km-k} \lambda_j^{-2} \cos(\lambda_j t) dt\wedge \Delta_\Sigma u_j\\
&= (-1)^{k(m-k)}\lambda_j^{-1} \sin(\lambda_j t) d_\Sigma u_j + (-1)^{k(m-k)} \cos(\lambda_j t) dt\wedge u_j.
\end{align*}
The proposition is proved.
\end{proof}

If we now consider arbitrary initial data $(d_\Sigma A_0,d_\Sigma \tilde{A}_0)\in d_\Sigma\Omega^{k-1,m-k-1}(\Sigma)$ we can find sequences of real coefficients $\{\alpha_i\}, \{\widetilde {\alpha}_i\}$ such that:
\begin{align*}
&d_\Sigma A_0=\sum_{i}  \alpha_i \lambda_i^{-1} d_\Sigma u_i, \\
&d_\Sigma \tilde{A}_0=\sum_i \tilde{\alpha}_i \ast u_i.
\end{align*}
On account of Proposition \ref{prop:3:solution_four_1} and Proposition \ref{prop:3:solution_four_2} it follows that the solutions corresponding to the given initial data are:
\begin{align*}
& d_M A =d_M \left\{ \sum_i \left[ \alpha_i \lambda_i^{-1} \cos(\lambda_i t) + (-1)^{k(m-k)}\tilde{\alpha}_i \lambda_i^{-1} \sin(\lambda_i t)\right] u_i  \right\} =\\
&=\ast_M d_M \left\{ \sum_{i} \left[ (-1)^{km+1}\alpha_i \lambda_i^{-2} \sin(\lambda_i t) + (-1)^k \tilde{\alpha}_i \lambda_i^{-2} \cos(\lambda_i t)\right] \ast_\Sigma d_\Sigma u_i   \right\}\\
&= \ast_M d_M \tilde{A}.
\end{align*}
Following the prescription discussed in the previous section to construct a quasifree state, we select the positive frequency part of $d_M A$ and of $d_M\tilde{A}$:
\begin{align*}
& \mathcal{P} d_M A= d_M \left\{  \sum_i \dfrac{1}{2} e^{-i \lambda_i t} \left[ \alpha_i \lambda_i^{-1} + (-1)^{k(m-k)} i \tilde{\alpha}_i \lambda_i^{-1} \right] u_i  \right\},\\
& \mathcal{P} d_M\tilde{A}= d_M \left\{ \sum_{i} \dfrac{1}{2} e^{-i\lambda_i t} \left[ - (-1)^{km} i \alpha_i \lambda_i^{-2} + (-1)^k \tilde{\alpha}_i \lambda_i^{-2}\right] \ast_\Sigma d_\Sigma u_i  \right\}.
\end{align*}
It ensues that:
\begin{align*}
& \mathcal{P}d_\Sigma A_0= \dfrac{1}{2}\sum_i \left[ \alpha_i \lambda_i^{-1} + i (-1)^{k(m-k)} \tilde{\alpha_i}\lambda_i^{-1} \right] d_\Sigma u_i,\\
& \mathcal{P}d_\Sigma \tilde{A}_0= \dfrac{1}{2} \sum_i \left[\tilde{\alpha_i} - i (-1)^{k(m-k)}  \alpha_i \right] \ast_\Sigma u_i.
\end{align*}

\begin{lem}\label{lem:3:integral_states}
The following equations hold:
\begin{enumerate}[(i)]
\item $\int_\Sigma u_i\wedge d_\Sigma \ast_\Sigma d_\Sigma u_j = (-1)^k \lambda_i^2 \delta_{ij}$
\item $\int_\Sigma \ast d_\Sigma u_i \wedge d_\Sigma u_j= (-1)^{mk} \lambda_i^2 \delta_{ij}.$
\end{enumerate}
\end{lem}
\begin{proof}
The proof is a straightforward application of Hodge theory. As far as (i) is concerned, we have:
\begin{align*}
\int_\Sigma u_i\wedge d_\Sigma \ast_\Sigma d_\Sigma u_j &=\int_\Sigma u_i\wedge (-1)^k \ast_\Sigma \delta_\Sigma d_\Sigma u_j\\
&= \int_\Sigma u_i \wedge \lambda_j^2 \ast_\Sigma u_j = (-1)^k\lambda_i^2 \delta_{ij}.
\end{align*}
As regards (ii), it holds:
\begin{align*}
\int_\Sigma \ast d_\Sigma u_i \wedge d_\Sigma u_j &= -(-1)^{m-k-1}\int_\Sigma d_\Sigma \ast_\Sigma d_\Sigma u_i \wedge u_j\\
&= (-1)^{m-k+(m-k)(k-1)}\int_\Sigma u_j\wedge d_\Sigma \ast_\Sigma d_\Sigma u_i= (-1)^{mk}\lambda_i^2 \delta_{ij}.
\end{align*}
\end{proof}

Using Lemma \ref{lem:3:integral_states}, we are now in a position to compute $\tau_\mathbb{C}$:
\begin{align*}
&\tau_\mathbb{C}(\mathcal{P}(d_MA=\ast_M d_M\tilde{A}), \mathcal{P}(d_MA'=\ast_Md_M \tilde{A}'))=\\
=& \int_\Sigma \overline{\mathcal{P}\tilde{A}}\wedge \mathcal{P}d_M A' -(-1)^{k(m-k)}\overline{\mathcal{P}A}\wedge \mathcal{P}d_M\tilde{A}\\
=& \int_\Sigma \dfrac{1}{4} \Bigg\{ \sum_{j'} \left[ (-1)^{km} i \alpha_{j'} \lambda_{j'}^{-2} + (-1)^k \tilde{\alpha}_{j'}\lambda_{j'}^{-2} \right]\times\\
&\times \sum_j \left[ \alpha'_j \lambda_j^{-2} +(-1)^{k(m-k)}i \tilde{\alpha}_j' \lambda_j^{-1} \right] \ast_\Sigma d_\Sigma u_{j'} \wedge d_\Sigma u_j \Bigg\} \\
&-(-1)^{k(m-k)}\int_\Sigma \dfrac{1}{4} \Bigg\{ \sum_{j'} \left[ \alpha_{j'}\lambda_{j'}^{-1} - (-1)^{k(m-k)} i \tilde{\alpha}_{j'} \lambda_{j'}^{-1} \right]\times\\
&\times \sum_j \left[ (-1)^{km+1} i {\alpha}_j'\lambda_j^{-2} + (-1)^k \tilde{\alpha}_j'\lambda_j^{-2} \right] u_{j'} \wedge d_\Sigma \ast_\Sigma d_\Sigma u_j\Bigg\}\\
=& \dfrac{1}{4}\sum_j \Bigg\{ \left( i\alpha_j  + (-1)^{k(m-k)}\tilde{\alpha}_j\right) \left( \alpha_j'\lambda_j^{-1} +(-1)^{k(m-k)}i\tilde{\alpha}_j'\lambda_j^{-1} \right) \\
& -(-1)^{k(m-k)} \left( \alpha_j\lambda_j^{-1} -(-1)^{k(m-k)}i\tilde{\alpha}_j \lambda_j^{-1} \right)\left( \tilde{\alpha}_j' \lambda_j-(-1)^{k(m-k)}i\alpha_j'\lambda_j \right)\Bigg\}\\
=& \dfrac{1}{2} \sum_j \left\{ i \lambda_j^{-1} \left[ \alpha_j \alpha_j' + \tilde{\alpha}_j\tilde{\alpha}_j' \right] + (-1)^{k(m-k)} \left[ \tilde{\alpha}_j\alpha_j'-\alpha_j\tilde{\alpha}_j' \right] \right\}.
\end{align*}
The positive-definite bilinear map $\mu$ given by:
\begin{align*}
\mu(d_M A, d_M A')=&\Im \tau_\mathbb{C}(d_M A=\ast_M d_M\tilde{A}, d_M A'=\ast_M d_M\tilde{A}')\\
=& \sum_j \dfrac{1}{2\lambda_j} \left( \alpha_j\alpha_j'+ \tilde{\alpha}_j \tilde{\alpha}_j' \right)
\end{align*}
defines the quasifree state:
\begin{eqnarray}\label{def:3:omega_mu}
&\omega_\mu: \mathfrak{CCR}(d_M\Omega^{k-1}\cap \ast_M d_M\Omega^{m-k-1}(M),\tau_u)\to \mathbb{C}& \nonumber\\
&\mathcal{W}(d_MA)\mapsto e^{-\frac{1}{4}\sum_i \frac{1}{\lambda_i}(\alpha_i^2 + \tilde{\alpha_i}^2)}.&
\end{eqnarray}
Recall that $\lambda_i >0$ for every $i$; the state is well-defined and, in particular, the exponent is convergent, on account of the fact that the Fourier expansion of the initial data is uniformly convergent.

We claim that the quasifree state $\omega_\mu$ is Hadamard in the weak sense described above. 
At variance with the two-dimensional case, we will just outline the procedure to obtain the state whose two-point function is a bi-distribution that restricts to $\omega_2$ on the subspace of interest, because there is no conceptual novelty confronted with a diverging amount of sterile computation. To construct such a state, consider the subspace of $\Omega^{k}(M)$ whose elements solve the wave equation $\square \omega=0$, namely $G_k(\Omega_c^k(M))$. By means of the Hodge decomposition theorem, it is possible to choose $L^2$ bases in $\delta_\Sigma\Omega^{k+1}(\Sigma)$, $d_\Sigma\Omega^{k-1}(\Sigma)$, $\mathcal{H}^k_\Delta(\Sigma)$ respectively. The initial data $\iota^\ast_\Sigma \omega=:\omega_\Sigma \in \Omega^k(\Sigma)$ can be decomposed with respect to such bases. We can then reconstruct the full solution in terms of the coefficients in the expansions of the initial data with respect to the chosen bases. Firstly, we define an analogous of the map $\mathcal{F}$, which projects the initial data on the subspace orthogonal to the harmonic forms, thus ruling out the zero mode.  Subsequently, resorting to the structure guaranteed by the ultrastatic metric, it is possible to separate out the positive frequency component and activate the machinery presented in Section \ref{section:3:Hadamard_2D} to define a positive-definite bilinear map, to which a quasifree state can be associated. If we choose consistently the basis in $d_\Sigma\Omega^{k-1}(\Sigma)$, upon identification of the coefficients of the decomposition, the state obtained this way reduces to $\omega_\mu$ on $d_M \Omega^{k-1}\cap \ast_M d_M \Omega^{m-k-1}(M)$. The proof to Proposition \ref{prop:3:omega_2D_is_Hadamard}, with minor adaptations, guarantees that it is, in addition, Hadamard.

\subsection{Examples}\label{subsection:3.4.2_examples}
In this section we will compute the diagram of the observables for a number of examples of four-dimensional spacetimes with compact Cauchy surface for $k=2$. In particular, we will observe that the torsion subgroup is always vanishing.

\begin{ex}[$M=\R\times \S^1\times \S^2$]
Noticing that the Cauchy surface is $\Sigma=\S^1\times \S^2$, let us calculate $H^{2,2}(\Sigma,\Z)$. By the K\"unneth formula we have:
\begin{align*}
H^2(\S^1\times \S^2;\Z)&=\oplus_{i=0}^2 H^i(\S^1;\Z)\otimes_\Z H^{2-i}(\S^2;\Z)\\
&= \Z\otimes_\Z \Z + 0 + 0 = \Z.
\end{align*}
Therefore, $H^{2,2}(\Sigma;\Z)=\Z^2$, which is free. Since $\Sigma$ is compact, it immediately ensues that:
$$
H^{1,1}(\Sigma;\T)=\left( H^{2,2}(\Sigma;\Z) \right)^\star= \T^2.
$$
The diagram of the observables then is:
\begin{equation}
\xymatrix{
& 0 \ar[d] & 0 \ar[d] & 0 \ar[d] & \\
0 \ar[r] & \T^2 \ar[d] \ar[r] & \dfrac{\Omega^{1,1}(\Sigma)}{\Omega^{1,1}_\Z(\Sigma)}\ar[r]\ar[d]& d\Omega^{1,1}(\Sigma) \ar[d] \ar[r] &0\\
0 \ar[r] & \T^2 \ar[d] \ar[r] & \hat{H}^{2,2}(\Sigma;\Z) \ar[r]\ar[d] & \Omega_\Z ^{2,2}(\Sigma) \ar[d] \ar[r] &0\\
0 \ar[r] & 0  \ar[d] \ar[r]& \Z^2 \ar[d] \ar[r]&\Z^2\ar[d] \ar[r] & 0\\
& 0 & 0 & 0 &
}
\end{equation}
\end{ex}

\begin{ex}[$M=\R\times \T^3$]
The K\"unneth formula gives:
\begin{align*}
H^{2}(\T^3;\Z)=&\oplus_{i=0}^2 H^i(\T;\Z) \otimes_\Z H^{2-i}(\T^2;\Z)\\
=& \Z \otimes_\Z \left( \oplus_{j=0}^2 H^j(\T;\Z)\otimes_\Z H^{2-j}(\T;\Z) \right) \oplus\\
&\oplus \Z \otimes \left( \oplus H^j(\T;\Z) \otimes_\Z H^{1-j}(\T;\Z) \right)\\
=& \Z \otimes_\Z \left( 0\oplus \Z \oplus 0 \right) \oplus \Z \otimes_\Z \left( \Z\oplus \Z \right) =\Z\oplus \Z^2= \Z^3.
\end{align*}
Therefore, $H^{2,2}(\Sigma;\Z)=\Z^6$ is free and, with the same reasoning, we find that $H^{1,1}(\Sigma;\T)=\left(H^{2,2}(\Sigma;\Z)\right)^\star=\T^6$.\\
The diagram of the observables becomes:
\begin{equation}
\xymatrix{
& 0 \ar[d] & 0 \ar[d] & 0 \ar[d] & \\
0 \ar[r] & \T^6 \ar[d] \ar[r] & \dfrac{\Omega^{1,1}(\Sigma)}{\Omega^{1,1}_\Z(\Sigma)}\ar[r]\ar[d]& d\Omega^{1,1}(\Sigma) \ar[d] \ar[r] &0\\
0 \ar[r] & \T^6 \ar[d] \ar[r] & \hat{H}^{2,2}(\Sigma;\Z) \ar[r]\ar[d] & \Omega_\Z ^{2,2}(\Sigma) \ar[d] \ar[r] &0\\
0 \ar[r] & 0  \ar[d] \ar[r]& \Z^6 \ar[d] \ar[r]&\Z^6\ar[d] \ar[r] & 0\\
& 0 & 0 & 0 &
}
\end{equation}
\end{ex}

\begin{ex}[$M=\R\times \S^3$]
In this case, the computation is trivial: the first column and the last row are vanishing, because $H^{2,2}(\S^3;\Z)=0$. Hence, there is no contribution from the topological sector. Actually, all the non-vanishing groups in the diagram are isomorphic: the diagram is nothing but the upper-right corner:
\begin{equation}
\xymatrix{
& 0 \ar[d] & 0 \ar[d] & 0 \ar[d] & \\
0 \ar[r] & 0 \ar[d] \ar[r] & \dfrac{\Omega^{1,1}(\Sigma)}{\Omega^{1,1}_\Z(\Sigma)}\ar[r]\ar[d]& d\Omega^{1,1}(\Sigma) \ar[d] \ar[r] &0\\
0 \ar[r] & 0 \ar[d] \ar[r] & \hat{H}^{2,2}(\Sigma;\Z) \ar[r]\ar[d] & \Omega_\Z ^{2,2}(\Sigma) \ar[d] \ar[r] &0\\
0 \ar[r] & 0  \ar[d] \ar[r]& 0 \ar[d] \ar[r]&0\ar[d] \ar[r] & 0\\
& 0 & 0 & 0 &
}
\end{equation}
\end{ex}

Looking at the diagrams, we realize that we already have at our disposal the procedure to construct a state for the topological sector. Consequently, we have all the tools to discuss the role and the properties of duality in our theory: this will constitute the topic of the fourth and last chapter. Before moving to such an issue, let us linger a little bit on the case of non-compact Cauchy surface.

\section{The non-compact case: comments and perspectives}\label{section:3.5_non_compact_case}
If we release the assumption of compact Cauchy surface, almost all of our constructions and results stop working. So far, we do not have a general theory comprising the totality of the cases of interest to physics. However, there are evidences seemingly indicating that the lack of compactness destines the analysis to be dealt with case by case.\\

The compactness of the spatial slices is a sufficient condition for the pre-symplectic orthogonal decomposition; nonetheless, it is not necessary. A more careful inspection of the proof of Proposition \ref{prop:3:existence_of_splitting} reveals that the actual weaker requirement is the isomorphism of the cohomology groups of the right degree. Namely:
$$
H^k_{c,\free}(\Sigma;\Z)\simeq H^{k}_{\free}(\Sigma;\Z), \qquad H^{m-k}_{c,\free}(\Sigma;\Z)\simeq H^{m-k}_\free(\Sigma;\Z).
$$
Also this requirement, however, is not necessary for a decomposition to exist. Consider, for instance, the spacetime $M=\R^2 \times \S^2$, endowed with an ultrastatic metric. In this case, $m=4$ and $k=2$. At the level of the Cauchy surface, we have:
\begin{align*}
H^1_c(\R\times \S^2;\T)=\left( H^2(\R\times \S^2;\Z) \right)^\star =(\Z)^\ast=\T,
\end{align*}
and
\begin{align*}
H^2_c(\R\times \S^2;\Z)=\left( H^1(\R\times \S^2;\T) \right)^\ast=\left( \Hom(H_1(\R\times \S^2);\T)\right)^\ast=0.
\end{align*}
In the last equation we made use of the universal coefficient theorem for cohomology, along with the fact that $\T$ is a divisible group, thus entailing that the Ext functor is vanishing. The diagram of the observables is therefore:
\begin{equation}
\xymatrix{
& 0 \ar[d] & 0 \ar[d] & 0 \ar[d] & \\
0 \ar[r] & \T^2 \ar[d] \ar[r] & \dfrac{\Omega^{1,1}_c(\Sigma)}{\Omega^{1,1}_{c,\Z}(\Sigma)}\ar[r]\ar[d]& d\Omega_c^{1,1}(\Sigma) \ar@/_1em/[l]_-a \ar[d] \ar[r] &0\\
0 \ar[r] & \T^2 \ar[d] \ar[r] & \hat{H}_c^{2,2}(\Sigma;\Z) \ar[r]\ar[d] & \Omega_{c,\Z} ^{2,2}(\Sigma) \ar[d] \ar[r] &0\\
0 \ar[r] & 0  \ar[d] \ar[r]& 0 \ar[d] \ar[r]&  0\ar[d] \ar[r] & 0\\
& 0 & 0 & 0 &
}
\end{equation}
We find that the first two rows are isomorphic, in view of the third being zero. To show that an orthogonal decomposition exists, we just have to split the first row. Let:
$$
a=a_1\times a_2: d\Omega^{1,1}_c(\Sigma;\Z) \to \dfrac{\Omega^{1,1}_c(\Sigma;\Z)}{\Omega^{1,1}_{c,\Z}(\Sigma;\Z)}
$$
be a splitting homomorphism, whose existence is given by the argument discussed at the beginning of the chapter. Observe that $\tau_{lr}$ is identically zero, since the lower-right corner is vanishing. At the same time:
\begin{align*}
\sigma((\iota \tilde{\kappa} u,\iota \tilde{\kappa} \tilde{u}),(\iota \tilde{\kappa} u',\iota \tilde{\kappa} \tilde{u}')) &= \left(\I{\iota\tilde{\kappa} \tilde{u}}\cdot \iota \tilde{\kappa} u' - \I{\iota \tilde{\kappa} u}\cdot \iota \tilde{\kappa} \tilde{u}'  \right)\mu\\
&= \left(\I{\iota\tilde{\kappa} \tilde{u}} \cdot \kappa\tilde{\iota} u' - \I{\iota \tilde{\kappa} u}\cdot \kappa \tilde{\iota} \tilde{u}'  \right)\mu =0.
\end{align*}
Besides:
\begin{align*}
&\sigma((\iota \times \iota) a  (dA,d\tilde{A}), (\iota\times\iota)a (dA',d\tilde{A}'))\\
= & (\iota a_2(d\tilde{A})\cdot \iota a_1 (dA') - \iota a_1 (dA) \cdot \iota a_2 (d\tilde{A}'))\mu\\
=& \iota \left( a_2 (d\tilde{A}) \wedge \curv \iota a_1 (dA') - a_1 (dA) \wedge \curv \iota  a_2 (d\tilde{A}') \right)\mu\\
=& \iota \left( a_2 (d\tilde{A}) \wedge d\, a_1 (dA') - a_1 (dA) \wedge d\,  a_2 (d\tilde{A}') \right)\mu\\
=& \iota \left( a_2 (d\tilde{A}) \wedge dA' - a_1 (dA) \wedge d\tilde{A'} \right)\mu\\
=& \int_\Sigma a_2 (d\tilde{A}) \wedge dA' - a_1 (dA) \wedge d\tilde{A'}\,\,\mod\Z\\
=& \int_\Sigma \tilde{A}	\wedge dA' - A\wedge d\tilde{A}\,\,\mod\Z= \tau_u((dA,d\tilde{A}),(dA',d\tilde{A}')),
\end{align*}
where in the last passage we performed integration by parts twice. Hence the splitting preserves the pre-symplectic structure. The orthogonality is a consequence of $\iota \tilde{\kappa}(\arg) \cdot \iota (\arg)=0$. In the end, what we obtain is:
\begin{equation}
(\hat{H}^{2,2}(\Sigma;\Z),\sigma)\simeq (\T^2,0)\oplus (d\Omega^{1,1}_c(\Sigma),\tau_u).
\end{equation}

The obstruction to a decomposition in general stems from the homomorphism $I$. In particular, in the previous chapter, we showed an example in which $I$ is neither surjective nor injective (see the end of Section \ref{subsec:2:rel_diff_coho}). We remind the reader that the pairing $\pa{\arg}{\arg}_c$ is weakly non degenerate and that the radical of $\sigma$ is given by:
$$
\mathrm{Rad}(\sigma)=\mathrm{ker}\left[ I: \hat{H}_c^{k,m-k}(\Sigma;\Z) \to \hat{H}^{k,m-k}(\Sigma;\Z) \right].
$$
Consider the following diagram:
\begin{equation}
\xymatrix{
0 \ar[r]&  H_c^{k-1,m-k-1}(\Sigma;\T) \ar[r] \ar[d]_-I & \hat{H}^{k,m-k}_c(\Sigma;\Z) \ar[d]^-I \ar[r] & \Omega^{k,m-k}_{c,\Z}(\Sigma) \ar[d]^-I \ar[r] &0\\
0 \ar[r]&  H^{k-1,m-k-1}(\Sigma;\T) \ar[r]  & \hat{H}^{k,m-k}(\Sigma;\Z)  \ar[r] & \Omega^{k,m-k}_{\Z}(\Sigma) \ar[r] &0}
\end{equation}
Being the right vertical map nothing but an inclusion, therefore injective, the diagram tells us that the central vertical map is injective only if the left one is such. Consequently, the kernel of $I:H_c^{k-1,m-k-1}(\Sigma;\T) \to H^{k-1,m-k-1}(\Sigma;\T)$ gives us information on how $I$ acts on compactly supported differential characters.

The effect of $I$ is in some cases to suppress all the information, mapping the whole compactly supported group to zero. This is the case of $M=\R^2\times \T^2$, with $k=2$. As a matter of fact, it holds:
\begin{enumerate}[(i)]
\item $H^{1}(\Sigma;\Z)=\Z^2$, $\qquad [d\theta_1],[d\theta_2]$;
\item $H^{2}(\Sigma;\Z)=\Z$, $\qquad [d\theta_1\wedge d\theta_2]$;
\item $H^1_c(\Sigma;\Z)=\Z$, $\qquad [f(x)dx]$;
\item $H^2_c(\Sigma;\Z)=\Z^2$, $\qquad [f(x) d\theta_2\wedge dx],[f(x) dx\wedge d\theta_1]$,
\end{enumerate}
where on the right we reported the generators of the groups. $I$ maps (iii) and (iv) to zero.

Another undesired consequence of $I$ is that, for manifolds with certain topological properties, it brings about, when computing the pre-symplectic product, the coupling between the groups in the diagram in a complicated way. Take, for example, $M=\R^3\times \S^1$. The diagram of the observables is:
\begin{equation}\label{ex:3:R3_S1_obs}
\xymatrix{
& 0 \ar[d] & 0 \ar[d] & 0 \ar[d] & \\
0 \ar[r] & 0 \ar[d] \ar[r] & \dfrac{\Omega^{1,1}_c(\Sigma)}{\Omega^{1,1}_{c,\Z}(\Sigma)}\ar[r]\ar[d]& d\Omega_c^{1,1}(\Sigma) \ar[d] \ar[r] &0\\
0 \ar[r] & 0 \ar[d] \ar[r] & \hat{H}_c^{2,2}(\Sigma;\Z) \ar[r]\ar[d] & \Omega_{c,\Z} ^{2,2}(\Sigma) \ar[d] \ar[r] &0\\
0 \ar[r] & 0  \ar[d] \ar[r]& \Z^2 \ar[d] \ar[r]&  \Z^2\ar[d] \ar[r] & 0\\
& 0 & 0 & 0 &
}
\end{equation}
On the other hand, the diagram of the configurations is given by:
\begin{equation}\label{ex:3:R3_S1_config}
\xymatrix{
& 0 \ar[d] & 0 \ar[d] & 0 \ar[d] & \\
0 \ar[r] & \T^2 \ar[d] \ar[r] & \dfrac{\Omega^{1,1}(\Sigma)}{\Omega^{1,1}_{\Z}(\Sigma)}\ar[r]\ar[d]& d\Omega^{1,1}(\Sigma) \ar[d] \ar[r] &0\\
0 \ar[r] & \T^2 \ar[d] \ar[r] & \hat{H}^{2,2}(\Sigma;\Z) \ar[r]\ar[d] & \Omega_{\Z} ^{2,2}(\Sigma) \ar[d] \ar[r] &0\\
0 \ar[r] & 0  \ar[d] \ar[r]& 0 \ar[d] \ar[r]& 0 \ar[d] \ar[r] & 0\\
& 0 & 0 & 0 &
}
\end{equation}
When computing the pre-symplectic product for the central group, we get terms of the form $h\cdot I \tilde{h}'$, for some $h, \tilde{h}'\in \hat{H}^2_c(\Sigma;\Z)$. Since the third row in the configuration diagram \eqref{ex:3:R3_S1_config} is vanishing, the central group is isomorphic to the upper central one; consequently, there exists a unique $\tilde{B}'\in \frac{\Omega^{1,1}(\Sigma)}{\Omega^{1,1}_{\Z}(\Sigma)}$ such that $\iota(\tilde{B}')=I\tilde{h}'$. This means that whatever the contribution to the differential character $h$ is, it must be coupled with an element in $\frac{\Omega^{1,1}(\Sigma)}{\Omega^{1,1}_{\Z}(\Sigma)}$. This fact, even though it does not possess the cogency of a no-go theorem, strongly suggests that an orthogonal decomposition can not be attained following slavishly the above procedure.
\chapter{Quantum Abelian Duality}\label{chapter:4}

In this chapter we will discuss the duality properties of our covariant quantum field theory. In the following, we assume $m\geq 2$ and $k\in\{1,\dots,m-1\}$.

\begin{defn}[Quantum duality]
Let $\mathfrak{A}^k,\mathfrak{A}^{\tilde{k}}:\mathsf{Loc}_m \to \mathsf{C^\ast Alg}$ be two quantum field theories. We call \emph{duality} between $\mathfrak{A}^k$ and $\mathfrak{A}^{\tilde{k}}$ a natural isomorphism $\eta:\mathfrak{A}^k\Rightarrow \mathfrak{A}^{\tilde{k}}$.
\end{defn}

To begin with, following \cite{BPhys}, we construct explicitly a duality for our quantum field theory adopting a bottom-up approach, moving from the configuration space. Having in mind an extension of the duality between the electric and the magnetic Maxwell fields, we introduce the homomorphism:
\begin{eqnarray}\label{eq:4:xi}
&\zeta: \C^k(M;\Z) \to \C^{m-k}(M;\Z) &\nonumber\\
&(h,\tilde{h})\mapsto (\tilde{h},-(-1)^{k(m-k)}h).&
\end{eqnarray}
The homomorphism $\zeta$ is an isomorphism, because its components are such, and it is, furthermore, natural, since the diagram:
\begin{equation}
\xymatrix{
\C^{m-k}(M';\Z) \ar[r]^-{\zeta}\ar[d]_-{f^\ast} & \C^k(M';\Z)\ar[d]^-{f^\ast}\\
\C^{m-k}(M;\Z) \ar[r]_-{\zeta} &\C^k(M;\Z)
}
\end{equation}
commutes for every objects $M,M^\prime$ and for every morphism $f:M\to M'$ in $\mathsf{Loc}_m$. Let $\zeta^\star:\C^k_{sc}(M;\Z)\to\C^{m-k}_{sc}(M;\Z)$ be the homomorphism dual to $\zeta$ with respect to \eqref{pairing:2:config_gauge_fields}, defined by:
\begin{equation*}
\pa{(h,\tilde{h})}{\zeta^\star(h',\tilde{h}')}:= \pa{\zeta(h,\tilde{h})}{(h',\tilde{h}')}
\end{equation*}
for every $(h,\tilde{h})\in \C^{m-k}(M;\Z)$, $(h',\tilde{h}')\in \C_{sc}^{k}(M;\Z)$. A straightforward computation reveals that:
$$
\zeta^\star(h',\tilde{h}')=(-(-1)^{k(m-k)}\tilde{h}',h').
$$
Hence, $\zeta^\star$ defines an isomorphism, which is natural on account of the commutativity of the diagram:
\begin{equation}
\xymatrix{
\C^{k}_{sc}(M;\Z) \ar[r]^-{\zeta^\star}\ar[d]_-{f_\ast} & \C^{m-k}_{sc}(M;\Z)\ar[d]^-{f_\ast}\\
\C^{k}_{sc}(M';\Z) \ar[r]_-{\zeta^\star} &\C^{m-k}_{sc}(M';\Z).
}
\end{equation}
A noteworthy feature of $\zeta^\star$ is that it preserves the pre-symplectic structure. In fact, for every $(h,\tilde{h}),(h',\tilde{h}')\in \C^k_{sc}(M;\Z)$, it holds:
\begin{align*}
\sigma(\zeta^\star(h,\tilde{h}),\zeta^\star(h',\tilde{h}'))=& \sigma((-(-1)^{k(m-k)}\tilde{h},h),(-(-1)^{k(m-k)}\tilde{h}',h'))\\
=& -(-1)^{k(m-k)} \pa{\ii{h}}{\i{\tilde{h}'}}_c + [-(-1)^{k(m-k)}]^2 \pa{\ii{\tilde{h}}}{\i{h'}}_c\\
=& \sigma((h,\tilde{h}),(h',\tilde{h}')).
\end{align*}
Exploiting the isomorphism \eqref{eq:2:iso_D_C_sc}, $\mathcal{O}: \C^k_{sc}(\arg;\Z) \Rightarrow \D^k(\arg;\Z)$, we obtain a natural isomorphism:
$$
\zeta^\star_\D: (\D^k(\arg;\Z),\tau) \Rightarrow (\D^{m-k}(\arg;\Z);\tau)
$$
between functors from $\mathsf{Loc}_m$ to $\mathsf{Ab}$.

\begin{prop}
The homomorphism:
\begin{equation}
\eta:=\mathfrak{CCR}(\zeta^\star_\D):\mathfrak{A}^k\to \mathfrak{A}^{m-k}
\end{equation}
establishes a duality between the quantum field theories $\mathfrak{A}^k,\mathfrak{A}^{m-k}:\mathsf{Loc}_m \to \mathsf{C^\ast Alg}$.
\end{prop}
\begin{proof}
$\eta$ is an isomorphism, because $\mathfrak{CCR}$ is a functor, and functors preserve the isomorphisms. As $\mathfrak{CCR}$ preserves compositions, the naturality of $\eta$ follows at once from that of $\zeta_\D^\star$.
\end{proof}

\section{Duality and splittings}\label{section:4.1_duality _and_splittings}
Let $M$ be an object in $\mathsf{Loc}_{2k}$ with compact Cauchy surface $\Sigma$. The first question which naturally arises is whether it is possible to choose splittings $a$, $b$ and $x$ as in Proposition \ref{prop:3:existence_of_splitting} that are, in addition, compatible with the duality map. The answer will be, as we shall show, positive. Let us first translate in mathematical terms such compatibility condition.

Mimicking the above construction, we can introduce duality maps on the other Abelian groups in diagram \eqref{diag:3:splitting}. As far as $d\Omega^{k-1,k-1}(\Sigma)$ is concerned, define:
\begin{eqnarray}\label{eq:4:duality_u}
&\zeta_u: d\Omega^{k-1,k-1}(\Sigma) \to d\Omega^{k-1,k-1}(\Sigma)&\nonumber\\
&(dA,d\tilde{A})\mapsto(-(-1)^{k^2}d\tilde{A}, dA).&
\end{eqnarray}
We say that the splitting $a$ is compatible with the duality if and only if:
\begin{equation*}
[\zeta^\star\circ(\iota\times \iota)\circ a](dA,d\tilde{A})= [(\iota\times \iota)\circ a\circ \zeta_u](dA,d\tilde{A}).
\end{equation*}
Recalling that a splitting of a product homomorphism between product spaces is always a product, we write $a=a_1\times a_2$. The left-hand side is:
\begin{align*}
[\zeta^\star\circ(\iota\times \iota)\circ a](dA,d\tilde{A})=(-(-1)^{k^2}\iota a_2 d\tilde{A}, \iota a_1 dA),
\end{align*}
while the right-hand side is:
\begin{align*}
[(\iota\times \iota)\circ  a\circ \zeta_u](dA,d\tilde{A})=(-(-1)^{k^2}\iota a_1 d\tilde{A}, \iota a_2 dA).
\end{align*}
The compatibility condition is given by $a_1=a_2$. Analogously, defining the homomorphisms:
\begin{eqnarray*}
&\zeta_{lr}:\dfrac{H^{k-1,k-1}(\Sigma;\R)}{H^{k-1,k-1}_{\free}(\Sigma;\Z)} \times H^{k,k}_\free(\Sigma;\Z) \to \dfrac{H^{k-1,k-1}(\Sigma;\R)}{H^{k-1,k-1}_{\free}(\Sigma;\Z)} \times H^{k,k}_\free(\Sigma;\Z)  &\nonumber\\
&((u, \tilde{u}),(v,\tilde{v}))\mapsto ((-(-1)^{k^2}\tilde{u},u),(-(-1)^{k^2}\tilde{v},v))&
\end{eqnarray*}
and
\begin{eqnarray*}
&\zeta_d: H^{k,k}_\tor(\Sigma;\Z) \to H^{k,k}_\tor(\Sigma;\Z)  &\nonumber\\
& (t,\tilde{t})\mapsto (-(-1)^{k^2}\tilde{t},t),&
\end{eqnarray*}
the compatibility conditions:
\begin{align*}
&\zeta^\star[(x_1\times x_2)(v,\tilde{v})]=(x_1\times x_2)(-(-1)^{k^2}\tilde{v},v),\\
& \zeta^\star[(\kappa\times \kappa)(b_1\times b_2)(t,\tilde{t})]=(\kappa\times \kappa)(b_1\times b_2)\zeta_d (t,\tilde{t}),
\end{align*}
yield the constraints $x_1=x_2$ and $b_1=b_2$.

\begin{thm}
Let $M$ be an object in $\mathsf{Loc}_m$ with compact Cauchy surface $\Sigma$ and $m=2k$. With reference to diagram \eqref{diag:3:splitting}, there exist splitting homomorphisms:
\begin{subequations}\label{eq:4:splitting}
\begin{align}
& a\times a: d\Omega^{k-1,k-1}(\Sigma) \to \dfrac{\Omega^{k-1,k-1}(\Sigma)}{\Omega_{\mathbb{Z}}^{k-1,k-1}(\Sigma)}, \label{eq:4:splitting_a}\\
& b\times b : H^{k,k}_{\tor}(\Sigma;\mathbb{Z}) \to H^{k-1,k-1}(\Sigma;\mathbb{T}), \label{eq:4:splitting_b}\\
& x\times x : H^{k,k}_{\free}(\Sigma;\mathbb{Z}) \to \hat{H} ^{k,k}(\Sigma;\mathbb{Z}), \label{eq:4:splitting_x}
\end{align}
\end{subequations}
satisfying the following conditions:
\begin{subequations}\label{eq:4:splitting_condition}
\begin{align}
& \sigma((x\, v,x\,\tilde{v}),(x\, v',x\, \tilde{v}'))=0, \label{eq:4:splitting_condition_xx}\\
& \sigma((\iota\times \iota) (a\, dA,a\, d\tilde{A}),(x\,v,x\,\tilde{v}))=0, \label{eq:4:splitting_condition_ax}\\
& \sigma((\kappa\times \kappa) (b\, t,b\, \tilde{t}), (x\, v,x\, \tilde{v}))=0 \label{eq:4:splitting_condition_bx},
\end{align}
\end{subequations}
for all $(v,\tilde{v}),(v',\tilde{v}')\in H^{k,k}_{\free}(\Sigma;\mathbb{Z})$, $(dA,d\tilde{A})\in d\Omega^{k-1,k-1}(\Sigma)$ and $(t,\tilde{t})\in H^{k,k}_{\tor}(\Sigma;\mathbb{Z})$.
\end{thm}
\begin{proof}
Since $H^k_\free(\Sigma;\Z)$ is a free Abelian group there exists $n\in\N$ such that $H^k_\free(\Sigma;\Z)\simeq \Z^n$. Choose a set of generators $\{z_i, i\in\{1,\dots,n\}\}$ for $\Z^n$. With abuse of notation, $z_i$ will denote also the  corresponding generators in $H^k_\free(\Sigma;\Z)$ via the isomorphism. As $H^k_\free(\Sigma;\Z)$ is a lattice in $H^k(\Sigma;\R)$, a basis for the former is also a basis for the latter. Pick a basis $\{r_i, i\in\{1,\dots,n\}\}$ in $H^{k-1}_\free(\Sigma;\Z)$ such that:
$$
r_i	\smile z_j=\delta_{ij}, \qquad 	\forall i,j\in\{1,\dots,n\}.
$$

For every $i$ in $\{1,\dots,n\}$, choose $h_i'\in \hat{H}^k(\Sigma;\Z)$ such that $[\curv h_i']=z_i$. In order to fulfil the condition \eqref{eq:4:splitting_condition_xx}, we look for $u_i\in \frac{H^{k-1}(\Sigma;\R)}{H^{k-1}_\free(\Sigma;\Z)}$ such that, denoting by $h_i$ the modified differential character:
\begin{equation*}
h_i:=h_i'+\iota \tilde{\kappa} u_i,
\end{equation*}
the equation $h_i\cdot h_i=0$ is satisfied. 

Introduce an array of real numbers $(c_{ij})\in \mathbb{M}(n;\R)$ defined as:
\begin{equation}
c_{ij}:\begin{cases}
\text{if }i< j\qquad & c_{ij}\,\,\mod \Z= \pa{h_i'}{h_j'}_c\\
\text{if }i>j\qquad & c_{ij}=(-1)^k c_{ji}.
\end{cases}
\end{equation}
Furthermore, set $c_{ii}=0$ for $k$ odd and $c_{ii}\,\,\mod\Z=\pa{h_i'}{h_i'}_c$ for $k$ even.
If we impose:
\begin{equation}
\pa{h_i}{h_j}_c=\pa{h_i'}{h_j'}_c + \pa{\iota \tilde{\kappa}u_i}{h'_j}_c + \pa{h'_i}{\iota \tilde{\kappa} u_j}_c\overset{\downarrow}{=}0
\end{equation}
for all $i,j\in\{1,\dots,n\}$, we get a set of equations of the form:
\begin{equation*}
c_{ij}=-\pa{u_i}{z_j}_f + (-1)^{k^2+1}\pa{u_j}{z_i}_f.
\end{equation*}
By writing each $u_i$ in terms of the elements of the basis as:
$$
u_i=\sum_{k=1}^n u_k^{(i)}r_k\quad \mod H^{k-1}_\mathrm{free}(\Sigma;\mathbb{Z}),
$$
we obtain eventually $n^2$ equations for the $u_i$ components:
\begin{equation}
c_{ij}=(-1)^{k^2+1}u_i^{(j)}-u_j^{(i)},\qquad i,j\in\{1,\dots,n\}
\end{equation}
They admit as a solution: 
\begin{equation}
u^{(i)}_j = -\frac{1}{2} c_{ij}.
\end{equation}
Hence, elements $h_i$ with the required properties exist and are well-defined. The homomorphism $x$ can then be defined as:
\begin{eqnarray}
&x: H^{k,k}_\free(\Sigma;\Z) \to \hat{H}^k(\Sigma;\Z) &\\
&z_i\mapsto h_i.&
\end{eqnarray}
Observe that the map $z_i\mapsto h_i'$ is a splitting by construction; the additional term involved in passing from $h_i'$ to $h_i$ does not affect the properties of the map in this respect.

As far as $a$ and $b$ are concerned, the proof retraces step by step the one for Proposition \ref{prop:3:existence_of_splitting}, with $a\times a$ and $b\times b$ instead of $a$ and $b$ respectively.
\end{proof}

Once we have the splittings, Theorem \ref{thm:3:existence_decomposition_from_splittings} gives a pre-symplectically orthogonal decomposition of $\hat{H}^{k,k}(\Sigma;\Z)$ which, additionally, preserves the duality.

\section{Duality and GNS representation}\label{section:4.2_duality_and_GNS}

Let $M$ be an object in $\mathsf{Loc}_{2k}$ with compact Cauchy surface. Recall that $\mathcal{A}_{lr}=\mathfrak{CCR}\left(\frac{H^{k-1,k-1}(M;\R)}{H^{k-1,k-1}_\free (M;\Z)}\times H^{k,k}_\free(M;\Z) ,\tau_{lr}\right)$ and consider the state for the topological sector $\omega_t:\mathcal{A}_{lr}\to \mathbb{C}$ (see Proposition \ref{prop:3:omega_t}). 
Let us construct the GNS triple $(\mathcal{H}_{\omega_t},\pi_{\omega_t},\Psi_{\omega_t})$.

Define:
\begin{eqnarray*}
&\pa{\arg}{\arg}_{\omega_t}: \mathcal{A}_{lr}\times \mathcal{A}_{lr} \to \mathbb{C} &\nonumber\\
&(\mathcal{W}(u,\tilde{u},v,\tilde{v}),\mathcal{W}(u',\tilde{u}',v',\tilde{v}'))\mapsto \omega_t \left(\mathcal{W}(u,\tilde{u},v,\tilde{v})^\ast  \mathcal{W}(u',\tilde{u}',v',\tilde{v}') \right). &
\end{eqnarray*}
The pairing $\pa{\arg}{\arg}_{\omega_t}$ enjoys the property of being a Hermitian semi-inner product. The Gelfand ideal $\mathcal{J}_{\omega_t}$ for the $C^\ast$-algebra $\mathcal{A}_{lr}$ is generated (as left ideal) by elements of the form: 
$$
\mathcal{W}(0,0,v,\tilde{v})-e^{-2\pi i(\pa{\tilde{u}}{v}_f -(-1)^{k(m-k)}\pa{u}{\tilde{v}}_f) } \mathcal{W}(u,\tilde{u},v,\tilde{v})
$$ 
for some $(u,\tilde{u},v,\tilde{v})\in \frac{H^{k-1,k-1}(M;\R)}{H^{k-1,k-1}_\free (M;\Z)}\times H^{k,k}_\free(M;\Z)$ (cfr. equation \eqref{eq:3:elements_gelfand_ideal}). We introduce the quotient:
$$
\mathcal{D}_{\omega_t}:=\mathcal{A}/\mathcal{J}_{\omega_t},
$$
and we adopt the notation:
$$
\ket{v,\tilde{v}}:=\left[\mathcal{W}(0,0,v,\tilde{v}) \right],
$$
where the square brackets denote the coset in the quotient.

The $\mathbb{C}$-linear space $\mathcal{D}_{\omega_t}$ can be endowed with the Hermitian inner product:
$$
\braket{\arg}{\arg}: \mathcal{D}_{\omega_t}\times \mathcal{D}_{\omega_t}\to \mathbb{C}
$$
$$
(\ket{v,	\tilde{v}}, \ket{v',\tilde{v}'})\mapsto \braket{v,\tilde{v}}{v',\tilde{v}'}=\pa{\mathcal{W}(0,0,v,\tilde{v})}{\mathcal{W}(0,0,v',\tilde{v}')}_{\omega_t}.
$$
The GNS Hilbert space is then given by the completion of $\mathcal{D}_{\omega_t}$ with respect to the norm induced by such a product:
$$
\mathcal{H}_{\omega_t}:=\overline{\mathcal{D}_{\omega_t}}^{\braket{\arg}{\arg}}.
$$
Define the representation $\pi_{\omega_t}$ as:
\begin{eqnarray*}
& \pi_{\omega_t}:\mathcal{A}\to \mathcal{BL}(\mathcal{H}_{\omega_t}) &\\
& \mathcal{W}(u,\tilde{u},v,\tilde{v})\mapsto \left( \ket{v',	\tilde{v'}}\mapsto \left[\mathcal{W}(u,\tilde{u},v,\tilde{v})\mathcal{W}(0,0,v',\tilde{v}')\right] \right). &
\end{eqnarray*}
The cyclic vector is represented by $\Psi_{\omega_t}:=\ket{0,0}$; as a matter of fact, by acting with the representation on $\ket{0,0}$ we can reconstruct all the dense subspace $\mathcal{D}_{\omega_t}$, which is in turn invariant under the action of the representation.\\

At the level of the algebra $\mathcal{A}$ the duality is implemented by the homomorphism:
\begin{equation}
\zeta_\mathcal{A}: \mathcal{A}\to \mathcal{A}, \qquad \mathcal{W}(u,\tilde{u},v,\tilde{v})\mapsto \mathcal{W}(-(-1)^{k^2}\tilde{u},u,-(-1)^{k^2}\tilde{v},v).
\end{equation}
The following result characterises the effect of $\zeta_\mathcal{A}$ on the GNS Hilbert space.

\begin{thm}
The quantum Abelian duality is implemented on $\mathcal{H}_{\omega_t}$ by the unitary operator given by the continuous linear extension of:
\begin{eqnarray}
&U:\mathcal{D}_{\omega_t}\to \mathcal{H}_{\omega_t}&\nonumber\\
&\ket{v,\tilde{v}}\mapsto \ket{-(-1)^{k^2}\tilde{v},v}.&
\end{eqnarray}
\end{thm}
\begin{proof}
To begin with, observe that $U$ actually implements the duality:
$$
\left[ \zeta_\mathcal{A} \mathcal{W}(0,0,v,\tilde{v}) \right]=\left[ \mathcal{W}(0,0,-(-1)^{k^2}\tilde{v},v) \right]=\ket{-(-1)^{k^2}\tilde{v},v}.
$$
The operator $U$ is linear and densely defined by construction. Its range coincides with $\mathcal{D}_{\omega_t}$ and is, therefore, dense in $\mathcal{H}_{\omega_t}$. Furthermore, it preserves the inner product; in fact, for every $\ket{v,\tilde{v}},\ket{v',\tilde{v'}}\in \mathcal{D}_{\omega_t}$, it holds:
\begin{equation*}
\begin{split}
&\braket{-(-1)^{k^2}\tilde{v},v}{-(-1)^{k^2}\tilde{v}',v'}=\omega_{t}( \mathcal{W}(0,0,-(-1)^{k^2}(\tilde{v}'-\tilde{v}),v'-v))=\\
&=\omega_{t}\left( \mathcal{W}(0,0,v'-v, \tilde{v}'-\tilde{v})\right)=\braket{v,\tilde{v}}{v',\tilde{v}'}.
\end{split}
\end{equation*}
In the central passage we used that $(-(-1)^{k^2}(\tilde{v}'-\tilde{v}),v'-v)=(0,0)$ if and only if $(v'-v, \tilde{v}'-\tilde{v})=(0,0)$. Hence $U$ is bounded, because it is, in particular, an isometry. The continuous linear extension theorem yields a unitary operator defined on the whole Hilbert space $\mathcal{H}_{\omega_t}$.
\end{proof}

\begin{rem}\label{rem:4:duality_topological}
If we consider a generic element $\mathcal{W}(u,\tilde{u},v,\tilde{v})$, the associated vector in $\mathcal{H}_{\omega_t}$ is given by:
\begin{equation*}
[\mathcal{W}(u,\tilde{u},v,\tilde{v})]= e^{2\pi i(\pa{\tilde{u}}{v}_f-(-1)^{k^2}\pa{u}{\tilde{v}}_f)}\ket{v,\tilde{v}}.
\end{equation*}
Observe that the exponential is invariant under duality; namely:
\begin{equation*}
[\mathcal{W}(-(-1)^{k^2}\tilde{u},u,-(-1)^{k^2}\tilde{v},v)]= e^{2\pi i(-(-1)^{k^2}\pa{u}{\tilde{v}}_f+\pa{\tilde{u}}{v}_f)}\ket{-(-1)^{k^2}\tilde{v},v}.
\end{equation*}
Hence, $U$ actually unambiguously implements duality on $\mathcal{H}_{\omega_t}$, independently of the choice of the representative of $\ket{v,\tilde{v}}$.
\end{rem}

The representation comprises elements implementing \emph{rotation} and \emph{translation} operators on the Hilbert space $\mathcal{H}_{\omega_t}$ in a sense made precise by the definitions below.

\begin{defn}
We call \emph{rotation operators} the bounded operators given by the continuous linear extension of:
\begin{eqnarray}\label{eq:4:rotation_operators_1}
&\mathcal{R}(u'):=\pi_{\omega_t}(\mathcal{W}(u',0,0,0)): \mathcal{D}_{\omega_t}\to \mathcal{H}_{\omega_t}&\nonumber\\
&\ket{v,\tilde{v}}\mapsto e^{-(-1)^{k^2} 2\pi i \pa{u'}{\tilde{v}}_f}\ket{v,\tilde{v}},&
\end{eqnarray}
\begin{eqnarray}\label{eq:4:rotation_operators_2}
&\tilde{\mathcal{R}}(\tilde{u'}):=\pi_{\omega_t}(\mathcal{W}(0,\tilde{u'},0,0)):\mathcal{D}_{\omega_t}\to \mathcal{H}_{\omega_t}&\nonumber\\
&\ket{v,\tilde{v}}\mapsto e^{2\pi i \pa{\tilde{u}'}{v}_f}\ket{v,\tilde{v}}. &
\end{eqnarray}
\end{defn}
The effect of the rotation operators is that of multiplying vectors by a phase which depends on the vector itself. Observe that these unitary operators are parametrized by $\frac{H^{k-1}(\Sigma;\R)}{H^{k-1}_\free(\Sigma;\Z)}$. If we pick a set $\mathcal{G}=\{ z_i \}_{i=1}^n$ of independent generators of $\frac{H^{k-1}(\Sigma;\R)}{H^{k-1}_\free(\Sigma;\Z)}$, \eqref{eq:4:rotation_operators_1} and \eqref{eq:4:rotation_operators_2} indentify a one-parameter strongly continuous unitary group for each element of $\mathcal{G}$, where the parameter is a phase. Stone's theorem guarantees that each of such families is generated by a densely defined, essentially self-adjoint operator on $\mathcal{H}_{\omega_t}$.

\begin{defn}
We call \emph{translation operators} the bounded operators given by the continuous linear extension of:
\begin{eqnarray}\label{eq:4:translation_1}
&\mathcal{T}(v'):=\pi_{\omega_t}(\mathcal{W}(0,0,v',0)):\mathcal{D}_{\omega_t}\to \mathcal{H}_{\omega_t}&\nonumber \\
&\ket{v,	\tilde{v}}\mapsto \ket{v+v',\tilde{v}},&
\end{eqnarray}
\begin{eqnarray}\label{eq:4:translation_2}
&\mathcal{T}(\tilde{v}'):=\pi_{\omega_t}(\mathcal{W}(0,0,0,\tilde{v}')):\mathcal{D}_{\omega_t}\to \mathcal{H}_{\omega_t}&\nonumber \\
&\ket{v,	\tilde{v}}\mapsto \ket{v,\tilde{v}+\tilde{v}'}.&
\end{eqnarray}
\end{defn}

The effect of a translation operator is that of translating the first or the second basis vectors' label of an amount depending on the argument of the operator itself.

In conclusion, let us introduce other operators that offer the opportunity of an effective interpretation of our model in particular cases. Notice that, being a finitely generated free Abelian group, $H^k_{\free}(M;\Z)$ is isomorphic to $\Z^n$ for some $n\in\N$. Choose a basis $\{ s_i \}_{i=1}^n$ of $H^k_\free(\Sigma;\Z)$; every $v\in H^k_\free(\Sigma;\Z)$ can be written as $v=\sum_{i=1}^n v^i s_i, \,v^i\in\Z$. Consider the densely defined unbounded linear operators:
\begin{eqnarray}\label{eq:4:momentum_1}
&\Pi_i: \mathcal{D}_{\omega_t}\to \mathcal{H}_{\omega_t}&\nonumber\\
&\ket{v,\tilde{v}} \mapsto v^i \ket{v,\tilde{v}}&
\end{eqnarray}
and
\begin{eqnarray}\label{eq:4:momentum_2}
&\tilde{\Pi}_j: \mathcal{D}_{\omega_t}\to \mathcal{H}_{\omega_t}&\nonumber\\
&\ket{v,\tilde{v}} \mapsto \tilde{v}^j\ket{v,\tilde{v}}.&
\end{eqnarray}
The maps $\Pi_i$ and $\tilde{\Pi}_j$, $i,j\in\{1,\dots,n\}$, are well-defined, since they consist in the multiplication by an integer number. It is a remarkable fact that, if the generating sets $\mathcal{G}$ and $\{s_i\}_{i=1}^n$ are chosen accordingly, the operators $\Pi_i$ and $\tilde{\Pi}_j$ are exactly the essentially self-adjoint generators of the rotation operators.


\begin{rem}
To stress once more the role of the state defined in Proposition \ref{prop:3:omega_t}, let us briefly recall the peculiar features of the representation it induces. With reference to \eqref{eq:4:rotation_operators_1}, \eqref{eq:4:rotation_operators_2}, \eqref{eq:4:momentum_1} and \eqref{eq:4:momentum_2}, this representation allows us to interpret the topological degrees of freedom as corresponding to two (families of) point particles quantized on a circle, whose momenta are described by essentially self-adjoint operators (with discrete spectrum) that are interchanged by the unitary operator implementing duality at the level of the chosen representation. In fact, it is remarkable that the two different families of rotation operators are intertwined by the duality operator. As a matter of fact:
\begin{align*}
U\mathcal{R}(u')\ket{v,\tilde{v}}&= e^{-(-1)^{k^2}2\pi i \pa{u'}{\tilde{v}}_f}\ket{-(-1)^{k^2}\tilde{v},v}\\
&=e^{2\pi i \pa{u'}{-(-1)^{k^2}\tilde{v}}_f}\ket{-(-1)^{k^2}\tilde{v},v}=\tilde{R}(u') U\ket{v,\tilde{v}}.
\end{align*} 
The same holds true for the momentum operators:
$$
U\Pi_i \ket{v,\tilde{v}}=\tilde{\Pi}_i U\ket{v,\tilde{v}}.
$$
\end{rem}
Therefore, the duality proves to be consistent and coherent with the entire construction.\\

%

Let us come, in the end, to discuss the role of duality for the dynamical sector of the theory. From the point of view of the whole spacetime $M$, the map $\zeta_u$, originally defined on the Cauchy surface $\Sigma$ by \eqref{eq:4:duality_u}, becomes: 
\begin{eqnarray}\label{eq:4:duality_u_M}
&\zeta_u: d\Omega^{k-1}\cap \ast d\Omega^{k-1}(M)\to d\Omega^{k-1}\cap \ast d\Omega^{k-1}(M) &\nonumber\\
&dA=\ast d\tilde{A}\mapsto (-1)^{k^2+1}d\tilde{A}=\ast dA.&
\end{eqnarray}
Introducing the notation $\mathcal{A}_u:=\mathfrak{CCR}(d\Omega^{k-1}\cap \ast d\Omega^{k-1}(M), \tau_u)$, the analogue of $\zeta_u$ for the algebra of the observable is given by:
\begin{eqnarray}\label{eq:4:duality_u_M_algebra}
&\zeta_{\mathcal{A}_u}: \mathcal{A}_u \to \mathcal{A}_u &\nonumber\\
& \mathcal{W}(dA=\ast d\tilde{A})\mapsto \mathcal{W}((-1)^{k^2+1}d\tilde{A}=\ast dA). &
\end{eqnarray}

Let us fix the spacetime to be $M=\R\times \S^1$. The first question we ask ourselves is whether the state $\omega_\mu$ defined in \eqref{def:3:omega_mu} is invariant under the action of the spacetime symmetry group $\mathfrak{S}=\R\times \T$. In general, to each symmetry, i.e.\ an isometric diffeomorphism of the spacetime, we can associate, by virtue of the quantum field theory functor \eqref{eq:2:LCQFT}, an automorphism of the Weyl $C^\ast$-algebra, which implements the symmetry at the level of the algebra of observables. For the case in hand, the group of the symmetry automorphisms can be built as follows. For every $(t,\phi)\in \R\times \T$, define the one-parameter family of maps:
\begin{eqnarray*}
&\beta_{(s,\phi)}: dC^\infty \cap \ast dC^\infty(M) \to dC^\infty \cap \ast dC^\infty(M)&\\
&(df=\ast d\tilde{f})\mapsto (df(\arg+s,\arg+\phi)=\ast d\tilde{f}(\arg+s,\arg+\phi)).&
\end{eqnarray*}
The counterpart at the level of the algebra is fully specified by:
\begin{eqnarray*}
&\alpha_{(s,\phi)}:\mathcal{A}_u\to \mathcal{A}_u&\\
&\mathcal{W}(df=\ast d\tilde{f})\mapsto \mathcal{W}(\beta_{(s,\phi)}(df=\ast d\tilde{f})). &
\end{eqnarray*}

\begin{prop}\label{prop:4:state_invariant_symm}
The state $\omega_\mu$ is invariant under the action of the spacetime symmetries, i.e.:
$$
\omega_\mu \circ \alpha_{(s,\phi)}=\omega_\mu,
$$
for all $(s,\phi) \in \R	\times \T$.
\end{prop}
\begin{proof}
The expression for $\omega_\mu(\mathcal{W}(df=\ast d\tilde{f}))$ is given by equation \eqref{eq:3:state_2D}, where in the sum at the exponent we have the square modulus of the coefficients of the Fourier decomposition of $df$ on $\R\times\S^1$. It is enough to prove the statement for $\alpha_{(0,\phi)}$ and $\alpha_{(s,0)}$ separately: the general assert then follows from the one-parameter group structure. On the one hand, the effect of $\alpha_{(0,\phi)}$ is the multiplication of the Fourier  coefficients by $e^{\pm 2\pi i k \phi}$: this has no consequences when computing the square modulus. On the other hand, the map $\alpha_{(s,0)}$ leads to the appearance of a prefactor $e^{-2\pi i k s}$ before the positive frequency terms and of a prefactor $e^{2\pi i k s}$ before the negative frequency terms. When computing \eqref{eq:3:mu_R_S1}, the two complex exponentials with opposite sign compensate and they give no contribution to the final result. The state, therefore, turns out to be invariant under the action of $\alpha_{(s,\phi)}$ for every $(s,\phi)\in \R\times \T$. 
\end{proof}

A similar result holds true also for the duality. 

\begin{prop}\label{prop:4:state_invariant_duality}
The state $\omega_\mu$ is invariant under the duality map $\zeta_{\mathcal{A}_u}$, i.e.:
$$
\omega_\mu \circ \zeta_{\mathcal{A}_u}=\omega_\mu.
$$
\end{prop}
\begin{proof}
When computing $\omega_\mu(\zeta_{\mathcal{A}_u}\mathcal{W}(df=\ast d\tilde{f}))$, the map $\zeta_{\mathcal{A}_u}$ swaps the coefficients of the Fourier decomposition, up to a grading. As the pre-factor $-(-1)^{k^2}$ has no relevance to the extent that we take the square modulus of the coefficients, the expression for the state remains unchanged and the assert follows.
\end{proof}

At the level of the algebra, the duality and the automorphisms implementing the symmetries of the spacetime commute, that is to say $\zeta_{\mathcal{A}_u}\circ \alpha_{(s,\phi)}=\alpha_{(s,\phi)}\circ \zeta_{\mathcal{A}_u}$ for every $(s,\phi)\in \R\times \T$. In fact:
{\small
\begin{align*}
\left[\zeta_{\mathcal{A}_u}\circ \alpha_{(s,\phi)}\right](\mathcal{W}(df=\ast d\tilde{f}))=&\zeta_{\mathcal{A}_u}\mathcal{W}(df(\arg+s,\arg+\phi)=\ast d\tilde{f}(\arg+s,\arg+\phi))\\
=&\mathcal{W}((-1)^{k^2+1} d\tilde{f}(\arg+s,\arg+\phi)=\ast df(\arg+s,\arg+\phi))\\
=& \alpha_{(s,\phi)}\mathcal{W}((-1)^{k^2+1}d\tilde{f}=\ast df)\\
=&\left[\alpha_{(s,\phi)}\circ \zeta_{\mathcal{A}_u}\right](\mathcal{W}(df=\ast d\tilde{f})).
\end{align*}}
This peculiar property proves to have interesting consequences:

\begin{thm}\label{thm:4:symm_duality_unitary_commuting}
Let $(\mathcal{H}_{\omega_\mu},\pi_{\omega_\mu}, \Psi_{\omega_\mu} )$ be the GNS triple of $\omega_\mu$. Then, the duality and the spacetime symmetries are implemented on $\mathcal{H}_{\omega_\mu}$ by unitary operators. Furthermore, these operators commute.
\end{thm}
\begin{proof}
Given an element $a\in\mathcal{A}_u$, let $[a]\in \mathcal{H}_{\omega_\mu}$ be the associated vector in the Hilbert space via the GNS construction. Let $\mathcal{D}_{\omega_\mu}$ be the dense subspace of $\mathcal{H}_{\omega_\mu}$ whose elements are finite linear combinations $\sum_i c_i [a_i]$. We define the duality operator as the continuous linear extension of:
\begin{eqnarray}
&V:\mathcal{D}_{\omega_\mu}\to \mathcal{H}_{\omega_\mu} &\nonumber\\
& \sum_i c_i [a_i] \mapsto \sum_i c_i [\zeta_{\mathcal{A}_{u}}a]. &
\end{eqnarray}
$V$ is well-defined, because $\zeta_{\mathcal{A}_u}$ is an automorphism of $\mathcal{A}_u$. The range of $V$ is dense in $\mathcal{H}_{\omega_\mu}$; moreover, $V$ preserves the inner product:
\begin{align*}
\left\langle \sum_i c_i [\zeta_{\mathcal{A}_{u}}a_i]  \,\vline\, \sum_j d_j [\zeta_{\mathcal{A}_{u}} b_j]  \right\rangle_{\mathcal{H}_{\omega_\mu}}&=\sum_i \sum_j  \overline{c_i} d_j \omega_\mu (\zeta_{\mathcal{A}_{u}}(a_i^\ast b_j))\\
&= \sum_i \sum_j  \overline{c_i} d_j \omega_\mu (a_i^\ast b_j)\\
&= \left\langle \sum_i c_i [a_i]  \,\vline\, \sum_j d_j [b_j]  \right\rangle_{\mathcal{H}_{\omega_\mu}},
\end{align*}
where in the second step we used that the state is invariant under duality. Hence $V$ is unitary. Analogously, define the one-parameter group of operators $\{U_{(s,\phi)},\,(s,\phi)\in \R\times\T \}$ implementing the symmetries of the spacetime as the continuous linear extension of:
\begin{eqnarray}
& U_{(s,\phi)}: \mathcal{D}_{\omega_\mu}\to \mathcal{H}_{\omega_\mu}&\nonumber\\
& \sum_i c_i [a_i] \mapsto \sum_i c_i [\alpha_{(s,\phi)} a]. &
\end{eqnarray}
From the same reasoning, resorting to the invariance of the state under the symmetries of the spacetime, it follows that $U_{(s,\phi)}$ is unitary for every $(s,\phi)\in \R\times \T$. 

As far as the commutativity is concerned, we have:
\begin{align*}
U_{(s,\phi)} V \left(\sum_i c_i [a_i]\right) &= \sum_i c_i [\alpha_{(s,\phi)} (\zeta_{\mathcal{A}_{u}} a_i)]\\
&=\sum_i c_i [\zeta_{\mathcal{A}_{u}}(\alpha_{(s,\phi)} a_i)]= VU_{(s,\phi)} \left(\sum_i c_i [a_i]\right),
\end{align*}
for every $\sum_i c_i [a_i]\in \mathcal{D}_{\omega_\mu}$. The theorem is thus proved.
\end{proof}

It is a well-known result that the GNS Hilbert space of a quasifree state is a Fock space (see, e.g., \cite{KW91,BDFY15,WAL79}). In fact, the following results hold:

\begin{prop}[{\cite[Proposition 3.1]{KW91}}]\label{prop:4:one_particle_structure}
Let $S$ be a real vector space on which a symplectic form $\tilde{\sigma}$ and a bilinear positive symmetric form $\mu$ are defined such that:
$$
\dfrac{1}{2}|\tilde{\sigma}(v,w)|\leq |\mu(v,v)|^{1/2} |\mu(w,w)|^{1/2}
$$
for every $v,w\in S$. Then, there exists a pair $(K,H)$, called \emph{one-particle Hilbert space structure}, such that $H$ is a complex Hilbert space and $K:S	\to H$ is a map satisfying:
\begin{enumerate}[(i)]
\item $K$ is $\R$-linear and $K(S)+iK(S)$ is dense in $H$;
\item $(Kv\,\vline\, Kw)_H=\mu(v,w)+\frac{i}{2}\tilde{\sigma}(v,w)$ for every $v,w\in S$;
\item A pair $(K',H')$ fulfils (i) and (ii) if and only if there exists an isometric surjective operator $V:H\to H'$ for which $VK=K'$.
\end{enumerate}
\end{prop}

\begin{thm}\label{thm:4:GNS_quasifree_state}
Let $\sigma:\mathcal{A}(S)\to \mathbb{C}$ be a quasifree state. The GNS triple $(\mathcal{H}_\sigma,\mathcal{D}_\sigma,\pi_\sigma, \Psi_\sigma)$ is specified as follows:
\begin{enumerate}[(i)]
\item $\mathcal{H}_\sigma$ is the symmetrized Fock space on the one-particle Hilbert space $H$ given by Proposition \ref{prop:4:one_particle_structure};
\item $\Psi_\sigma$ is the vacuum of the Fock space;
\item $\mathcal{D}_\sigma$ is the dense subspace of finite linear combinations of $\Psi_\sigma$ and the vectors of the form:
$$
\Phi(v_1)\dots \Phi(v_n)\Psi_\sigma,
$$
for $n\in\N$ and $v_1,\dots,v_n\in S$, where, denoting by $a^\dagger(v)$ and $a(v)$ the creation and the annihilation operator corresponding to $v\in S$, $\Phi(v)$ is the essentially self-adjoint operator on $\mathfrak{F}(\mathcal{H}_\sigma)$ defined by:
$$
\Phi(v)=2^{-1/2}(a^\dagger(v)+a(v)), \qquad v\in S;
$$
\item The representation $\pi_\sigma$ is completely specified by:
$$
\pi_\sigma(\mathcal{W}(v))= \exp\left( i\Phi(v) \right), \qquad\quad \forall v\in S.
$$
In particular, $\pi_\sigma(\mathcal{W}(v))$ is a unitary operator for every $v\in S$.
\end{enumerate}
\end{thm}


We can then push the investigation a step forward and wonder whether the duality operator is second quantized. We find that the answer is positive.


Proposition \ref{prop:4:one_particle_structure} and Theorem \ref{thm:4:GNS_quasifree_state}, upon performing the identifications $S=d\Omega^{0,0}(\S^1)$, $\tilde{\sigma}=\tau_u$ and $\sigma=\omega_\mu$, yield a new GNS triple, unitarily isomorphic to $(\mathcal{H}_{\omega_\mu},\pi_{\omega_\mu},\Psi_{\omega_\mu})$ on account of the GNS theorem, whose Hilbert space is a bosonic Fock space. Let $H$ be the one-particle Hilbert space and $K:d\Omega^{0,0}(\Sigma)\to H$ the map given by the one-particle structure. The operator implementing duality on $H$ is the linear extension of:
\begin{eqnarray}
&\tilde{V}: H\to H&\nonumber\\
&K(df,d\tilde{f})\mapsto K(-(-1)^{k^2}d\tilde{f}, df), &\nonumber\\
&iK(df,d\tilde{f})\mapsto iK(-(-1)^{k^2}d\tilde{f}, df). &
\end{eqnarray}

It is a straightforward check that $\tilde{V}$ is unitary, as, when computing the inner product, we realise that the grading compensates for the swap of the components, in a way very similar to what happens to the exponential in Remark \ref{rem:4:duality_topological}. Its second quantization $\Gamma(\tilde{V})$ provides the counterpart of $V$ on the bosonic Fock space, thus confirming our claim.

\chapter*{Conclusions} \addcontentsline{toc}{chapter}{Conclusions} \markboth{Conclusions}{Conclusions}

The original result of the present work is the construction of states for a quantum field theory on differential cohomology for spacetimes with compact Cauchy surface, implementing the quantum Abelian duality.\\

In Section \ref{section:2.3_quantum_field_theory} we illustrated how the relevant space of observables, namely $\C^k_{sc}(M;\Z)$, fits in the commutative diagram of short exact sequences \eqref{diag:3:constr_pre_sympl_str} and can be endowed with a pre-symplectic structure. Section \ref{section:3.1_presymplectic_strunctures} was devoted to show that, resorting to the properties provided by differential cohomology, the pre-symplectic structure on the central object can be propagated to the other elements of the diagram in a consistent way. Furthermore, the existence of suitable splittings compatible with such structures has been proved in Section \ref{section:3.2_Presymplectic_decomposition}, yielding the following symplectically orthogonal decomposition:
\begin{align}\label{concl:pres_decomp}
\left(\C^k(M;\Z),\sigma\right)&=\left(H_{\tor}^{k,m-k}(M;\Z),\tau_d\right)\nonumber\\
&\oplus \left(H^{k,m-k}_{\free}(M;\Z)\times \frac{H^{k-1,m-k-1}(M;\R)}{H^{k-1,m-k-1}(M;\Z)},\tau_{lr}\right) \nonumber\\
&\oplus \left(d\Omega^{k-1}\cap \ast d\Omega^{m-k-1}(M),\tau_u\right).
\end{align}

The $\mathfrak{CCR}$ functor \eqref{eq:2:CCR_functor} assigns to $\left(\C^k(M;\Z),\sigma\right)$ an object in the category $\mathsf{C^\ast Alg}$, which, on account of \eqref{concl:pres_decomp}, can be presented as the tensor product of the algebras associated to the direct summands. Observing that the torsion subgroup is vanishing in a wide class of examples (see Section \ref{section:3:differential_characters} and Section \ref{subsection:3.4.2_examples}), we considered the algebra for the dynamical sector $\mathcal{A}_u$ and the algebra for the topological sector $\mathcal{A}_{lr}$ and we exhibited a state of the form:
\begin{equation}
\omega=\omega_\mu \otimes \omega_{t}: \mathcal{A}_u\otimes \mathcal{A}_{lr}\to \mathbb{C},
\end{equation}
where $\omega_\mu$ is Hadamard in a weak sense and $\omega_t$ is a state for the topological sector (see equations \eqref{eq:3:state_2D},\eqref{def:3:omega_mu},\eqref{eq:3:state_topol_0} and \eqref{eq:3:state_topol}). Notice that the construction of $\omega_t$ performed in Section \ref{subsection:3.3.4_state_topological_sector} is quite general and prescinds from the dimensionality of the spacetime.\\

Special attention has been devoted to the two- and the four-dimensional cases in Section \ref{section:3.3_2dim_case} and Section \ref{section:3.4_4D_case} respectively. For $M=\R\times \S^1$, we derived an explicit form for differential characters in Section \ref{section:3:differential_characters} and we computed the pre-symplectic product in Section \ref{subsection:3.3.2_comput_presympl_prod}. Furthermore, we proved that a Hadamard state exists, which is ground (see Section \ref{section:3:Hadamard_2D}). In Section \ref{subsection:4.3.1_Hadamard_state_general} we provided, in arbitrary spacetime dimension, an explicit formula for $\omega_\mu$ in terms of the coefficients of the Fourier decomposition of the initial data on a compact Cauchy surface.\\

At last, in Chapter \ref{chapter:4} quantum Abelian duality has been investigated. After showing the existence of a duality for our theory, in Section \ref{section:4.1_duality _and_splittings} we proved that it is possible to choose the splittings of the short exact sequences in such a way that they are compatible with the duality. In Section \ref{section:4.2_duality_and_GNS} we proved that, at the level of the GNS triple of $\omega$, the duality is implemented by the unitary operator:
\begin{equation}
V \otimes U: \mathcal{H}_{\omega_\mu}\otimes \mathcal{H}_{\omega_t}\to \mathcal{H}_{\omega_\mu}\otimes \mathcal{H}_{\omega_t}.
\end{equation}
In particular, moving to the Fock representation, $V$ turns out to be second quantized.
Furthermore, in  Proposition \ref{prop:4:state_invariant_symm} and Proposition \ref{prop:4:state_invariant_duality} we pointed out how the state $\omega_\mu$ is invariant both under the spacetime symmetries and under duality. Theorem \ref{thm:4:symm_duality_unitary_commuting} then guarantees that at the level of the representation, the spacetime symmetries are implemented by unitary operators commuting with the duality.\\

The question of how to build a state for an arbitrary globally hyperbolic spacetime still remains open. A possible continuation of the present work could be the pursuit of an alternative strategy to achieve a decomposition of the space of the observables for the general case. Nonetheless, as discussed in Section \ref{section:3.5_non_compact_case}, some clues suggest that if we put no further restrictions on the topology of the underlying manifold the analysis is destined to be a case by case inspection.
\chapter*{Acknowledgements}\addcontentsline{toc}{chapter}{Acknowledgements} \markboth{Acknowledgements}{Acknowledgements}

\thispagestyle{empty}

It is a pleasure for me to address my most grateful thank to Claudio and Marco, my advisor and my co-advisor, for their invaluable help and constant support: without them, this work would have never seen its end. Almost one year ago, bad luck overtook them and they happened to have a student who, besides all the shortcomings of an average Master's student, is teetotaller and unresponsive to any football ardour or enthusiasm, at least one of them without even knowing it in advance. In spite of this, they managed to make the thesis research literally fun and stimulating all the way.\\

I also wish to thankfully acknowledge the kind invitation of professor Christian B\"ar to visit the Institut f\"ur Mathematik of the University of Potsdam last April, where part of this work has been done.\\

\emph{Pavia, 21 July 2016}
\vspace{0.8cm}
\begin{flushright}
Matteo Capoferri \hspace{2cm}{\color{white}.}
\end{flushright}

\clearpage{\pagestyle{empty}\cleardoublepage}

\nocite{*}
\addcontentsline{toc}{chapter}{Bibliography}

\bibliography{bibliography.bib}

\begin{bibdiv}
\begin{biblist}

\bib{ARA71b}{article}{
      author={Araki, H.},
       title={On quasifree states of the canonical commutation relations {II}},
        date={1971},
     journal={Pub. R.I.M.S},
      volume={7},
       pages={121\ndash 152},
}

\bib{ARA71a}{article}{
      author={Araki, H.},
      author={Shiraishi, M.},
       title={On quasifree states of the canonical commutation relations {I}},
        date={1971},
     journal={Pub. R.I.M.S},
      volume={7},
       pages={105\ndash 120},
}

\bib{B15}{article}{
      author={B{\"a}r, Christian},
       title={{Green-Hyperbolic} {Operators} on {Globally} {Hyperbolic}
  {Spacetimes}},
        date={2015},
     journal={Comm. Math. Phys.},
      volume={333},
      number={3},
       pages={1585\ndash 1615},
}

\bib{BB14}{book}{
      author={B\"ar, C.},
      author={Becker, C.},
       title={{Differential} {Characters}},
      series={Lect. Notes Math.},
   publisher={Springer},
        date={2014},
      number={2112},
}

\bib{BMath}{article}{
      author={Becker, C.},
      author={Benini, M.},
      author={Schenkel, A.},
      author={Szabo, R.~J.},
       title={{Cheeger-Simons} differential characters with compact support and
  {Pontryagin} duality},
        date={2015},
      eprint={1511.00324},
}

\bib{BPhys}{article}{
      author={Becker, C.},
      author={Benini, M.},
      author={Schenkel, A.},
      author={Szabo, R.~J.},
       title={Abelian {Duality} on {Globally} {Hyperbolic} {Spacetimes}},
        date={2016},
     journal={Comm. Math. Phys.},
        note={doi: 10.1007/s00220-016-2669-9},
}

\bib{BDFY15}{book}{
      editor={Brunetti, R.},
      editor={Dappiaggi, C.},
      editor={Fredenhagen, K.},
      editor={Yngvason, J.},
       title={{Advances} in {Algebraic} {Quantum} {Field} {Theory}},
   publisher={Springer International Publishing},
        date={2015},
}

\bib{BDHS14}{article}{
      author={Benini, M.},
      author={Dappiaggi, C.},
      author={Hack, T.-P.},
      author={Schenkel, A.},
       title={A {$C^\ast$-Algebra} for {Quantized} {Principal}
  {U(1)-Connections} on {Globally} {Hyperbolic} {Lorentzian} {Manifolds}},
        date={2014},
     journal={Comm. Math. Phys.},
      volume={332},
      number={1},
       pages={477\ndash 504},
}

\bib{BFV03}{article}{
      author={Brunetti, R.},
      author={Fredenhagen, K.},
      author={Verch, R.},
       title={{The} {Generally} {Covariant} {Locality} {Principle} -- {A} {New}
  {Paradigm} for {Local} {Quantum} {Field} {Theory}},
        date={2003},
     journal={Comm. Math. Phys.},
      volume={237},
      number={1},
       pages={31\ndash 68},
}

\bib{BGP}{book}{
      author={B{\"a}r, C.},
      author={Ginoux, N.},
      author={Pf{\"a}ffle, F.},
       title={{Wave} {Equations} on {Lorentzian} {Manifolds} and
  {Quantization}},
   publisher={American Mathematical Society},
        date={2007},
}

\bib{BR03}{book}{
      author={Bratteli, O.},
      author={Robinson, D.~W.},
       title={{Operator} {Algebras} and {Quantum} {Statistical} {Mechanics} 2},
      series={Texts and Monographs in Physics},
   publisher={Springer-Verlag},
        date={2003},
}

\bib{BS06}{article}{
      author={Bernal, A.~N.},
      author={S{\'a}nchez, M.},
       title={{Further} {Results} on the {Smoothability} of {Cauchy}
  {Hypersurfaces} and {Cauchy} {Time} {Functions}},
        date={2006},
     journal={Lett. Math. Phys.},
      volume={77},
      number={2},
       pages={183\ndash 197},
}

\bib{BSS14}{article}{
      author={Becker, C.},
      author={Schenkel, A.},
      author={Szabo, R.~J.},
       title={Differential cohomology and locally covariant quantum field
  theory},
        date={2014},
      eprint={1406.1514},
}

\bib{CS85}{article}{
      author={Cheeger, J.},
      author={Simons, J.},
       title={Differental characters and geometric invariants},
        date={1985},
     journal={Lect. Notes Math.},
      volume={1167},
      number={50},
}

\bib{DB60}{article}{
      author={DeWitt, B.~S.},
      author={Brehme, R.W.},
       title={Radiation damping in a gravitational field},
        date={1960},
     journal={Ann. Phys.},
      number={9},
       pages={220\ndash 259},
}

\bib{FIO2}{article}{
      author={Duistermaat, J.~J.},
      author={H{\"o}rmander, L.},
       title={Fourier integral operators. {II}},
        date={1972},
     journal={Acta Mathematica},
      volume={128},
      number={1},
       pages={183\ndash 269},
}

\bib{Dim80}{article}{
      author={Dimock, J.},
       title={{Algebras} of {Local} {Observables} on a {Manifold}},
        date={1980},
     journal={Comm. Math. Phys.},
      volume={77},
      number={219},
}

\bib{DL12}{article}{
      author={Dappiaggi, C.},
      author={Lang, B.},
       title={{Quantization} of {Maxwell's} {Equations} on {Curved}
  {Backgrounds} and {General} {Local} {Covariance}},
        date={2012},
     journal={Lett. Math. Phys.},
      volume={101},
      number={3},
       pages={265\ndash 287},
}

\bib{FL16}{article}{
      author={Fewster, C.~J.},
      author={Lang, B.},
       title={{Dynamical} {Locality} of the {Free} {Maxwell} {Field}},
        date={2016},
     journal={Ann. Henri Poinc.},
      volume={17},
      number={2},
       pages={401\ndash 436},
}

\bib{FMSb}{article}{
      author={Freed, D.~S.},
      author={Moore, G.~W.},
      author={Segal, G.},
       title={{Heisenberg} groups and noncommutative fluxes},
        date={2007},
     journal={Ann. Phys.},
      volume={322},
      number={1},
       pages={236 \ndash  285},
}

\bib{FMSa}{article}{
      author={Freed, D.~S.},
      author={Moore, G.~W.},
      author={Segal, G.},
       title={The {Uncertainty} of {Fluxes}},
        date={2007},
     journal={Comm. Math. Phys.},
      volume={271},
      number={1},
       pages={247\ndash 274},
}

\bib{FIO1}{article}{
      author={H\"ormander, L.},
       title={Fourier integral operators. {I}},
        date={1971},
     journal={Acta Mathematica},
      volume={127},
      number={1},
       pages={79\ndash 183},
}

\bib{Haag}{book}{
      author={Haag, R.},
       title={{Local} {Quantum} {Physics}: {Fields}, {Particles}, {Algebras}},
      series={Text and Monographs in Physics},
   publisher={Springer},
        date={1996},
}

\bib{HAT02}{book}{
      author={Hatcher, A.},
       title={Algebraic {Topology}},
   publisher={Cambridge University Press},
        date={2001},
}

\bib{HK64}{article}{
      author={Haag, R.},
      author={Kastler, D.},
       title={An {Algebraic} {Approach} to {Quantum} {Field} {Theory}},
        date={1964},
     journal={J. Math. Phys.},
      volume={5},
      number={7},
       pages={848\ndash 861},
}

\bib{Hor}{book}{
      author={H{\"o}rmander, L.},
       title={{The} {Analysis} of {Linear} {Partial} {Differential} {Operators}
  {I}: {Distribution} {Theory} and {Fourier} {Analysis}},
      series={Classics in Mathematics},
   publisher={Springer-Verlag},
        date={2003},
}

\bib{KW91}{article}{
      author={Kay, B.~S.},
      author={Wald, R.~M.},
       title={Theorems on the uniqueness and thermal properties of stationary,
  nonsingular, quasifree states on spacetimes with a bifurcate {Killing}
  horizon},
        date={1991},
     journal={Physics Report},
      volume={207},
      number={2},
       pages={49\ndash 136},
}

\bib{LEE03}{book}{
      author={Lee, J.~M.},
       title={{Introduction} to {Smooth} {Manifolds}},
      series={Graduate Texts in Mathematics},
   publisher={Springer-Verlag},
        date={2003},
      number={218},
}

\bib{ML95}{book}{
      author={MacLane, S.},
       title={Homology},
   publisher={Springer-Verlag},
        date={1995},
}

\bib{MAS91}{book}{
      author={Massey, W.~S.},
       title={A basic course in algebraic topology},
      series={Graduate Texts in Mathematics},
   publisher={Springer},
        date={1992},
      number={127},
}

\bib{Mor13}{book}{
      author={Moretti, V.},
       title={{Spectral} {Theory} and {Quantum} {Mechanics} - {With} an
  {Introduction} to the {Algebraic} {Formulation}},
   publisher={Springer-Verlag},
        date={2013},
}

\bib{MAN73}{article}{
      author={Manuceau, J.},
      author={Sirugue, M.},
      author={Testard, D.},
      author={Verbeure, A.},
       title={The smallest {$C^*$-algebra} for canonical commutations
  relations},
        date={1973},
     journal={Comm. Math. Phys.},
      volume={32},
      number={3},
       pages={231\ndash 243},
}

\bib{Rad96}{article}{
      author={Radzikowski, M.~J.},
       title={Micro-local approach to the {Hadamard} condition in quantum field
  theory on curved space-time},
        date={1996},
     journal={Comm. Math. Phys.},
      volume={179},
      number={3},
       pages={529\ndash 553},
}

\bib{Ros97}{book}{
      author={Rosenberg, S.},
       title={The {Laplacian} on a {Riemannian} {Manifold}: {An} {Introduction}
  to {Analysis} on {Manifolds}},
      series={London Mathematical Society Student Texts},
   publisher={Cambridge University Press},
        date={1997},
      number={31},
}

\bib{SHU13}{thesis}{
      author={Schubert, S.},
       title={{{\"U}ber} die {Charakterisierung} von {Zust{\"a}nden}
  hinsichtlich der {Erwartungswerte} quadratischer {Operatoren}},
        type={Master's Thesis},
        date={2013},
}

\bib{SS08}{article}{
      author={Simons, J.},
      author={Sullivan, D.},
       title={Axiomatic characterization of ordinary differential cohomology},
        date={2008},
     journal={J. Topol.},
      volume={1},
      number={45},
}

\bib{Str05}{book}{
      author={Strocchi, F.},
       title={An {Introduction} to the {Mathematical} {Structure} of {Quantum}
  {Mechanics}},
   publisher={World Scientific},
        date={2005},
}

\bib{SV00}{article}{
      author={Sahlmann, H.},
      author={Verch, R.},
       title={{Passivity} and {Microlocal} {Spectrum} {Condition}},
        date={2000},
     journal={Comm. Math. Phys.},
      volume={214},
      number={3},
       pages={705\ndash 731},
}

\bib{WAL79}{article}{
      author={Wald, R.~M.},
       title={{Existence} of the {S} {Matrix} in {Quantum} {Field} {Theory} in
  {Curved} {Space-Time}},
        date={1979},
     journal={Ann. Phys.},
      volume={118},
       pages={490\ndash 510},
}

\bib{WAL94}{book}{
      author={Wald, R.~M.},
       title={{Quantum} {Field} {Theory} on {Curved} {Spacetime} and {Black}
  {Hole} {Thermodynamics}},
      series={Chicago Lectures in Physics},
   publisher={University of Chicago Press},
        date={1994},
}

\bib{WER83}{book}{
      author={Werner, F.~W.},
       title={Foundations of {Differentiable} {Manifolds} and {Lie} {Groups}},
      series={Graduate Texts in Mathematics},
   publisher={Springer-Verlag},
        date={1983},
      number={94},
}

\end{biblist}
\end{bibdiv}

\end{document}